\documentclass[11pt,a4paper,reqno,french,tikz]{amsart}
\usepackage[utf8]{inputenc} \usepackage[T1]{fontenc}
\usepackage{amsthm,amsmath,amsfonts,amssymb,amsxtra,appendix,bookmark,dsfont,bm,mathrsfs,amstext,amsopn,mathrsfs,mathtools,comment,cite,hyperref,color,xcolor,cite}

\newtheorem{theorem}{Theorem}[section]\newtheorem{lemma}[theorem]{Lemma}\newtheorem{proposition}[theorem]{Proposition}\newtheorem{corollary}[theorem]{Corollary}\newtheorem{remark}[theorem]{Remark}

\let\C\relax\newcommand{\C}{\mathbb{C}}\newcommand{\R}{\mathbb{R}}\newcommand{\N}{\mathbb{N}}

\newcommand\cA{\mathcal{A}}\newcommand\cB{\mathcal{B}}\newcommand\cC{\mathcal{C}}\newcommand\cD{\mathcal{D}}\newcommand\cE{\mathcal{E}}\newcommand\cF{\mathcal{F}}\newcommand\cH{\mathcal{H}}\newcommand\cL{\mathcal{L}}\newcommand\cO{\mathcal{O}}\newcommand\cP{\mathcal{P}}\newcommand\cU{\mathcal{U}}

\DeclareMathOperator{\tr}{Tr}\DeclareMathOperator{\Span}{Span}\DeclareMathOperator{\Ker}{Ker}\DeclareMathOperator{\im}{Im}\DeclareMathOperator{\dist}{dist}
\def\d{{\rm d}}

\DeclareMathOperator{\re}{Re}

\renewcommand{\ge}{\geqslant}\renewcommand{\le}{\leqslant}

\newcommand{\pa}[1]{\left( #1 \right)} 
\newcommand{\bpa}[1]{\big( #1 \big)} 
\newcommand{\seg}[1]{\left[ #1 \right]} 
\newcommand{\ab}[1]{\left|#1\right|} 
\newcommand{\ps}[1]{\left< #1 \right>} 
\newcommand{\norb}[2]{ \| #1 \|_{#2} } 
\newcommand{\nor}[2]{ \|  #1 \|_{#2} } 
\newcommand{\nore}[1]{ \|  #1  \|_{ee} } 
\newcommand{\norm}[1]{  \| #1 \|} 
\newcommand{\proj}[1]{\left| #1 \right> \left< #1 \right|} 
\newcommand{\ketbra}[2]{\left| #1 \right> \left< #2 \right|} 

\newcommand\vp{\varphi} 
\newcommand{\ep}{\varepsilon} 
\let\p\relax\newcommand{\p}{\psi} 

\newcommand{\f}[2]{\frac{#1}{#2}} 
\newcommand{\mymax}[1]{\underset{\substack{#1}}{\text{\normalfont{max}}}\;} 
\newcommand{\mylim}[1]{\underset{\substack{#1}}{\text{\normalfont{lim}}}\;} 
\newcommand{\mymin}[1]{\underset{\substack{#1}}{\text{\normalfont{min}}}\;} 
\newcommand{\mysup}[1]{\underset{\substack{#1}}{\text{\normalfont{sup}}}\;} 
\newcommand{\myinf}[1]{\underset{\substack{#1}}{\text{\normalfont{inf}}}\;} 
\def\1{{\mathds{1}}}

\newcommand{\Nz}{\N \cup \{0\}}
\newcommand{\Np}{\N}

\newcommand{\dm}[1]{\cD_{#1}}

\newcommand{\hs}[1]{\left( #1 \right)_{\text{2}}} 
\newcommand{\rr}{G}
\newcommand{\hh}{\mathfrak{H}}
\newcommand{\hls}{\mathfrak{S}_2}
\newcommand{\ggp}{\ab{H + a}^{\f 12}}
\newcommand{\ggm}{\ab{H + a}^{-\f 12}}
\newcommand{\bphi}{{\bm{\phi}}}
\newcommand{\bpsi}{{\bm{\psi}}}
\newcommand{\bvp}{{\bm{\varphi}}}
\newcommand{\restr}[2]{#1_{\mkern 1mu \vrule height 2ex\mkern2mu #2}}
\newcommand{\PHP}{\restr{\pa{\cP H \cP}}{\cP \cH \rightarrow \cP \cH}}
\newcommand{\PHlP}{\restr{\pa{\cP H(\lambda) \cP}}{\cP \cH \rightarrow \cP \cH}}
\newcommand{\PHlOP}{\restr{\pa{\cP H^0 \cP}}{\cP \cH \rightarrow \cP \cH}}

\usepackage{tikz, pgfplots}
\usetikzlibrary{pgfplots.fillbetween}
\usepackage{scrextend}

\author{Louis Garrigue}
\address[Louis Garrigue]{Laboratoire ``analyse géométrie modélisation'', CY Cergy Paris Université, 95302 Cergy-Pontoise, France} 
\email{louis.garrigue@cyu.fr}

\author{Benjamin Stamm}
\address[Benjamin Stamm]{Institute of Applied Analysis and Numerical Simulation, University of Stuttgart, 70569 Stuttgart, Germany
}
\email{benjamin.stamm@mathematik.uni-stuttgart.de}

\title[Reduced basis method and perturbation theory]{On reduced basis methods for eigenvalue\\problems, and on its coupling\\with perturbation theory}
\date{\today}
\begin{document}
\maketitle

\begin{abstract}
In this article, we study eigenvalue problems associated to self-adjoint operators and their approximation obtained by subspace projection, as used in the reduced basis method for instance. We provide error bounds between the exact eigenmodes and the approximated ones and also consider degenerate cases in the analysis. When the operator depends on a parameter, we apply the bounds assuming that the reduced space contains the derivatives of the eigenfunction with respect to the parameter. Finally, we provide some numerical examples that reflect the analytical results.
\end{abstract}

\section{Introduction}

\subsection{Projected space methods for eigenvalue problems}%
\label{sub:Projected space methods for eigenvalue problems}

A classical issue in eigenvalue problems is to reduce the number of degrees of freedom of the studied systems by extracting only the relevant ones, leading to subspace methods, the full considered Hilbert space $\cH$ being too large to be addressed in its exact form. 
In particular, reduced basis method (RBM) approximations aim at approximating $\cH$ by a well-chosen low-dimensional subset $\cP \cH$, created via an orthogonal projector $\cP$. 
Denoting the exact self-adjoint operator by $H$, then the approximated operator of interest is
\begin{align*}
\PHP, 
\end{align*}
as the restriction of $\cP H \cP$ to $\cP \cH \rightarrow \cP \cH$, whose eigenmodes we are interested in. Among other works, reduced basis problems have been investigated in the context of eigenvalue problems in~\cite{MadPatPer99,MacMadOli00,FumManPar16,CanDusMad17,CanDusMad18,HerStaWes22,BreHerWes23,TauDusEhr24,ManStaZen25}, the case of several eigenvalues was examined in~\cite{HorWohDic17,CanDusMad20,BofHalPri24} where more details about the RBM are considered, while here, we only focus on the approximation results. Moreover, in the previously mentioned works, \textit{a posteriori} bounds were derived, while we provide \textit{a priori} ones in this document. See also~\cite{Aronszajn51,Weinberger74,Greenlee83,BabOsb91,Gould12,KnyMehXu14} for other mathematical documents about the approximation of eigenvalue problems. 

In Theorem~\ref{thm:main_deg_thm} and Proposition~\ref{prop:non-deg_main_thm}, we first provide bounds enabling to estimate the error between the eigenmodes of the exact operator $H$ and the ones of the approximated operator $\PHP$ in a general setting. We treat the degenerate (i.e. when several eigenvalues are equal) and almost-degenerate cases by using the formalism of density matrices and clusters of eigenmodes, and we treat the non-degenerate case with an eigenvector formalism. 
We sought to provide general bounds which could be applied to diverse settings, for instance to derive effective operators in quantum mechanics, or to discretize the Hilbert space for numerical purposes. In the vector case, it is showed that those errors are controlled by the key quantity $\cP^\perp \phi$, where $\phi$ is the exact eigenvector of $H$. Since $\cP$ is chosen so that $\phi$ is expected to be ``almost'' contained in $\cP \cH$, $\cP^\perp \phi$ is small in norm. 
A similar quantity also controls the error in the case of clusters of eigenmodes.

\subsection{Coupling with perturbation theory}%
\label{sub:Coupling with perturbation theory}

We then apply our bounds to a reduced basis method in the context of parametrized eigenvalue problems, which uses the derivatives of the eigenvectors with respect to the parameter, to build the reduced space. To the best of our knowledge, this method was used for the first time in an article by Noor and Lowder~\cite{NooLow74} for computational engineering science, see also~\cite{MurHaf88,NaiKeaLan98,ItoRad01,TouUmr08} and their references. It was more recently used in nuclear physics~\cite{FraHeIps18}. Since then, many works showed the very interesting performance of this method applied to quantum mechanics, for instance in~\cite{KonEksHeb20,DemDugEks20,FurGarMil20,DemFroTic21,SarLee21,DriQuiGiu21,SarLee22,MelDriFur22,DugEksFur24}, providing perspectives to improve several areas of quantum physics beyond perturbation theory. We think that it could also lead to the derivation of more precise theoretical models, ongoing projects indicate this. The method is interpreted as a resummation of the perturbation series, and it yields a systematic way of forming effective operators. Other and different ways of coupling reduced basis methods and perturbation theory were also developped, we present several examples in Appendix~\ref{sec:Other kinds}. We refer to the method at stake here as the coupling between reduced basis methods (RBM) and perturbation theory (PT), which we call RBM+PT. The situation is illustrated in Figure~\ref{fig:main_idea}, on which we represent the spectra of $H(\lambda)$ and of $\PHlP$, where $H(\lambda)$ is the exact self-adjoint operator, depending on one parameter $\lambda \in \R$. Denoting one eigenvector of $H(\lambda)$ by $\phi(\lambda)$, the assumption of RBM+PT is that 
\begin{align}\label{eq:rbm_pt_hypo} 
\pa{\f{\d^n}{\d \lambda^n} \phi(\lambda)}_{\mkern 1mu \vrule height 2ex\mkern2mu \lambda = 0} \in \cP \cH, \qquad \text{ for all } n \in \{0,\dots,\ell\}.
\end{align}
It has been practically observed that the corresponding eigenvector of the effective operator $\PHlP$ is very close to the exact one, much closer than the perturbative approximation. To explain this: a heuristics is given in Section~\ref{sub:Heuristics}; quantitative bounds on the error between the exact and the approximate eigenvectors are provided in Corollary~\ref{cor:main_ec_vec}, more precisely, ~\eqref{eq:eigenvectors_bound} gives the leading order; to treat degeneracies and almost degeneracies, bounds on density matrices are given in Corollary~\ref{cor:main_ec_dm} and Theorem~\ref{thm:deg_vecs}.
Finally, we present in Section~\ref{sec:illustration}, on a one-body Schrödinger operator, the many benefits of RBM+PT compared to PT, and of RBM+PT compared to building a reduced basis from excited states taken at $\lambda = 0$, illustrating the efficiency of the method.

\begin{figure}[h!]
\begin{center}
\begin{minipage}{0.5\textwidth}
\begin{tikzpicture}[scale=0.8]
  \draw[scale=1, domain=-3:3, smooth, variable=\x, blue, name path=C] plot ({\x}, {-0.03*\x + 0.01*\x*\x + 4});
  \draw[scale=1, domain=-3:3, smooth, variable=\x, blue, name path=D] plot ({\x}, {0.03*\x*\x + 0.01*\x*\x*\x + 3.5});
  \draw[scale=1, domain=-3:3, smooth, variable=\x, red] plot ({\x}, {-0.3*\x - 0.1*\x*\x + 0.05*\x*\x*\x + 2.5});
  \draw[scale=1, domain=-3:3, smooth, variable=\x, blue] plot ({\x}, {0.2*\x - 0.2*\x*\x - 0.01*\x*\x*\x + 2.2});
  \draw[scale=1, domain=-3:3, smooth, variable=\x, red] plot ({\x}, {0.2*\x + 0.05*\x*\x - 0.00*\x*\x*\x + 2.5});
  \draw[scale=1, domain=-3:3, smooth, variable=\x, blue, name path=A] plot ({\x}, {-0.1*\x - 0.1*(\x-1)*(\x-1) + 0.01*\x*\x*\x + 1.3});
  \draw[scale=1, domain=-3:3, smooth, variable=\x, blue, name path=B] plot ({\x}, {0.05*(\x-1) - 0.05*(\x-1)*(\x-1) + 0.01*(\x-1)*(\x-1)*(\x-1) + 1.6});
  \draw[scale=1, domain=-3:3, smooth, variable=\x, blue] plot ({\x}, {0.1*\x + 0.02*(\x-1)*(\x-1) + 0.01*\x*\x*\x + 0.7});
\tikzfillbetween[of=A and B]{blue, opacity=1};
\tikzfillbetween[of=C and D]{blue, opacity=1};
  \draw[->] (-3.4, 0) -- (3.4, 0) node[right] {$\lambda$};
  \draw[-] (0, 0) -- (0, -0.1); \node at (0,-0.4) {$0$};
  \draw[-] (3, 0) -- (3, -0.1); \node at (3,-0.4) {$\lambda_0$};
  \draw[-] (-3, 0) -- (-3, -0.1); \node at (-3,-0.4) {$-\lambda_0$};
  \draw[->] (0, 0) -- (0, 5) node[above] {$\sigma(H(\lambda))$};
\end{tikzpicture}
\end{minipage}\hfill
\begin{minipage}{0.5\textwidth}
\begin{tikzpicture}[scale=0.8]
  \draw[scale=1, domain=-3:3, smooth, variable=\x, blue] plot ({\x}, {0.1*\x - 0.07*\x*\x - 0.01*\x*\x*\x + 4.5});
  \draw[scale=1, domain=-3:3, smooth, variable=\x, blue] plot ({\x}, {0.2*\x + 0.07*\x*\x - 0.01*\x*\x*\x + 3.2});
  \draw[scale=1, domain=-3:3, smooth, variable=\x, red] plot ({\x}, {-0.3*\x - 0.1*\x*\x + 0.05*\x*\x*\x + 2.5});
  \draw[scale=1, domain=-3:3, smooth, variable=\x, red] plot ({\x}, {0.2*\x + 0.05*\x*\x - 0.00*\x*\x*\x + 2.5});
  \draw[scale=1, domain=-3:3, smooth, variable=\x, blue] plot ({\x}, {0.3*\x + 0.13*(\x-1)*(\x-1) + 0.01*\x*\x*\x -0.2});
  \draw[scale=1, domain=-3:3, smooth, variable=\x, blue] plot ({\x}, {0.2*(\x-1) - 0.04*(\x-1)*(\x-1) - 0.01*\x*(\x-1)*(\x-1) + 1});
  \draw[->] (-3.4, 0) -- (3.4, 0) node[right] {$\lambda$};
  \draw[-] (0, 0) -- (0, -0.1); \node at (0,-0.4) {$0$};
  \draw[-] (3, 0) -- (3, -0.1); \node at (3,-0.4) {$\lambda_0$};
  \draw[-] (-3, 0) -- (-3, -0.1); \node at (-3,-0.4) {$-\lambda_0$};
  \draw[->] (0, 0) -- (0, 5) node[above] {$\sigma\bpa{\PHlP}$};
\end{tikzpicture}
\end{minipage}
\end{center}\label{fig:main_idea}
\caption{RMB+PT approximates very well the targeted eigenmodes (whose eigenvalues are in red), but can fail to reproduce the other ones (eigenvalues in blue). One can put (in $\cP \cH$) Taylor series of several eigenvectors if one wants to model several eigenmodes, as on this figure where the two eigenmodes in red are taken into account, while blue ones are not. In quantum physics, most of the time one is interested in only a few eigenmodes, those which are at the interface between occupied and unoccupied spectrum.}
\end{figure}



\newpage

\subsection{Organization of the document}%
\label{sub:orga of the document}

\begin{itemize}
\item In Section~\ref{sec:definitions_general}, we present the norms that we use, the formalism of density matrices, the projection $\cP$ onto a subspace, the targeted eigenmodes and the partial inverses

\item In Section~\ref{sec:main result on the RBM}, we give the results on the reduced basis method in the general case. We treat clusters of eigenmodes in Section~\ref{sub:Density matrices} using density matrices, and one eigenmode in Section~\ref{sub:One eigenmode} using eigenvectors. Our main results quantify the error between the exact eigenmodes and the approximate ones.

\item In Section~\ref{sec:definitions_rbm_pt}, we present the context of perturbation theory, where the considered operator depends on a real parameter. We introduce an analytic family of operators $H(\lambda)$ and choose a cluster of eigenmodes. The reduced operator $\PHlP$ has the same targeted eigenmodes as $H(\lambda)$ when $\lambda = 0$, hence Rellich's theorem enable to analytically continue the corresponding branches, this is explained in Sections~\ref{sub:Starting point for the reduced operator} and~\ref{sub:Analytic branches}.

\item In Section~\ref{sec:Application to RBM+PT}, we present the main results concerning the coupling between RBM and PT. In Section~\ref{sub:Heuristics}, we start by providing the idea of the mechanism at stake. Under the assumption~\eqref{eq:rbm_pt_hypo} defining RBM+PT, we then provide bounds between exact and approximate eigenmodes, in the perturbative regime. We treat clusters of eigenmodes in Section~\ref{sub:Clusters of eigenmodes} and the case of one eigenmode in~\ref{sub:One eigenmode}. We treat degenerate eigenvectors in Section~\ref{sub:Vectors in the degenerate case}, which needs a refined analysis using clusters of eigenmodes and density matrices.

\item In Section~\ref{sec:illustration}, we compare RBM+PT with other approximations, and we give numerical results in the case of a one-dimensional one-body Schrödinger operator. In Section~\ref{sub:RBM+PT versus PT} we compare RBM+PT with PT alone, we note that both in the perturbative regime and in the non-perturbative regimes, RBM+PT is more efficient than PT. In Section~\ref{sub:RBM+PT vs RBM with excited states}, we compare RBM+PT with RBM+ES, where ES stands for ``excited states'', RBM+ES being the RBM approximation where the approximation space is built from the eigenvectors of $H(0)$, that is $\cP \cH = \Span \pa{\phi_\mu(0)}_{0 \le \mu \le \beta}$. Some simulations suggest that the errors of RBM+PT and RBM+ES have similar behaviors with respect to $\lambda$ in the non-perturbative regime, but RBM+PT needs much less vectors than RBM+ES to reach a given precision.

\item Sections~\ref{sec:Proof of Theorem mauin},~\ref{sec:Proof of Theorem non deg} provide the proofs the Section~\ref{sec:main result on the RBM}. Section~\ref{sec:Bounds on the Rayleigh-Schrödinger series in perturbation theory} is devoted to provide bounds on the Rayleigh-Schrödinger series $(\phi^k_\mu)_{k \in \Np}$. Section~\ref{sec:density matrix perturbation theory} presents perturbation theory for density matrices, which is the setting enabling to treat degenerate and almost degenerate eigenmodes. Sections~\ref{sec:Proof of EC} and~\ref{sec:Proof deg vec} conclude the proofs of Section~\ref{sec:Application to RBM+PT} on RBM+PT.

\item We then conclude the document with a series of appendices. We start by providing a table of the main notations used in this document, in Appendix~\ref{sec:Table of notations}. In Appendix~\ref{sec:Other kinds}, we present other methods coupling RBM (also called the variational approximation) and perturbation theory in a different way, and existing in the literature. In Appendix~\ref{sec:interm_normalization}, we give precisions about the intermediate normalization technique used in the physics literature. In Appendix~\ref{sec:appendix_error_bounds_DM} we present some relationships between sets of vectors and the density matrices that they create. In Appendix~\ref{sec:Alternative bound} we give a bound between the approximate density matrix of eigenvectors obtained by RBM and the exact density matrix of eigenvectors, providing an alternative to~\eqref{eq:explicit_diff}.
\end{itemize}

\section{Definitions and assumptions\\for the general reduced basis method}%
\label{sec:definitions_general}

We choose a standard but general mathematical setting which can address common operators involved in quantum mechanics, including Dirac operators, many-body Schrödinger operators and Bloch transforms of periodic operators.

\subsection{Hilbert space and operators}%
\label{sub:First definitions}

Let $\cH$ be a separable Hilbert space, endowed with the complex scalar product $\ps{\cdot,\cdot}$ and the corresponding norm $\norm{\cdot}$. 
As a convention, we will denote elements $\vp \in \cH$ as vectors even if $\cH$ is infinite-dimensional and can be viewed as functions. 

The object of study in the article is a closed self-adjoint operator $H$ acting on $\cH$, and we want to approximate some of its eigenmodes by using a subspace or reduced basis method.

For linear operators $B:\cH\to \cH$, we will denote by
\begin{align*}
	\norm{B} := \mysup{\p \in \cH \backslash \{0\}} \f{\norm{B \p}}{\norm{\p}}
\end{align*}
the canonical operator norm.
Further, let us take a closed self-adjoint operator $A$ of $\cH$, possibly unbounded, which will serve to define the energy norm. 
We assume that $A$ has a dense domain and a dense form domain and we will always assume that 
\begin{align}\label{eq:cA_cH} 
c_A := \nor{A^{-1}}{}< +\infty, \qquad c_H := \nor{A^{-1} H A^{-1}}{}< +\infty.
\end{align}
For vectors $\vp \in \cH$, the energy norm is then defined by
\begin{align}\label{eq:energy_norm_vects} 
\nor{\vp}{e} := \nor{A \vp}{}.
\end{align}
This norm is stronger than the norm $\nor{\cdot}{} $. We will also make use of the notation $\nor{\cdot}{e,0} := \nor{\cdot}{} $ and $\nor{\cdot}{e,1} := \nor{\cdot}{e}$, so that for any $\vp \in \cH$ and $\delta \in \{0,1\}$, there holds $\nor{\vp}{e,\delta} = \nor{A^\delta \vp}{}$. 

For example, consider the case of the Schrödinger operator $H = -\Delta + v$ acting on $\cH = L^2(\R^3)$. It is then natural to choose $A = \sqrt{-\Delta}$ and $\nor{\cdot}{e}$ is equivalent to the Sobolev norm $H^1(\R^3)$.

Further, for any operators $B, D$ on $\cH$, the Hilbert-Schmidt scalar product is denoted by $\hs{B,D} := \tr B^* D$, its norm $\nor{B}{2} := \sqrt{\tr B^* B}$, and the corresponding normed space is the space of Hilbert-Schmidt operators, 
denoted by 
\begin{align}\label{eq:def_hls} 
	\hls := \{ B : \cH \rightarrow \cH, \nor{B}{2} < +\infty\}. 
\end{align}
 For $\delta \in \{0,1\}$ and any $B \in \hls$, we use the notation
\begin{align}\label{eq:energy_norm} 
\nor{B}{2,\delta} := \nor{A^\delta B }{2}.
\end{align}
The norm $\nor{\cdot}{2,1}$, called the energy norm, is the natural one for density matrices (i.e. operators), as $\nor{\cdot}{e}$ is the natural norm for vectors. 

\subsection{Density matrices}%
\label{sub:Density matrices}

For any $\vp \in \cH$, we denote by $P_\vp$ the orthogonal projector onto $\C \vp= \Span \{\vp\}$. For any orthogonal projection $P$, we will use the notation $P^\perp :=  1 - P$. 

The notion of the density matrices is a key-concept to formalize eigenvalue problems with degeneracies and cluster of eigenvalues. 
They are analogous objects as eigenvectors, but for problems considering several eigenvectors and consist of the spectral projector onto the eigenspace spanned by those eigenvectors. For any set $\bvp := (\vp_\mu)_{\mu=1}^\nu \in \cH^\nu$ of orthonormal vectors, we define the corresponding density matrix 
\begin{align}\label{eq:def_dm} 
\dm{\bvp} := \sum_{\alpha=1}^{\nu} \proj{\vp_\alpha}= \sum_{\alpha=1}^{\nu} P_{\vp_\alpha},
\end{align}
using the convenient bra-ket notation. By orthonormality, $\dm{\bvp}$ is an orthogonal projection on $\cH$, that is $\dm{\bvp}^2 = \dm{\bvp}^* = \dm{\bvp}$. We denote by 
\begin{align*}
\cU_\nu := \{U \in \C^{\nu \times \nu} \;|\; U^* U = 1\}
\end{align*}
the group of unitary matrices of dimension $\nu$ and for any $U \in \cU_\nu$ we define its action $U \bvp := ((U \bvp)_\alpha)_{\alpha=1}^\nu$ on $\cH^\nu$ where $(U \bvp)_\alpha := \sum_{\beta=1}^{\nu} U_{\alpha \beta} \vp_\beta$. Then, there holds $\dm{U \bvp} = \dm{\bvp}$ for all $U \in \cU_\nu$ and set $\bvp \in \cH^\nu$ of orthonormal vectors.
Therefore, $\dm{\bvp}$ is independent on $\bvp$ but only depends on the space that $\bvp$ spans.

\subsection{Space reduction by projection}%
\label{sub:Consider a reduced space}

Let us take an orthogonal projection $\cP$ on $\cH$, we assume that $\cP$ is neither the identity nor the null projection to avoid the trivial cases, and we recall that $\cP^\perp = 1 - \cP$. The reduced space is $\cP \cH$, it can, but does not need to, be infinite-dimensional and we will need to assume that
\begin{align}\label{eq:ccP} 
c_\cP := \nor{A \cP A^{-1}}{}< +\infty.
\end{align}
Our central object will be $\restr{\pa{\cP \cH \cP}}{\cP \cH \rightarrow \cP \cH} : \cP \cH \rightarrow \cP \cH$, which is the restriction of $\cP H \cP$ to $\cP \cH$, hence it is an operator of $\cP \cH$, while $\cP H \cP$ is an operator of $\cH$. If $d := \dim \cP \cH$ is finite, we can represent this operator by a $d \times d$ matrix. We consider $\restr{\pa{\cP \cH \cP}}{\cP \cH \rightarrow \cP \cH}$ as an approximation of $H$, in the sense that its eigenmodes will well approximate the ones of $H$ we are interested in. Remark that since $\cP \neq 1$, $\sigma\pa{\cP H \cP} = \sigma\bpa{\restr{\pa{\cP \cH \cP}}{\cP \cH \rightarrow \cP \cH}} \cup \{0\}$ because $\im \cP^\perp \subset \Ker \cP H \cP$.




\subsection{Target set of eigenmodes}%
\label{sub:Choose sets of eigenmodes}

For a given $\nu \in \Np$, we focus in this work on an arbitrary set of $\nu$ eigenvalues 
\begin{align*}
\{E_\mu\}_{1 \le \mu \le \nu} \subset \sigma(H)
\end{align*}
in the spectrum of $H$. They are counted with multiplicity and their eigenvectors are denoted by $\phi_\mu$, normed such that $\nor{\phi_\mu}{} = 1$, and grouped into $\bphi := \pa{\phi_\mu}_{\mu=1}^\nu$. We denote the associated density matrix by
\begin{align*}
\Gamma := \sum_{\mu=1}^{\nu} \proj{\phi_\mu} = \dm{\bphi}.
\end{align*}
The purpose of taking clusters of eigenvalues, i.e. $\nu \ge 2$, is to be able to treat the almost-degenerate and degenerate cases, i.e. when eigenvalues are close or even equal. If the eigenvalues are not close, one can take the non-degenerate case $\nu = 1$ since no singular quantity in the following analysis will appear. Note that the eigenvalues $E_\mu$ are not necessarily sorted in increasing order and not necessarily on the bottom of the spectrum.

For any operator $B$, we denote by $\sigma_{\textup{d}}(B)$ the discrete spectrum of $B$, and we assume that $\PHP$ has at least $\nu$ eigenvalues in its discrete spectrum.
We consider $\nu$ of them, we denote them by
\begin{align}\label{eq:cE_in_spectrum} 
	\{\cE_\mu\}_{1 \le \mu \le \nu} \subset \sigma_{\textup{d}}\bpa{\PHP},
\end{align}
and the corresponding normalized eigenvectors by $\psi_\mu$, grouped into $\bpsi := \pa{\psi_\mu}_{\mu=1}^\nu$. 
We define the associated density matrix by
\begin{align*}
\Lambda := \sum_{\mu=1}^{\nu} \proj{\psi_\mu} = \dm{\bpsi}.
\end{align*}
Note that $\Lambda \cP =\Lambda$ and thus $\Lambda\cP^\perp=0$. 

We will study the closeness between $\phi_\mu$ and $\psi_\mu$ for any $\mu \in \{1,\dots,\nu\}$, so each level $\mu$ of $H$ corresponds to a level $\mu$ of $\PHP$. But they are not sorted in increasing order, so for instance if we follow a variational approach, the label $\mu$ can denote the $3^{\text{rd}}$ level of $H$ and the $5^{\text{th}}$ level of $\PHP$, and $\phi_\mu - \psi_\mu$ can be small. For example Figure~\ref{fig:main_idea} illustrates this principle, where $H$ depends on a parameter.

\subsection{Definition of partial inverses}%
\label{sub:Pseudo-inverse}


Finally, we need to introduce a technical definition combined with a further assumption. 
For any self-adjoint operator $B$, if $\{e_\mu\}_{\mu=1}^\alpha \subset \sigma_{\textup{d}}(B)$, then there exists $\kappa_B > 0$ such that
\begin{align}\label{eq:spectrum_gap_general} 
\pa{\sigma(B) \backslash \{e_\mu\}_{\mu=1}^\alpha} \cap \bpa{\cup_{\mu=1}^\alpha \; ]e_\mu - \kappa_B , e_\mu + \kappa_B[ \;} = \varnothing.
\end{align}
In addition to~\eqref{eq:cE_in_spectrum} we will also assume that
\begin{align}\label{eq:dim_nu} 
\dim \left( \cP \cH \cap \bigoplus_{\mu = 1}^\nu \Ker \bpa{\cP H \cP - \cE_{\mu}}\right) = \nu,
\end{align}
to ensure that all the eigenvectors associated to the approximated eigenvalues $\{\cE_\mu\}_{\mu=1}^\nu$ are taken into account. For any 
\begin{align*}
z \in \{\cE_\mu\}_{\mu=1}^\nu \cup \pa{\C \backslash \sigma\bpa{\PHP}}
\end{align*}
 we define
\begin{align}\label{def:pseudoinv_z_PHP}
\pa{z - \cP H \cP}_\perp^{-1} := 
\left\{
\begin{array}{ll}
\pa{\pa{z - \cP H \cP}_{\mkern 1mu \vrule height 2ex\mkern2mu \Lambda^\perp \cP \cH \rightarrow \Lambda^\perp \cP \cH}}^{-1} & \mbox{on } \Lambda^\perp \cP \cH, \\
0 & \mbox{on } \Lambda \cH \oplus \cP^\perp \cH,
\end{array}
\right.
\end{align}
extended by linearity on $\cH$. Note that we used the space decomposition $\cH = \Lambda \cH \oplus \Lambda^\perp \cP \cH\oplus \cP^\perp\cH$.
We also define
\begin{align}\label{eq:def_R} 
R_\mu := \pa{\cE_\mu - \cP H \cP}_\perp^{-1}.
\end{align}
Then, by~\eqref{eq:cE_in_spectrum} and~\eqref{eq:spectrum_gap_general} there exists $\kappa_{\cP H \cP} > 0$ such that $\nor{R_\mu}{} \le \kappa_{\cP H \cP}^{-1}$ for any $\mu \in \{1,\dots,\nu\}$.

\section{Main results on the general reduced basis method}
\label{sec:main result on the RBM}

In this section we present our main result, which is a bound on the difference between exact and approximated eigenmodes. This quantifies how close the approximation is from the exact quantity. It is a generic estimate that does not yet consider the parametrized setting, which is left for Section~\ref{sec:Application to RBM+PT}. We use the notations as introduced in Section~\ref{sec:definitions_general}. The case of clusters of eigenvalues is treated in Theorem~\ref{thm:main_deg_thm} and the non-degenerate case is treated in Proposition~\ref{prop:non-deg_main_thm}. 

\subsection{Clusters of eigenmodes}%
\label{sub:Clusters of eigenmodes}

The following result is adapted to the degenerate or almost-degenerate case, that is when several eigenvalues that we take into account are equal or close. Our goal is to quantify the difference between the eigenvectors of $H$ and the ones of $\PHP$, we use the formalism of density matrices recalled in Section~\ref{sub:Density matrices}.


\begin{theorem}[Error between exact eigenmodes and reduced basis eigenmodes]\label{thm:main_deg_thm} 
	Take a Hilbert space $\cH$, and a closed self-adjoint operator $A$ which is chosen to form the energy norm defined in~\eqref{eq:energy_norm}. Take a closed self-adjoint operator $H$ which eigenmodes will be approximated. Consider an orthogonal projector $\cP$, assume that $H$ and $\PHP$ have at least $\nu$ eigenvalues (counted with multiplicity). We consider $\nu$ eigenmodes of respectively $H$ and $\PHP$, denoted by respectively $(E_\mu,\phi_\mu)$ and $(\cE_\mu,\psi_\mu)$, where $\mu \in \{1,\dots,\nu\}$, ${\nor{\phi_\mu}{} = \nor{\psi_\mu}{} = 1}$ and we assume~\eqref{eq:cE_in_spectrum} and~\eqref{eq:dim_nu} to hold. We define ${\bphi := (\phi_\mu)_{\mu=1}^\nu}$, ${\bpsi := (\psi_\mu)_{\mu=1}^\nu}$, and the corresponding density matrices ${\Gamma := \dm{\bphi}}$, $\Lambda := \dm{\bpsi}$. We assume that $c_\cP, c_A < +\infty$, where those quantities are defined in Section~\ref{sec:definitions_general}, and that all the quantities involved in the following are finite. For $\delta \in \{0,1\}$, we have
 \begin{align}\label{eq:explicit_diff} 
\Gamma - \Lambda = \sum_{\mu=1}^\nu  \pa{(1 + H R_\mu )\cP^\perp \Gamma P_{\psi_\mu}+ P_{\psi_\mu} \Gamma \cP^\perp (1 + R_\mu H)} + \Omega,
\end{align}
where
\begin{multline}\label{eq:bound_Omega} 
	\nor{\Omega}{2,\delta} \le c_A^\delta \nor{ \cP^\perp \Gamma }{2,\delta}^2 +\bpa{1 + \pa{c_A c_\cP^2}^\delta}  \pa{1 + c_A (1+c_A)\nor{A \Lambda}{} }^{2\delta}\nor{ \Gamma - \Lambda }{2,\delta}^2 \\
	+ 2 \pa{c_\cP^\delta \hspace{-0.1cm}+\hspace{-0.1cm} \nu c_A^{\delta}   \norb{\cP^\perp H \Lambda}{}\mymax{1 \le \mu \le \nu} \hspace{-0.1cm} \norb{A^\delta R_\mu}{}}  \pa{c_A \pa{1 + c_A \nor{A \Lambda}{}} }^\delta \hspace{-0.1cm}\nor{\cP^\perp \Gamma}{2,\delta}\hspace{-0.1cm} \nor{\Gamma - \Lambda}{2,\delta}.
\end{multline}
\end{theorem}
A proof is given in Section~\ref{sec:Proof of Theorem mauin}. Since $\cP^\perp \Lambda = 0$, we have 
\begin{align*}
\nor{\cP^\perp \Gamma}{2,\delta} = \nor{\cP^\perp A^{-\delta} A^{\delta }\pa{\Gamma - \Lambda}}{2,\delta} \le (1 + c_\cP)^\delta \nor{\Gamma - \Lambda}{2,\delta},
\end{align*}
hence we see in~\eqref{eq:bound_Omega} that $\Omega$ is quadratic in $\nor{\Gamma - \Lambda}{2,\delta}$, and thus negligible in~\eqref{eq:explicit_diff} when $\nor{\Gamma - \Lambda}{2,\delta}$ is small. The leading term for $\Gamma - \Lambda$ is thus $\sum_{\mu=1}^\nu  \pa{(1 + H R_\mu )\cP^\perp \Gamma P_{\psi_\mu}+ \textup{adj.}}$, as emphazised in~\eqref{eq:explicit_diff}, where ``\textup{adj.}'' denotes the adjoint operator of the previous one. Moreover, $\cP$ is in general chosen such that $\Gamma$ approximately belongs to $\cP \cH$, hence $\cP^\perp \Gamma$ is small.

\subsection{One eigenmode}%
\label{sub:One eigenmode}

In the case where we treat only one eigenmode, one can obtain more precision error bounds, which is the object of the following result. Since we only consider one eigenmode, we drop the subscripts 1 labeling the different eigenvectors and write $\phi := \phi_1$, $\psi := \psi_1$, $E := E_1$, $\cE := \cE_1$, and $R := R_1$. We recall that $R = \pa{\cE - \cP H \cP}_\perp^{-1}$, see~\eqref{eq:def_R}.

\begin{proposition}[More detail in the non-degenerate case]\label{prop:non-deg_main_thm} 
	Consider the same assumptions as in Theorem~\ref{thm:main_deg_thm} for $\nu = 1$ ignoring the subscripts $1$ such that $(E,\phi)$ is an eigenmode of $H$ and $(\cE,\psi)$ an eigenmode of $\PHP$, so that ${H \phi = E \phi}$ and ${\cP H \cP \psi = \cE \psi}$. In a gauge where $\ps{\psi,\phi} \in \R$, there holds
\begin{align}
	\phi - \psi &= \pa{1 + R H} \cP^\perp \phi - \tfrac 12 \nor{\phi - \psi}{}^2 \psi + \pa{\cE - E} R \pa{\phi - \psi}, \label{eq:equality_diff} \\
	E - \cE &= \ps{\cP^\perp \phi, \pa{\cE - H} \pa{1 + R H} \cP^\perp \phi} + (E-\cE) \nor{\phi - \psi}{}^2  \label{eq:equality_E_diff} \\
		& \hspace{-0.5cm}- \nor{\phi - \psi}{}^2 \re \ps{\cP^\perp \phi, (H-E) (\phi - \psi)}+\pa{E-\cE}^2 \ps{\phi - \psi, R \pa{\phi - \psi}}.\nonumber
\end{align}
\end{proposition}

As will be made precise in Remark~\ref{rq:cs_vec}, the leading term of the error $\phi - \p$ is $\pa{1 + R H} \cP^\perp \phi$ while the remaining part is at least quadratic in $\nor{\phi - \p}{} $. 
We emphasize that for a good choice of $\cP$, $\phi$ is close to $\cP \cH$ which implies that $\cP^\perp \phi$ is small. The proof of Proposition~\eqref{prop:non-deg_main_thm} is presented in Section~\ref{sec:Proof of Theorem non deg}.

\subsection{Remarks}%
\label{sub:Remarks}
Let us now proceed with some remarks.

\begin{remark}[Invariance under unitary transforms]
	All the quantities involved in~\eqref{eq:explicit_diff}, \eqref{eq:bound_Omega} are invariant under the unitary transformations $\bphi \rightarrow U \bphi$ and $\bpsi \rightarrow V \bpsi$, for any $U,V \in \cU_\nu$, except for the term $P_{\psi_\mu}$ in~\eqref{eq:explicit_diff}. Nevertheless, this term is invariant under the transformation $\pa{\psi_\mu}_{\mu=1}^\nu \rightarrow \bpa{\vp_\mu}_{\mu=1}^\nu = V \pa{\psi_\mu}_{\mu=1}^\nu$ for any $V \in \cU_\nu$ as long as $\vp_\mu \in \Ker \pa{H - E_\mu}$ for all $\mu \in \{1,\dots,\nu\}$.
\end{remark}

\begin{remark}[Error in individual eigenvectors and in sum of eigenvalues of the cluster]
	From Theorem~\ref{thm:main_deg_thm}, for $\nor{\Gamma - \Lambda}{2,\delta}$ small enough and $\delta \in \{0,1\}$,
\begin{align}\label{eq:bbb} 
	\nor{\Gamma - \Lambda}{2,\delta} \le 4 \nu \pa{c_A\nor{A \Lambda }{}}^\delta  \nor{\cP^\perp \Gamma}{2,\delta} \mymax{1 \le \mu \le \nu} \nor{A^\delta (1+ R_\mu H) \cP^\perp A^{-\delta}}{},
\end{align}
see Section~\ref{sub:Inequalities} to have more details on how to obtain this inequality. Moreover, by Lemma~\ref{lem:compare_errors}, there exists a rotation $U \in \cU_\nu$ such that
\begin{align}\label{eq:bound_rot} 
	\sum_{\mu=1}^{\nu} \nor{A\bpa{\bphi_\mu - \pa{U \bpsi}_\mu}}{} \le c \nor{A \cP^\perp \Gamma}{2}
\end{align}
for some constant $c$, and again by Lemma~\ref{lem:compare_errors}, and yet for another constant $c$, the error in the sums of eigenvalues is quadratic, that is
\begin{align}\label{eq:sum_eigenvals_converge} 
	\ab{\sum_{\mu=1}^{\nu} \pa{E_\mu - \cE_\mu}} \le c \nor{A \cP^\perp \Gamma}{2}^2.
\end{align}
Hence the errors in individual eigenvectors (left-hand side of~\eqref{eq:bound_rot}) and in the sum of eigenvalues (left-hand side of~\eqref{eq:sum_eigenvals_converge}) are controlled by the key quantity $\nor{A \cP^\perp \Gamma}{}$.
\end{remark}

\begin{remark}[Main consequences of Proposition~\ref{prop:non-deg_main_thm}]\label{rq:cs_vec} 
If $\nor{\phi - \psi}{e}$ is small enough, then Proposition~\ref{prop:non-deg_main_thm} yields
\begin{align}
	\nor{\phi - \psi}{e,\delta} &\le 2 \pa{1 + c_H \nor{ARA}{}}^\delta \nor{A^\delta \cP^\perp \phi}{}, \label{eq:ineq_cons} \\
	\ab{E - \cE}& \le 4\pa{c_H + c_A^2 \ab{E}} \pa{1 + c_H \nor{ARA}{} }^2 \nor{A\cP^\perp \phi}{}^2, \label{eq:proof_E_cE_bounded}
\end{align}
see Section~\ref{sub:Inequalities} to have more details on the derivation of those inequalities. Thus those errors are controlled by the key quantity $\nor{A\cP^\perp \phi}{}$.
\end{remark}

\begin{remark}[Scaling of the different terms in~\eqref{eq:equality_diff} and~\eqref{eq:equality_E_diff}]
The term $\cP^\perp \phi = \cP^\perp \pa{\phi - \psi}$ is controlled by $\nor{\phi - \psi}{e}$ in norm. 
When $\nor{\phi - \psi}{e}$ is small, the leading term in~\eqref{eq:equality_diff} is $\pa{1 + R H} \cP^\perp \phi$. Then, the second leading term is of order 2 and $\pa{\cE - E} R \pa{\phi - \psi}$ is of order 3 in the eigenvector error. 
In~\eqref{eq:equality_E_diff}, the leading term is $\ps{\cP^\perp \phi, \pa{\cE - H} \pa{1 + R H} \cP^\perp \phi}$ (which is of order 2), the second term is of order 4, the third term is of order 4 and the last one of order 6, all measured in the eigenvector error $\nor{\phi - \psi}{e}$. Hence we can write
\begin{align*}
\phi - \psi = (1+RH) \cP^\perp \phi + O_{\nor{\phi - \psi}{e,\delta} \rightarrow 0}\pa{\nor{\phi - \psi}{e,\delta}^2}.
\end{align*}
\end{remark}

\begin{remark}[General approximability]
Consider the vector case corresponding to Proposition~\ref{prop:non-deg_main_thm}. We numerically see that making $\cP$ larger by adding more vectors to $\cP \cH$ decreases the error, in general. This can be expected from the form of the leading term $\pa{1 + R H} \cP^\perp \phi$, in which, for any vector $\vp \in \cH$, $\cP^\perp \vp$ decreases. However, as will be seen in Section~\ref{sec:illustration}, there are some exceptional cases where making $\cP$ larger increases the error.
\end{remark}

\begin{remark}[Variational formulation]
	One can choose a variational approach. For instance if $\mu \in \Nz$, and $E_\mu$ and $\cE_\mu$ are defined as the $\mu^{\text{th}}$ lowest eigenvalues (which we assume to exist) of resp. $H$ and $\PHP$, then
\begin{align}\label{eq:variational} 
	E_\mu = \myinf{B \subset \cH \\ \dim B = \mu +1} \mymax{\vp \in B \\ \nor{\vp}{}=1} \ps{\vp, H \vp} \le \cE_\mu = \myinf{B \subset \cP \cH \\ \dim B = \mu +1} \mymax{\vp \in B \\ \nor{\vp}{}=1} \ps{\vp, H \vp}.
\end{align}
Depending on $\cP$, it is not necessarily $\cE_\mu$ which is close to $E_\mu$, but it can be $\cE_\nu$ for $\nu < \mu$. The drawback of this approach is that it is not able to treat eigenmodes being above the continuous spectrum.

Moreover, in the physics literature, the reduced basis method is often called the variational approximation, as we will discuss in Appendix~\ref{eq:variational}. However, we do not used this last terminology because of the fact that the variational point of view~\eqref{eq:variational} can only treat the isolated eigenvalues below the essential spectrum.

Note that the variational formulation can also refer to the fact that $\ps{v,H\phi}=E \ps{v,\phi}$ for any $v \in \cH$ and $\ps{u,H\psi} = \cE \ps{u,\psi}$ for any $u \in \cP \cH$.
\end{remark}

\begin{remark}[Comparison to the Feshbach-Schur operator]
	Using the effective operator $\PHP$ can be put in perspective with the Feshbach-Schur technics, see~\cite{Schur18,BacFroSig98b,DusSigSta21} and~\cite[Section 3.1]{SjoZwo07}, in which the effective operator is $\restr{\pa{\cP (H + H r_\perp H) \cP}}{\cP \cH \rightarrow \cP \cH}$, where $r_\perp := \pa{\pa{E - H}_{\mkern 1mu \vrule height 2ex\mkern2mu \cP^\perp \cH \rightarrow \cP^\perp \cH}}^{-1}$.
\end{remark}

\section{Definitions and assumptions for RBM+PT}%
\label{sec:definitions_rbm_pt}

In this section we present the context of perturbation theory, to prepare to couple it with the reduced basis method. We start by introducing some definitions and making some assumptions, which will enable to apply Rellich's theorem and Theorems~\ref{thm:main_deg_thm} and~\ref{prop:non-deg_main_thm}.

\subsection{Analytic family of operators}%
\label{sub:Analytic family of operators}

We present here assumptions which will be sufficient to use the Rellich theorem on the existence of analytic eigenmodes.

First, let us take a closed self-ajoint operator $H^0$ acting on a separable Hilbert space $\cH$, such that $\sigma(H^0) \neq \R$, so there exists $r \in \R$ and $\ep > 0$ such that 
\begin{align}\label{eq:non_empty_resolvent} 
\sigma(H^0) \, \cap \, ]r-\ep,r+\ep[ \, = \varnothing.
\end{align}
As in Section~\ref{sub:First definitions}, we take a self-adjoint operator $A$ acting on $\cH$, implementing the energy norm. For instance, one could take $A$ to be $\ab{H^0 - r}^{\f 12}$ but we do not necessarily make this choice.

We then choose a simple case for the family of operators, that is, we consider $M \in \Nz$, a finite number of self-adjoint operators $H^n$ for $n \in \{0,\dots,M\}$ such that 
 where $D(\cdot)$ denotes the domain of an operator, and such that
\begin{align}\label{eq:cinfty0} 
\mymax{0 \le n \le M} \nor{H^n \pa{H^0  - r}^{-1}}{} < +\infty, \qquad \mymax{0 \le n \le M} \nor{A^{-1} H^n A^{-1}}{} < +\infty.
\end{align}
We also define $H^n := 0$ for any $n \ge M +1$. Finally, we define the analytic family of operators $H(\lambda)$ with
\begin{align*}
H(\lambda) := \sum_{n=0}^{+\infty} \lambda^n H^n.
\end{align*}

\subsection{Choose a set of eigenvalues of $H(\lambda)$}%
\label{sub:choose_set_eigvals}

Take $\nu \in \Np$. Let us assume that $H^0$ has at least $\nu$ eigenvalues 
 \begin{align}\label{eq:E_discr_spec} 
\{E_\mu^0\}_{\mu =1}^\nu \subset \sigma_{\text{d}}(H^0)
 \end{align}
 in the discrete spectrum, counted with multiplicity but not necessarily sorted in increasing order. By~\eqref{eq:spectrum_gap_general}, there exists $\kappa_{H^0} > 0$ such that
\begin{align}\label{eq:spectral_free} 
\pa{\sigma(H^0) \backslash \{ E_\mu^0 \}_{\mu=1}^\nu} \cap \bpa{\cup_{\mu=1}^\nu \; ]E_\mu^0 - \kappa_{H^0} , E_\mu^0 + \kappa_{H^0}[ \;} = \varnothing,
\end{align}
let us assume that
\begin{align}\label{eq:complete_eigen} 
\dim \; \bigoplus_{\mu = 1}^\nu \Ker \pa{H^0 - E_{\mu}^0} = \nu.
\end{align}
Hence the spectrum of $H^0$ in $\cup_{\mu=1}^\nu \; ]E_\mu^0 - \kappa_{H^0} , E_\mu^0 + \kappa_{H^0}[$ is formed by exactly $\nu$ eigenmodes. 

Rellich's theorem states that the isolated clusters of eigenmodes of parameter dependent operators are also analytic in $\lambda$, see \cite[Theorem XII.3 p4]{ReeSim4}, \cite[Problem XII.17, p71]{ReeSim4}, \cite[Theorem 1.4.4 p25]{Simon15} and~\cite[Theorem 1 p21]{Baumgartel85} for instance. The extension to infinite-dimensional space also holds under some technical assumptions, see~\cite{Kato}, ~\cite[Lemma p16]{ReeSim4}, ~\cite[Theorem XII.8 p15]{ReeSim4} and \cite[Theorem XII.13 p22]{ReeSim4}, that are satisfied since we assumed~\eqref{eq:cinfty0}, see more details about this on Section~\ref{sub:A note on analyticity}.

We denote by $\pa{E_\mu(\lambda),\phi_\mu(\lambda)}$ the eigenmodes of $H(\lambda)$, analytic in $\lambda$, respecting $E_\mu(\lambda) = E^0_\mu$ and $\ps{\phi_\mu(\lambda) , \phi_\alpha(\lambda)} = \delta_{\mu\alpha}$ for any $\mu,\alpha \in \{1,\dots,\nu\}$. The basis of the vectors is not fixed by those conditions, meaning that taking smooth maps $\theta_\mu : \R \rightarrow \R$, the eigenvectors $e^{i\theta_\mu(\lambda)} \phi_\mu(\lambda)$ also respect the previous conditions. 

For any $\lambda \in ]-\lambda_0,\lambda_0[$, we define $\Gamma(\lambda) := \dm{\bphi(\lambda)}$ (we recall that the definition of density matrices is in~\eqref{eq:def_dm}) and the partial inverse
\begin{align}\label{def:pseudoinv_lambda_K}
K_\mu(0) := 
\left\{
\begin{array}{ll}
	\pa{\pa{E_\mu(0) -  H(0) }_{\mkern 1mu \vrule height 2ex\mkern2mu  \pa{\Gamma(0)}^\perp \cH \rightarrow \pa{\Gamma(0)}^\perp \cH}}^{-1} & \mbox{on }  (\Gamma(0))^\perp \cH, \\
0 & \mbox{on } \Gamma(0) \cH,
\end{array}
\right.
\end{align}
extended by linearity on $\cH$. By~\eqref{eq:spectral_free} we have $\nor{K_\mu(0)}{}\le \kappa_{H^0}^{-1} $, and we assume that
\begin{align}\label{eq:AKA_bound} 
	\mymax{1 \le \mu \le \nu} \nor{A K_\mu(0) A}{} < +\infty.
\end{align}

\subsection{Starting point for $\PHlP$}%
\label{sub:Starting point for the reduced operator}

The starting point of the analysis of the reduced operator will be $\lambda = 0$, on which the eigenmodes under study of the exact and reduced operators are equal. So the first step consists in exploiting this fact.

Let us consider an orthogonal projection $\cP$, where $\cP \cH$ can be infinite-dimensional. The hypothesis of RBM+PT, which we will see later, implies that the exact eigenvector $\phi_\mu(0)$ belongs to $\cP \cH$, hence $\cP H(0) \cP \phi_\mu(0) = E_\mu(0) \phi_\mu(0)$, so $\phi_\mu(0)$ is also an eigenvector of $(\cP H(0) \cP)_{\cP \cH \rightarrow \cP \cH}$ with eigenvalue $E_\mu(0)$. We need to assume that 
 \begin{align}\label{hypo:stable_spectrum}
	 \{E_\mu(0)\}_{\mu =1}^\nu \subset \sigma_{\text{d}} \bpa{\restr{\pa{\cP H(0) \cP}}{\cP \cH \rightarrow \cP \cH} }
 \end{align}
and
\begin{align}\label{eq:stable_eigenspaces} 
\dim \; \cP \cH \cap \bigoplus_{\mu=1}^\nu \Ker ( \cP H(0) \cP - E_\mu(0)) = \nu.
\end{align}
Those last assumptions mean that the reduction from $\cH$ to $\cP \cH$ does not produce spectral pollution close to the $E_\mu(0)$'s for $H(0)$.

\subsection{Analytic branches for $\PHlP$}%
\label{sub:Analytic branches}

To be able to apply Rellich's theorem for $\PHlP$, we make several assumptions. Let us assume that $\sigma\pa{\PHlOP} \neq \R$, so there exists $r_\cP \in \R$ and $\ep_\cP > 0$ such that 
\begin{align}\label{eq:non_empty_resolvent_P} 
\sigma\pa{\PHlP} \, \cap \, ]r_\cP-\ep_\cP,r_\cP+\ep_\cP[ \, = \varnothing,
\end{align}
uniformly in $\lambda \in ]-\lambda_0, \lambda_0[$ for some $\lambda_0 > 0$. Assume that
\begin{align}\label{eq:cinfty0_P} 
	\mysup{n \in \Np} \nor{\cP H^n \cP \pa{\pa{\cP H^0 \cP  - r_\cP}_{\mkern 1mu \vrule height 2ex\mkern2mu \cP \cH \rightarrow \cP \cH}}^{-1}}{} < +\infty.
\end{align}
Rellich's theorem ensures the existence of $\nu$ eigenmodes $\pa{\cE_\mu(\lambda), \psi_\mu(\lambda)}_{\mu =1}^{ \nu}$ of $\PHlP$, analytic in $\lambda \in ]-\lambda_0,\lambda_0[$ where $\lambda_0 >0$, such that $\cE_\mu(0) = E_\mu(0)$, $\psi_\mu(0) = \phi_\mu(0)$ and $\ps{\psi_\mu(\lambda) , \psi_\alpha(\lambda)} = \delta_{\mu\alpha}$ for any $\mu,\alpha \in \{1,\dots,\nu\}$. We take $\lambda_0$ small enough so that for some $\kappa_{H} > 0$ (which does not depend on $\lambda$) and any $\lambda \in ]-\lambda_0 , \lambda_0[$,
\begin{multline}\label{eq:non_deg_lambda}
	\bpa{\sigma\bpa{\PHlP} \backslash \{\cE_\mu(\lambda)\}_{\mu=1}^\nu} \\
\cap \bpa{\cup_{\mu=1}^\nu \; ]\cE_\mu(\lambda) - \kappa_{H}  , \cE_\mu(\lambda) + \kappa_{H} [ \;} = \varnothing,
\end{multline}
meaning that the rest of the spectrum remains far from $\{\cE_\mu(\lambda)\}_{\mu=1}^\nu$, uniformly in $\lambda$. Together with~\eqref{eq:stable_eigenspaces}, this implies that for any $\lambda \in ]-\lambda_0 , \lambda_0[$,
\begin{align*}
\dim \; \cP \cH \cap \bigoplus_{\mu=1}^\nu \Ker ( \cP H(\lambda) \cP - E_\mu(\lambda)) = \nu.
\end{align*}
For any $\lambda \in ]-\lambda_0 , \lambda_0[$ we can hence define
\begin{multline}\label{def:pseudoinv_R_lambda}
R_\mu(\lambda) \\
:= 
\left\{
\begin{array}{ll}
\pa{\pa{\cE_\mu(\lambda) - \cP H(\lambda) \cP}_{\mkern 1mu \vrule height 2ex\mkern2mu \Lambda(\lambda)^\perp \cP \cH \rightarrow \Lambda(\lambda)^\perp \cP \cH}}^{-1} & \mbox{on } \Lambda(\lambda)^\perp \cP \cH, \\
0 & \mbox{on } \Lambda(\lambda) \cH \oplus \cP^\perp(\lambda) \cH,
\end{array}
\right.
\end{multline}
extended by linearity on $\cH$. From~\eqref{eq:non_deg_lambda} we have $\nor{R_{\mu}(\lambda)}{} \le \kappa^{-1}_{H(\lambda)}$, and we assume that there exists $M_{\cP,R} >0 $ such that
\begin{align}\label{eq:hypo_RP} 
\nor{A R_{\mu}(\lambda) A^{-1}}{} \le c_{\cP,R},
\end{align}
uniformly in $\lambda \in ]-\lambda_0,\lambda_0[$ and $\mu \in \{1,\dots,\nu\}$.

\section{Main results on RBM+PT}%
\label{sec:Application to RBM+PT}

We now apply the results of the previous section to a parameter dependent operator $H(\lambda)$, which is the setting of RBM+PT. We refer to Figure~\ref{fig:main_idea} to illustrate our reasoning. 

\subsection{Heuristics}%
\label{sub:Heuristics}

Let us first explain briefly the idea behind RBM+PT, at order $\ell \in \Nz$. Consider an eigenvector $\phi(\lambda) = \sum_{n=0}^{+\infty} \lambda^n \phi^n$ of $H(\lambda)$. If $\phi^0,\phi^1,\dots,\phi^\ell \in \cP\cH$, then $\cP \cH$ contains $\vp(\lambda) := \sum_{n=0}^{\ell} \lambda^n \phi^n$, i.e. $\cP \cH$ is able to produce this perturbative approximation vector, so we can expect RBM+PT to be at least as good as PT. By~\eqref{eq:equality_diff} (or~\eqref{eq:ineq_cons}) we have that $\phi(\lambda)-\psi(\lambda)$ is controlled by the key quantity
\begin{align*}
\cP^\perp \phi(\lambda) = \sum_{n=0}^{+\infty} \lambda^n \cP^\perp \phi^n \underset{\substack{\phi^n \in \cP\cH \\ \forall 0 \le n \le \ell}}{=} \;\sum_{n=\ell+1}^{+\infty} \lambda^n \cP^\perp \phi^n = \lambda^{\ell +1} \cP^\perp \phi^{\ell +1} + O(\lambda^{\ell+2}).
\end{align*}
So the leading order of the error in RBM+PT is controlled by $\lambda^{\ell +1} \cP^\perp \phi^{\ell +1}$, while the leading order of the error in PT is $\lambda^{\ell +1} \phi^{\ell +1}$. The acceleration factor of RBM+PT with respect to PT is thus approximately $\nor{\phi^{\ell +1}}{} / \nor{\cP^\perp \phi^{\ell +1}}{}$.

Locally around $\lambda = 0$ adding the derivatives of $\phi(\lambda)$ at $\lambda = 0$ is the optimal way of making $\cP \cH$ grow, in the goal of modeling $\phi(\lambda)$.

\subsection{Clusters of eigenmodes}%
\label{sub:Clusters of eigenmodes}

For any $\lambda \in ]-\lambda_0,\lambda_0[$, we recall that $\Gamma(\lambda) := \dm{\bphi(\lambda)}$. For any $n \in \Nz$ and any $\mu \in \{1,\dots,\nu\}$, we define
\begin{align}\label{eq:def_der} 
	\phi_\mu^n := \frac{1}{n!} \pa{\frac{\d^n}{\d \lambda^n}  \phi_\mu(\lambda)}_{\mkern 1mu \vrule height 2ex\mkern2mu \lambda = 0}, \qquad \qquad  \Gamma^{n}  := \frac{1}{n!}\pa{\frac{\d^{n}}{\d \lambda^{n}}  \Gamma(\lambda)}_{\mkern 1mu \vrule height 2ex\mkern2mu \lambda = 0}.
\end{align}
Section~\ref{sec:density matrix perturbation theory} is dedicated to the study of those coefficients $\Gamma^n$, see Proposition~\ref{prop:dmpt} to see how to obtain them with recursive formulas. Let us also define
\begin{align*}
\xi_{\textup{RBM+PT},\ell}^{\textup{deg}} &:= \nor{\sum_{\mu=1}^{\nu} \pa{1 + R_\mu(0)H(0)} \cP^\perp \Gamma^{\ell +1}P_{\phi_\mu(0)} + \text{adj.}  }{2,\delta}.
\end{align*}

The main theorem of this section is about the closeness of the density matrix $\Gamma(\lambda)$ associated to the exact operator $H(\lambda)$ with the one of the approximate operator $\PHlP$, when $\cP \cH$ contains the first $\ell + 1$ derivatives of $\Gamma(\lambda)$.

\begin{corollary}[RBM+PT in the perturbative regime, for clusters of eigenmodes]\label{cor:main_ec_dm}
As in Section~\ref{sub:Analytic family of operators}, consider a family $H(\lambda) := \sum_{n=0}^{M} \lambda^n H^n$, where $(H^n)_{0 \le n \le M}$ are closed self-adjoint operators. Assume~\eqref{eq:non_empty_resolvent}, \eqref{eq:cinfty0}, \eqref{eq:E_discr_spec}, \eqref{eq:complete_eigen}, \eqref{eq:AKA_bound} on the $H^n$'s, and consider an orthogonal projector $\cP$ such that~\eqref{hypo:stable_spectrum}, ~\eqref{eq:stable_eigenspaces}, \eqref{eq:non_empty_resolvent_P}, \eqref{eq:cinfty0_P} and~\eqref{eq:hypo_RP}. Consider $\nu$ analytic families of eigenmodes $\pa{E_\mu(\lambda), \phi_\mu(\lambda)}_{\mu=1}^\nu$ of $H(\lambda)$ with the phasis such that $\ps{\phi_\mu(\lambda),\phi_\alpha(\lambda)}{}=\delta_{\mu\alpha}$. Then for $\lambda$ small enough, $\PHlP$ has $\nu$ eigenmodes $\pa{\cE_\mu(\lambda), \psi_\mu(\lambda)}_{\mu=1}^\nu$, analytic in $\lambda$ such that $\cE_\mu(0) = E_\mu(0)$, $\psi_\mu(0) = \phi_\mu(0)$ and $\ps{\psi_\mu(\lambda), \psi_\alpha(\lambda)}{}= \delta_{\mu\alpha} $ for any $\mu,\alpha \in \{1,\dots,\nu\}$. We define $\bphi(\lambda) := \pa{\phi(\lambda)}_{\mu=1}^\nu$, $\bpsi(\lambda) := \pa{\psi(\lambda)}_{\mu=1}^\nu$, and the density matrices $\Gamma(\lambda) := \dm{\bphi(\lambda)}$ and $\Lambda(\lambda) := \dm{\bpsi(\lambda)}$ (see the definition of $\cD$ in~\eqref{eq:def_dm}). Given $\ell \in \Nz$, if
\begin{align}\label{eq:cond_der} 
\forall n \in \{0,\dots,\ell\}, \qquad \im \Gamma^n \subset \cP \cH,
\end{align}
then there exists $\lambda_0 > 0$ such that for any $\lambda \in ]-\lambda_0,\lambda_0[$ and $\delta \in \{0,1\}$,
\begin{align}\label{eq:eigendm_bound}
\ab{\nor{\Gamma(\lambda) - \Lambda(\lambda)}{2,\delta} - \ab{\lambda}^{\ell +1} \xi^{\textup{deg}}_{\textup{RBM+PT},\ell}} &\le c \pa{\ab{\lambda} b}^{\ell +2},
\end{align}
where $b$ and $c$ are independent of $\lambda$ and $\ell$. 
\end{corollary}

We give a proof in Section~\ref{sec:Proof of EC}. The next result provides a practical way of building the reduced space used in~\eqref{eq:cond_der}, via an explicit and simple basis of vectors.
\begin{lemma}[Building the reduced space for density matrices]\label{lem:build_red_space_dm} 
Consider the context of Corollary~\ref{cor:main_ec_dm}. Take $(\vp_\mu)_{\mu=1}^\nu \in \cH^\nu$ to be a basis of the unperturbed space $\oplus_{\mu=1}^\nu \Ker \pa{H(0) - E_\mu(0)}$. Then
\begin{align*}
\bigoplus_{n=0}^\ell \im \Gamma^n &= \Span\pa{\pa{\phi^n_\mu}_{1 \le \mu \le \nu}^{0 \le n \le \ell}} =\Span\pa{\pa{\Gamma^n \vp_\mu}_{1 \le \mu \le \nu}^{0 \le n \le \ell}}.
\end{align*}
\end{lemma}
A proof is provided in Section~\ref{sec:Proof of EC}. 

\subsection{One eigenmode}%
\label{sub:One eigenmode}

We now state the corresponding result as Corollary~\ref{cor:main_ec_dm} but in the non-degenerate case and for vectors. As in Proposition~\ref{prop:non-deg_main_thm} we drop the subscripts 1, so $R(\lambda) := R_1(\lambda)$, $\phi(\lambda) := \phi_1(\lambda)$, $\psi(\lambda) := \psi_1(\lambda)$, $E(\lambda) := E_1(\lambda)$, $\cE(\lambda) := \cE_1(\lambda)$, $\phi^n := \phi^n_1$. We define
\begin{align*}
	\xi_{\textup{RBM+PT},\ell}^{\textup{non-deg}} := \nor{ \pa{1 + R(0) H(0)} \cP^\perp \phi^{\ell +1}}{e,\delta}.
\end{align*}
\begin{corollary}[RBM+PT in the perturbative regime, one eigenmode]\label{cor:main_ec_vec}
We make the same assumptions as in Corollary~\ref{cor:main_ec_dm}, we take $\nu =1$ and remove the subscripts 1, and we take $\ell \in \Nz$. We recall that $(E(\lambda), \phi(\lambda))$ is an eigenmode of $H(\lambda)$ and that $(\cE(\lambda), \psi(\lambda))$ is an eigenmode of $\PHlP$ such that $\cE(0) = E(0)$ and $\psi(0) = \phi(0)$. We choose the phasis of $\phi(\lambda)$ and $\psi(\lambda)$ such that $\ps{\phi^0 , \phi(\lambda)} \in \R_+$ and $\ps{\phi(\lambda) , \psi(\lambda)} \in \R$. If
\begin{align}\label{eq:cond_der_vec} 
\forall n \in \{0,\dots,\ell\}, \qquad \pa{\tfrac{\d^n}{\d \lambda^n}  \phi(\lambda)}_{\mkern 1mu \vrule height 2ex\mkern2mu \lambda = 0} \in \cP \cH,
\end{align}
then there exists $\lambda_0 > 0$ such that for any $\lambda \in ]-\lambda_0,\lambda_0[$ and $\delta \in \{0,1\}$,
\begin{align}\label{eq:eigenvectors_bound}
\ab{\nor{\phi(\lambda) - \psi(\lambda)}{e,\delta} - \ab{\lambda}^{\ell +1} \xi^{\textup{non-deg}}_{\textup{RBM+PT},\ell}} &\le c \pa{\ab{\lambda} b}^{\ell +2},
\end{align}
where $b$ and $c$ are independent of $\lambda$ and $\ell$. 
\end{corollary}
We provide a proof in Section~\ref{sec:Proof of EC}.

\subsection{Remarks}%
\label{sub:Remarks}
Now, several remarks seem in order.

\begin{remark}[Error on eigenvalues]
	We define
	\begin{align*}
	\xi_{\textup{RBM+PT},\ell}^{\textup{non-deg},E} &:= \ab{\ps{\cP^\perp \phi^{\ell +1}, (H(0) - E(0)) \pa{1 + R(0) H(0)} \cP^\perp \phi^{\ell +1}}}.
	\end{align*}
Using Lemma~\ref{lem:err_eigenvals_err_eigenvects} and ~\eqref{eq:eigenvectors_bound}, we have
\begin{align}
\ab{\ab{E(\lambda) - \cE(\lambda)} - \ab{\lambda}^{2(\ell +1)} \xi^{\textup{non-deg},E}_{\textup{RBM+PT},\ell}} &\le c \pa{\ab{\lambda} b}^{2\ell +3}. \label{eq:eigenvals_bound}
\end{align}
\end{remark}

\begin{remark}[Error with explicit constant]
Inequality~\eqref{eq:eigenvectors_bound} could be written as
\begin{align*}
	\nor{\phi(\lambda) - \psi(\lambda)}{e,\delta} = \ab{\lambda}^{\ell +1} \xi^{\textup{non-deg}}_{\textup{RBM+PT},\ell} + O_{\lambda \rightarrow 0}\pa{ \pa{\ab{\lambda} b}^{\ell +2}},
\end{align*}
where $O_{\lambda \rightarrow 0}(\cdot)$ denote a function bounded in $\lambda$ and $\ell$ in a neighborhood of $\lambda = 0$.
\end{remark}

\begin{remark}[Equality of perturbation terms]
A consequence of~\eqref{eq:eigenvectors_bound} and~\eqref{eq:eigenvals_bound} is that for all $n \in \{0,\dots,\ell\}, k \in \{0,\dots,2\ell +1\}$,
\begin{align*}
\pa{\tfrac{\d^n }{\d \lambda^n} \psi(\lambda)}_{\mkern 1mu \vrule height 2ex\mkern2mu \lambda = 0} = \pa{\tfrac{\d^n }{\d \lambda^n} \phi(\lambda)}_{\mkern 1mu \vrule height 2ex\mkern2mu \lambda = 0}, \qquad \pa{\tfrac{\d^k }{\d \lambda^k} \cE(\lambda)}_{\mkern 1mu \vrule height 2ex\mkern2mu \lambda = 0} = \pa{\tfrac{\d^k }{\d \lambda^k} E(\lambda)}_{\mkern 1mu \vrule height 2ex\mkern2mu \lambda = 0},
\end{align*}
meaning that the first terms of the perturbative series of the considered eigenmode of $\PHlP$ are the same as the first terms of the perturbative series of the considered eigenmode of $H(\lambda)$.
\end{remark}

\begin{remark}[Intermediate normalization]
	Intermediate normalization is reviewed in Appendix~\ref{sec:interm_normalization}. Instead of building the reduced space $\cP \cH$ from the $\phi^n$'s, one can form it by using the eigenvectors in intermediate normalization, denoted by $\Phi^n$. Using this last normalization is more convenient because it involves less computations. From~\eqref{eq:recover_phi} we have
\begin{align}\label{eq:equality_spans} 
\Span\pa{\Phi^k, \; 0 \le k \le \ell}  =  \Span\pa{\phi^k, \; 0 \le k \le \ell}.
\end{align}
Hence one can form the reduced space of RBM+PT by using either intermediate or unit normalization perturbation vectors, this is equivalent.
\end{remark}



\begin{remark}[Comparison to perturbation theory]
	We provide a comparison of RBM+PT with PT in Section~\ref{sec:illustration}. 
\end{remark}

\begin{remark}[Generalization to higher-dimensional parameter space]
Corollary~\ref{cor:main_ec_dm} is stated for a one-dimensional parameter space, parametrized by $\lambda$, but one can straightforwardly extend it to general parameter spaces.
\end{remark}

\subsection{Vectors in the degenerate case}%
\label{sub:Vectors in the degenerate case}

The bounds of Corollary~\ref{cor:main_ec_dm} do not enable to obtain bounds on individual eigenvectors and individual eigenvalues in the degenerate case. Nevertheless, following a different strategy of proof can lead to such bounds and this is the purpose of this section.

\subsubsection{Assumptions on derivatives}%
\label{ssub:Assumptions on derivatives}

Let us define
\begin{align*}
	\hh := \pa{\Gamma(0) H^1 \Gamma(0)}_{\mkern 1mu \vrule height 2ex\mkern2mu  \Gamma(0) \cH \rightarrow \Gamma(0) \cH},
\end{align*}
the restriction of $\Gamma(0) H^1 \Gamma(0)$ as an operator of $\Gamma(0) \cH$. Let us make the following hypothesis on the $E_\mu(\lambda)$'s, but we could make them on the $\cE_\mu(\lambda)$'s, this is equivalent since $\Gamma(0) H^1 \Gamma(0) = \Lambda(0) H^1 \Lambda(0)$. We assume that
\begin{align}\label{eq:assumption_eq_E} 
\forall \alpha, \beta \in \{1,\dots,\nu\}, \; E_\alpha(0) = E_\beta(0),
\end{align}
i.e. the system is exactly degenerate. For any $\alpha \in \{1,\dots,\nu\}$, we define
\begin{align*}
E_\alpha'(0) := \pa{\tfrac{\d}{\d \lambda} E_\alpha(\lambda)}_{\mkern 1mu \vrule height 2ex\mkern2mu \lambda = 0},
\end{align*}
and it is well-known from first-order perturbation theory (see~\cite{Kato} for instance) that the $E_\alpha'(0)$'s are the eigenvalues of $\hh$. We take $\mu \in \{1,\dots,\nu\}$ and we make the assumption that 
\begin{align}\label{eq:E1_non_deg} 
\text{the eigenvalue } E_\mu'(0) \text{ is non-degenerate for $\hh$},
\end{align}
implying that for any $\alpha \in \{1,\dots,\nu\} \backslash \{\nu\}$, $E_\alpha'(0) \neq E_\mu'(0)$. Thus there exists $\kappa_{\hh} > 0$ such that 
\begin{align*}
\pa{\sigma \pa{\hh} \backslash \{E_\mu'(0)\} }\cap\,  ]E_\mu'(0) -\kappa_{\hh},E_\mu'(0)  + \kappa_{\hh}[  \, = \varnothing,
\end{align*}
and we can define
\begin{align}\label{eq:def_rr} 
\rr_\mu(0) := 
\left\{
\begin{array}{ll}
	\pa{ \pa{E_\mu'(0) - \hh}_{\mkern 1mu \vrule height 2ex\mkern2mu P_{\phi_\mu(0)}^\perp \Gamma(0) \cH \rightarrow P_{\phi_\mu(0)}^\perp \Gamma(0) \cH} }^{-1} & \hspace{-0.3cm}	\mbox{on } P_{\phi_\mu(0)}^\perp \Gamma(0) \cH, \\
0 & \hspace{-0.3cm}\mbox{on } \Gamma(0)^\perp \cH \oplus \C \phi_\mu(0),
\end{array}
\right.
\end{align}
extended by linearity on all of $\cH$. More explicitely, we have
\begin{align*}
	\rr_\mu (0)= \sum_{\substack{1 \le \alpha \le \nu \\ \alpha \neq \mu}} \pa{E_\mu'(0) - E_\alpha'(0)}^{-1} P_{\phi_\alpha(0)}.
\end{align*}
We then define, for $\delta \in \{0,1\}$,
\begin{align*}
	\xi^{\textup{deg}}_{\textup{RBM+PT},\mu,\ell} := \nor{\pa{1 + \rr_\mu(0) H^1 }\pa{1 + R_\mu(0) H^0 } \cP^\perp   \phi_\mu^{\ell +1}(0)}{e,\delta}.
\end{align*}

\subsubsection{Statement of the result}%
\label{ssub:Statement of the result}

We are now ready to state our last result on RBM+PT.

\begin{theorem}[Degenerate case with vectors]\label{thm:deg_vecs}
	We make the same assumptions as in Corollary~\ref{cor:main_ec_dm} except~\eqref{eq:cond_der}, so we consider a cluster of $\nu$ eigenmodes $\pa{E_\mu(\lambda) , \phi_\mu(\lambda)}_{\mu=1}^\nu$. Moreover, let us assume~\eqref{eq:assumption_eq_E}, take some $\mu \in \{1,\dots,\nu\}$ and assume~\eqref{eq:E1_non_deg}. We choose the phasis of $\phi_\mu(\lambda)$ and $\psi_\mu(\lambda)$ such that $\ps{\phi_\mu^0 , \phi_\mu(\lambda)} \in \R_+$ and $\ps{\phi_\mu^0 , \psi_\mu(\lambda)} \in \R_+$. Take $\ell \in \Nz$ and $\delta \in \{0,1\}$. If
\begin{align}\label{eq:cond_der_vec_deg} 
	\forall n \in \{0,\dots,\ell\}, \qquad \pa{\tfrac{\d^n}{\d \lambda^n}  \phi_\mu(\lambda)}_{\mkern 1mu \vrule height 2ex\mkern2mu \lambda = 0} &\in \cP \cH, \\
\forall \alpha \in \{1,\dots,\nu\}, \qquad \phi_\alpha(0) &\in \cP \cH \nonumber,
\end{align}
then there exists $\lambda_0 > 0$ such that for any $\lambda \in ]-\lambda_0,\lambda_0[$,
\begin{align}\label{eq:eigenvectors_bound_deg}
\ab{\nor{\phi_\mu(\lambda) - \psi_\mu(\lambda)}{e,\delta} - \ab{\lambda}^{\ell +1} \xi^{\textup{deg}}_{\textup{RBM+PT},\mu,\ell}} &\le c \pa{\ab{\lambda} b}^{\ell +2},
\end{align}
where $b$ and $c$ are independent of $\lambda$ and $\ell$. 
\end{theorem}

As in~\eqref{eq:sum_eigenvals_converge}, Corollary~\ref{cor:main_ec_dm} only provides a convergence of the density matrices and of the sum of eigenvalues in a cluster, not a convergence of the individual eigenvectors and eigenvalues. Hence Theorem~\ref{thm:deg_vecs} provides more information. An error in individual eigenvalues can be deduced from an error in individual eigenvectors by Lemma~\ref{lem:err_eigenvals_err_eigenvects}. The proof of Theorem~\ref{thm:deg_vecs}, provided in Section~\ref{thm:deg_vecs}, is very different from the ones of the previous results, and uses a purely perturbative approach.

\section{Comparison between RBM+PT and other methods}%
\label{sec:illustration}

In this section, we present a numerical experiment investigating RBM+PT in the perturbative regime. We consider non-degenerate levels, and the vector case, as treated in Corollary~\ref{cor:main_ec_dm}.

\subsection{Operators $H^n$}%
\label{sub:Operators}

We will work with periodic one-dimensional Schrödinger operators. Take $\cH = L_{\text{per}}^2(\R)$ to be the space of $L^2$ functions with period $L >0$, take $V_j : \R \rightarrow \R$ for $j \in \{0,1,2\}$ three smooth functions, $H^0 = -\Delta + V_0$, $H^1 = V_1$, $H^2 = V_2$ and $H^n = 0$ for any $n \ge 3$. We represent the $V_j$'s in Figure~\ref{fig:pot_sol} together with their ground states, i.e. of $H^0+V_j$, denoted by $u_j$.

\begin{figure}[h!]
\begin{center} 
\includegraphics[height=5cm,trim={0cm 0cm 0cm 0cm},clip]{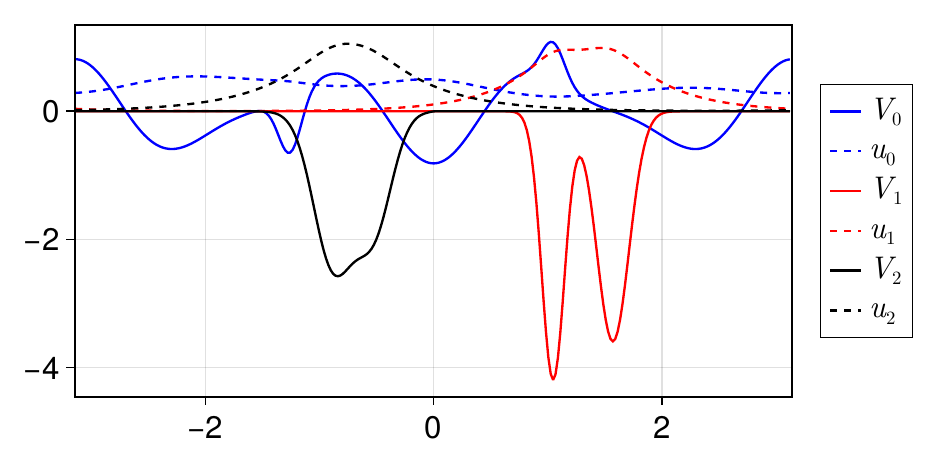}
\caption{Potentials $V_j$ for $j \in \{0,1,2\}$ and the ground states $u_j$ of $-\Delta + V_j$.}\label{fig:pot_sol}
\end{center}
\end{figure} 

\subsection{RBM+PT versus PT}%
\label{sub:RBM+PT versus PT}

\subsubsection{Perturbation theory}%
\label{ssub:Perturbation theory}

We define the approximation of $\phi_\mu(\lambda)$ given by PT and the corresponding eigenvalue approximation
\begin{align}\label{eq:pert_ap} 
\vp_\mu(\lambda) := \f{\sum_{n=0}^{\ell} \lambda^n \phi^n_\mu}{\nor{\sum_{n=0}^{\ell} \lambda^n \phi^n_\mu}{}},\qquad \qquad e_\mu(\lambda) := \ps{\vp_\mu(\lambda), H(\lambda) \vp_\mu(\lambda)}.
\end{align}
It is well-known that those quantities coming from perturbation theory respect the following bounds.
\begin{lemma}\label{lem:perturbation_theory_bound} 
	Let us make the definitions and assumptions of Sections~\ref{sub:Analytic family of operators} and~\ref{sub:choose_set_eigvals}. By defining $\xi^{\textup{non-deg}}_{\textup{PT},\ell} := \nor{\phi^{\ell+1}_\mu}{e,\delta}$ and $\xi^{\textup{non-deg},E}_{\textup{PT},\ell} := \ab{E^{2(\ell+1)}_\mu}$, for $\delta \in \{0,1\}$ we have
\begin{align*}
\ab{\nor{\phi_\mu(\lambda) - \vp_\mu(\lambda)}{e,\delta} - \ab{\lambda}^{\ell +1} \xi^{\textup{non-deg}}_{\textup{PT},\ell}} &\le c \pa{\ab{\lambda} b}^{\ell +2}, \\
\ab{\ab{E_\mu(\lambda) - e_\mu(\lambda)} - \ab{\lambda}^{2(\ell +1)} \xi^{\textup{non-deg},E}_{\textup{PT},\ell}} &\le c \pa{\ab{\lambda} b}^{2\ell +3},
\end{align*}
where $b$ and $c$ are independent of $\lambda$ and $\ell$. 
\end{lemma}
A proof is provided in Section~\ref{sec:Proof of Lemma}. 

\subsubsection{Acceleration factor}%
\label{ssub:Acceleration factor}

The errors given by RBM+PT and PT have the same order in $\ab{\lambda}$ but have different constants. The relevant quantity enabling to compare RBM+PT and PT in the asymptotic regime is $\xi_0 := 1$ and for $\ell \ge 1$,
\begin{align*}
	\xi_\ell := \mylim{\lambda \rightarrow 0} \f{\nor{\phi_\mu(\lambda) - \vp_\mu(\lambda)}{}}{\nor{\phi_\mu(\lambda) - \psi_\mu(\lambda)}{}} = \f{\xi^{\textup{non-deg}}_{\textup{PT},\ell,0}}{\xi_{\textup{RBM+PT},\ell,0}^{\textup{non-deg}}} = \f{\nor{\phi^{\ell+1}_\mu}{} }{\nor{ \pa{1 + R_\mu(0) H(0)} \cP^\perp \phi^{\ell+1}_\mu}{}},
\end{align*}
but one could also use $\nor{ \pa{1 + R_\mu(0) H(0)} \cP^\perp \phi^{\ell+1}_\mu}{e}^{-1} \nor{\phi^{\ell+1}_\mu}{e}$, which is very close. This quantifies the acceleration that RBM+PT provides with respect to PT. The larger $\xi_\ell$ is, the most efficient $\cP$ is. We numerically found situations such that $\xi_\ell < 1$ so RBM+PT is not necessarily better than PT, but in general we observe $\xi_\ell > 1$.


\subsubsection{Varying $\ell$}%
\label{ssub:Varying ell}
In this section, we aim at making $\ell$ vary. We choose $\cP = \cP^\ell$ to be the orthogonal projection onto 
\begin{align*}
\Span\pa{ \pa{\tfrac{\d^n}{\d \lambda^n}  \phi(\lambda)}_{\mkern 1mu \vrule height 2ex\mkern2mu \lambda = 0}, 0 \le n \le \ell}
\end{align*}
where $\phi(\lambda)$ is the eigenvector corresponding to the lowest eigenvalue $E(\lambda)$ of $H(\lambda)$, and we denote by $\psi^\ell(\lambda)$ the eigenvector of lowest eigenvalue $\cE^\ell(\lambda)$ of $\cP^\ell H(\lambda) \cP^\ell$. As in~\eqref{eq:pert_ap}, we define the perturbative approximations 
\begin{align}\label{eq:pt_approx} 
\vp^\ell(\lambda) := \f{\sum_{n=0}^{\ell} \lambda^n \phi^n}{\nor{\sum_{n=0}^{\ell} \lambda^n \phi^n}{}}, \qquad \text{and} \qquad 
e^\ell(\lambda) := \ps{\vp^\ell(\lambda), H(\lambda) \vp^\ell(\lambda)}.
 \end{align}

In Figure~\ref{fig:ell_varies}, we plot the different errors against $\lambda$ for small $\lambda$ near $\lambda = 0$. The asymptotic slopes correspond to~\eqref{eq:eigenvectors_bound} and~\eqref{eq:eigenvals_bound}. We see that the perturbation regime (the value of $\lambda$ for which the asymptotic slopes of $\lambda \rightarrow 0$ are considered) for PT is precisely attained around $\lambda \simeq 0.5$. On the contrary, in the case of RBM+PT, the start of the asymptotic regime is less systematic and depends on~$\ell$.

In Figure~\ref{fig:xi_against_ell}, we display the acceleration constant $\xi_\ell$ with respect to $\ell$. We also define $\xi_\ell^{\text{simple}} :=  \nor{\phi^{\ell+1}_\mu}{}/\nor{\cP^\perp \phi^{\ell+1}_\mu}{} $ to show in Figure~\ref{fig:xi_against_ell} that this simpler quantity is close to $\xi_\ell$. We see that the asymptotic behavior when $\ell \rightarrow +\infty$ is $\xi^{\textup{non-deg}}_{\textup{PT},\ell} \simeq c_{\text{pert}} s_{\text{pert}}^\ell$ and $\xi^{\textup{non-deg}}_{\textup{RBM+PT},\ell} \simeq c_{\text{ec}} s_{\text{ec}}^\ell$ with $s_{\text{ec}} < s_{\text{pert}}$. Hence we can conjecture that 
\begin{align*}
\nor{\phi(\lambda) - \vp^\ell(\lambda)}{} \simeq r_{\text{pert}} \pa{\ab{\lambda} q_{\text{pert}}}^\ell, \qquad 
\nor{\phi(\lambda) - \psi^\ell(\lambda)}{} \simeq r_{\text{ec}} \pa{\ab{\lambda} q_{\text{ec}}}^\ell
 \end{align*}
 with $q_{\text{ec}} < q_{\text{pert}}$, as if RBM+PT had the same error behavior as PT but where the perturbative regime is attained sooner than for perturbative theory.

\begin{figure}[h!]
\begin{center}

\hspace{-2cm}
\begin{minipage}{0.5\textwidth}
\includegraphics[width=8cm,trim={0.4cm 0cm 0cm 0.2cm},clip]{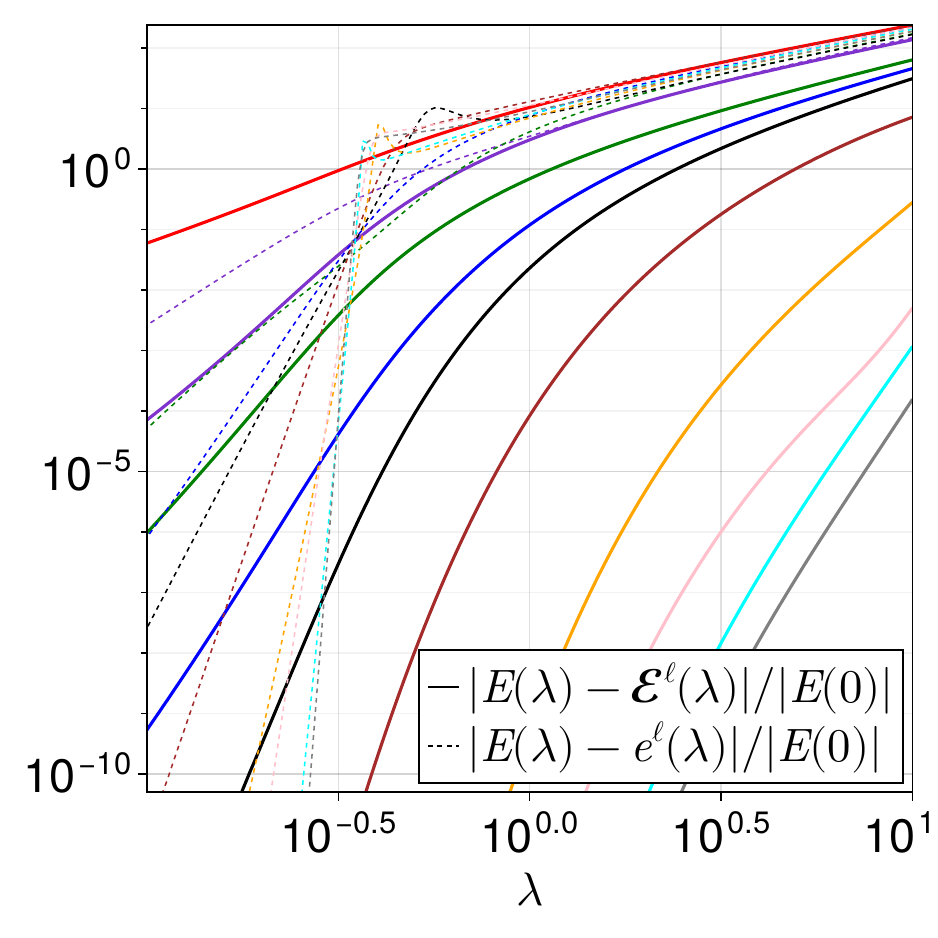}
\end{minipage}\hfill
\begin{minipage}{0.5\textwidth}
\includegraphics[width=8cm,trim={0.4cm 0cm 0cm 0.2cm},clip]{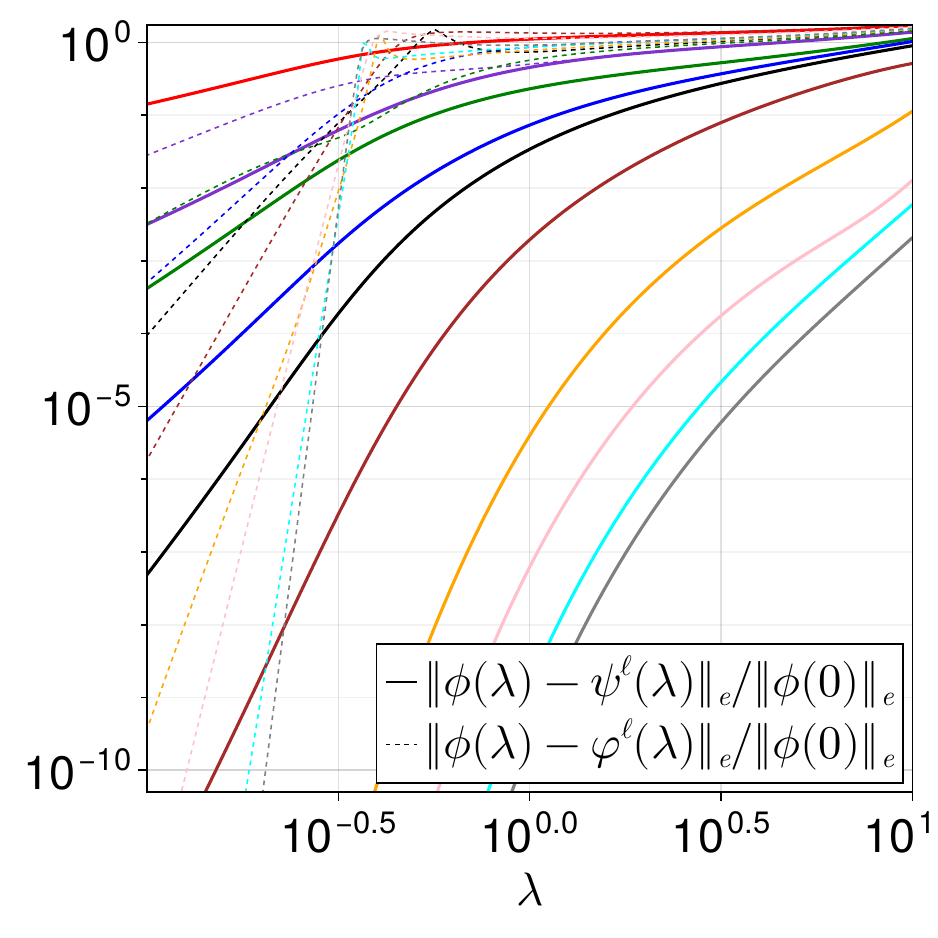}
\end{minipage}

\hspace{-0cm}\includegraphics[height=1cm,trim={0.2cm 0cm 37cm 0cm},clip]{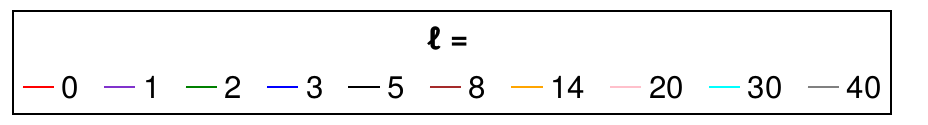}
\caption{Varying $\ell$, in the case of only one eigenmode. The asymptotic slopes near $0$ are $\ell +1$ on the right and $2\ell +2$ on the left, corresponding to~\eqref{eq:eigenvectors_bound} and~\eqref{eq:eigenvals_bound}. The error with respect to the PT approximation is in dashed line while the error with respect to the RBM+PT approximation is in plain line.}\label{fig:ell_varies}
\end{center}
\end{figure}

\begin{figure}[h!]
\begin{center}
\hspace{-7cm}
\begin{minipage}{0.9\textwidth}
\includegraphics[height=5cm,trim={0.2cm 0cm 0cm 0cm},clip]{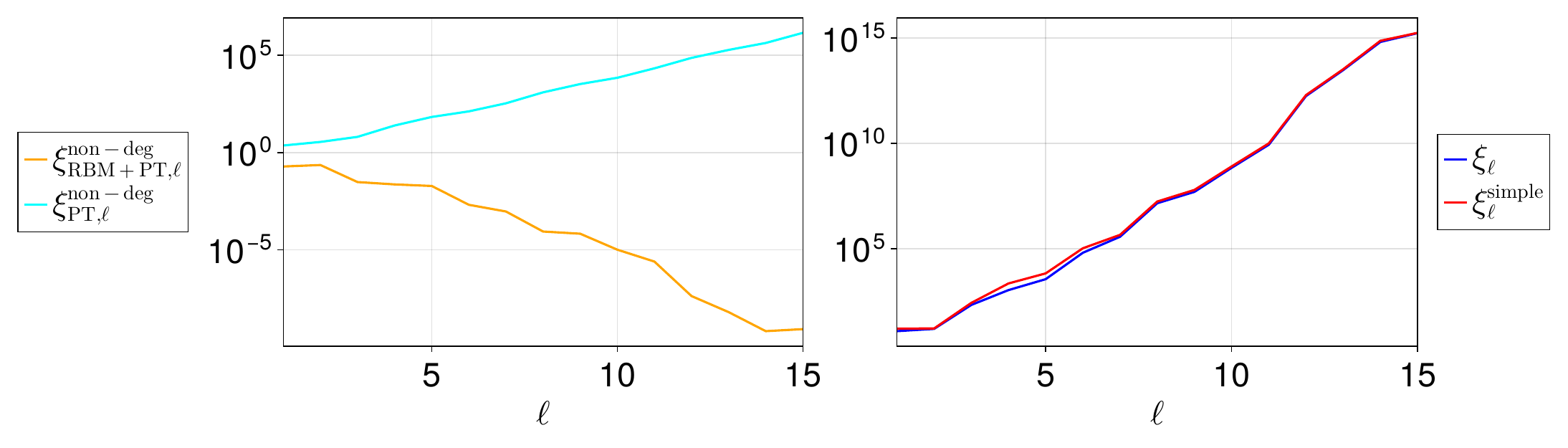}
\end{minipage}\hfill
\caption{Left: values of $\xi^{\textup{non-deg}}_{\textup{PT},\ell}$ and $\xi^{\textup{non-deg}}_{\textup{RBM+PT},\ell}$ against $\ell$, for $\delta = 0$. Right: acceleration factor $\xi_\ell$ against $\ell$. We approximately observe a behavior $\ln_{10} \xi_\ell \simeq 1.2 \ell$.}\label{fig:xi_against_ell}
\end{center}
\end{figure}

\subsubsection{Conclusions}%
\label{ssub:Conclusions}

From Figure~\ref{fig:xi_against_ell}, we see that 
\begin{align*}
	\xi_\ell^{\text{simple}} = \f{\nor{\phi^{\ell+1}_\mu}{}}{\nor{\cP^\perp \phi^{\ell+1}_\mu}{}}
\end{align*}
 is a good quantification of the asymptotic (when $\lambda \rightarrow 0$) acceleration rate of RBM+PT with respect to PT. We observed that it has a behavior close to $\xi_\ell^{\text{simple}} \simeq c \times 16^\ell$ as $\ell$ is large.

Moreover, as Figure~\ref{fig:ell_varies} shows, PT starts to provide a good approximation only in the perturbative regime, that is when $\lambda$ becomes small enough, while RBM+PT provides a good approximation for much larger values of $\lambda$.

Thus, in both perturbative and non-perturbative regimes, RBM+PT is a better approximation than PT in this example.

\subsection{RBM+PT vs RBM with excited states}%
\label{sub:RBM+PT vs RBM with excited states}
Let us now compare RBM+PT with another choice for the reduced basis. We define $\psi_\beta$ as the ground state of the operator $\PHlP$ when
\begin{align*}
	\cP \cH = \Span \pa{ \phi_\mu(0) \;|\; 0 \le \mu \le \beta}.
\end{align*}
This is a more traditional way of building reduced spaces, and we call it RBM+ES. In Figure~\ref{fig:rbm+es}, we represent the errors between the exact and the approximated eigenvectors, where $\psi^\beta$ is the ground state of the operator $\PHlP$ when
\begin{align*}
	\cP \cH = \Span \bpa{ \phi^k_0 \;|\; 0 \le k \le \beta},
\end{align*}
i.e. corresponding to RBM+PT. We already know that in the asymptotic regime, RBM+PT is more efficient, but we see that in this case, at $\lambda$ not small RBM+PT is again more competitive than RBM+ES. For instance, RBM+PT with 6 states has a similar precision than RBM+ES with 21 states.

\begin{figure}[h]
\begin{center}
\begin{minipage}{0.8\textwidth}
\includegraphics[width=8.2cm,trim={0.4cm 0cm 0cm 0.2cm},clip]{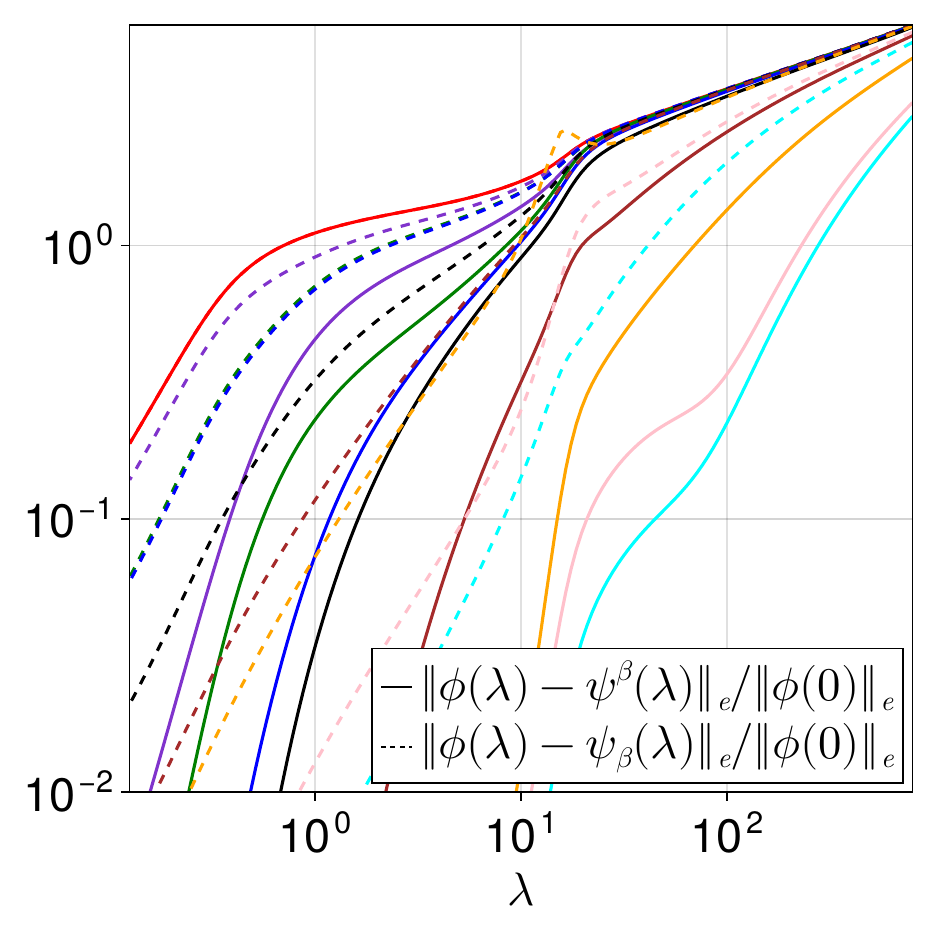}
\end{minipage}\hfill
\begin{minipage}{0.2\textwidth}
	\hspace{-1cm}
\includegraphics[height=6cm,trim={0cm 22cm 0cm 22.5cm},clip]{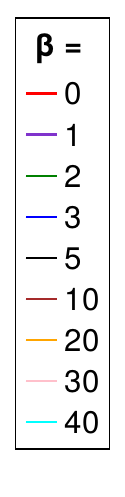}
\end{minipage}
\caption{Comparison of the errors between exact and approximated eigenstates. Here, $\psi^\beta$ corresponds to the error with RBM+PT, with a reduced space containing the $\beta +1$ first derivatives; $\psi_\beta$ corresponds to the error when the reduced space is built from the $\beta +1$ first eigenstates at $\lambda = 0$.}\label{fig:rbm+es}
\end{center}
\end{figure}


\section{Proof of Theorem~\ref{thm:main_deg_thm}}
\label{sec:Proof of Theorem mauin}

We recall that for any self-adjoint operators $B, C$ of $\cH$,
\begin{align}\label{eq:bound_norm} 
\nor{BC}{2} \le \nor{B}{} \nor{C}{2}.
\end{align}
Moreover, for any $u,v \in \cH$, we have
\begin{align}\label{eq:norm_form} 
\nor{\left| u \right> \left< v \right| }{2} = \nor{\left| u \right> \left< v \right| }{} = \nor{u}{} \nor{v}{}.
\end{align}

\subsection{Plan of the proof}%

We decompose the error $\Gamma - \Lambda$ into several terms, which will be possible to handle individually. It is based on the space decomposition $\cH = \cP \cH \oplus \cP^\perp \cH$. First, we have $\phi_\alpha \in \cP \cH$ for all $\alpha \in \{1,\dots,\nu\}$, so 
\begin{align*}
\cP \Lambda = \sum_{\mu=1}^{\nu} \left| \cP \phi_\alpha \right> \left< \phi_\alpha \right|  = \sum_{\mu=1}^{\nu} \left|  \phi_\alpha \right> \left< \phi_\alpha \right| = \Lambda,
\end{align*}
hence
\begin{align}\label{eq:commu} 
\cP \Lambda = \Lambda \cP = \Lambda, \qquad \qquad \cP^\perp \Lambda = \Lambda \cP^\perp = 0.
\end{align}
Then we can decompose the error $\Gamma - \Lambda$ in the following way
\begin{align}\label{eq:init_dec} 
	\Gamma - \Lambda & \hspace{-0.15cm}\hspace{-0.2cm}\hspace{-0.15cm}\underset{\substack{1 = \cP + \cP^\perp}}{=} \hspace{-0.2cm} \cP^\perp \pa{\Gamma - \Lambda} \cP^\perp + \cP \pa{\Gamma - \Lambda} \cP^\perp +\cP^\perp \pa{\Gamma - \Lambda} \cP +\cP \pa{\Gamma - \Lambda} \cP \nonumber \\
& \underset{\substack{\eqref{eq:commu}}}{=} \;\cP^\perp \Gamma \cP^\perp + \cP \Gamma \cP^\perp +\cP^\perp \Gamma \cP +\cP \pa{\Gamma - \Lambda} \cP.
\end{align}
We will follow those steps :
\begin{itemize}
\item in Section~\ref{sub:Treating first terms} we show how to treat the first terms $\cP^\perp \Gamma \cP^\perp$, $\cP \Gamma \cP^\perp$, and $\cP^\perp \Gamma \cP$,
\item in Section~\ref{sub:A first treatment of} we present a first way of treating the term $\cP \pa{\Gamma - \Lambda} \cP$, by following the decomposition $\cP \cH = \Lambda \cH \oplus \Lambda^\perp \cP \cH$. This decomposes $\cP \pa{\Gamma - \Lambda} \cP$ into three types of terms. We treat the different obtained subterms in Sections~\ref{sub:Treating lambda}, \ref{sub:Treating lambdaperp}, \ref{sub:Definition of partial inverses} and \ref{sub:Treating lambda_lambda_perp}. We gather all the computations in Section~\ref{sub:First form for},
\item in Appendix~\ref{sec:Alternative bound} we present a second way of treating $\cP \pa{\Gamma - \Lambda} \cP$, base on a resolvent computation, leading to a different kind of inequalities.
\end{itemize}

\subsection{Treating $\cP^\perp \Gamma \cP^\perp$, $\cP \Gamma \cP^\perp$ and $\cP^\perp \Gamma \cP$}
\label{sub:Treating first terms}

We start by treating the first terms of~\eqref{eq:init_dec}. We have
\begin{multline*}
	\nor{\cP^\perp \Gamma \cP^\perp }{2,\delta} \underset{\substack{\Gamma = \Gamma^2}}{=} \; \nor{A^\delta \cP^\perp \Gamma^2 \cP^\perp}{2} \underset{\substack{\eqref{eq:bound_norm}}}{\le} \; \nor{ \cP^\perp \Gamma }{2,\delta} \nor{ \cP^\perp \Gamma }{} \\
	= \nor{ \cP^\perp \Gamma }{2,\delta} \nor{A^{-\delta} A^\delta \cP^\perp \Gamma }{} \le c_A^\delta \nor{ \cP^\perp \Gamma }{2,\delta}^2.
\end{multline*}
Then, 
\begin{align*}
& \nor{A^\delta \cP^\perp \Gamma \cP}{2} + \nor{A^\delta \cP \Gamma \cP^\perp}{2} \le c_A^\delta \pa{\nor{A^\delta \cP^\perp \Gamma \cP A^\delta}{2} + \nor{A^\delta \cP \Gamma \cP^\perp A^\delta}{2}} \\
&\qquad = 2 c_A^\delta \nor{A^\delta \cP^\perp \Gamma \cP A^\delta}{2} = 2 c_A^\delta \nor{A^\delta \cP^\perp \Gamma^2 A^{\delta} A^{-\delta} \cP A^\delta}{2} \\
&\qquad \le 2 c_A^\delta \nor{A^\delta \cP^\perp \Gamma}{2} \nor{\Gamma A^\delta}{} \nor{A^{-\delta} \cP A^\delta}{}  = 2 \pa{c_A c_\cP \nor{A \Gamma}{} }^\delta\nor{\cP^\perp \Gamma}{2,\delta}.
\end{align*}

\subsection{Decomposition of $\cP \pa{\Gamma - \Lambda} \cP$}%
\label{sub:A first treatment of}

We present a first treatment of the operator $\cP \pa{\Gamma - \Lambda} \cP$, based on the decomposition $\cP = \cP (\Lambda + \Lambda^\perp) = \cP \Lambda^\perp + \Lambda$. More precisely,
\begin{align}\label{eq:dec_PdiffP} 
	\cP \pa{\Gamma - \Lambda} \cP &= \cP \Lambda^\perp \pa{\Gamma - \Lambda} \Lambda + \Lambda\pa{\Gamma - \Lambda} \cP \Lambda^\perp \nonumber \\
&\qquad \qquad \qquad + \Lambda \pa{\Gamma - \Lambda} \Lambda + \cP \Lambda^\perp \pa{\Gamma - \Lambda} \Lambda^\perp \cP\nonumber \\
&\underset{\substack{\eqref{eq:commu}}}{=} \;\cP \Lambda^\perp  \Gamma \Lambda + \Lambda \Gamma \cP \Lambda^\perp + \Lambda \pa{\Gamma - \Lambda} \Lambda + \cP \Lambda^\perp \Gamma \Lambda^\perp \cP.
\end{align}
In the following sections, we provide inequalities for each of those terms.

\subsection{Treating $\Lambda \pa{\Gamma - \Lambda} \Lambda$}%
\label{sub:Treating lambda}

On the first hand,
\begin{align*}
	\pa{\Lambda \pa{\Gamma - \Lambda} \Lambda}^2 = \pa{\Lambda \Gamma \Lambda - \Lambda}^2 = \Lambda \Gamma \Lambda \Gamma \Lambda - 2\Lambda \Gamma \Lambda+ \Lambda 
\end{align*}
hence
\begin{align*}
\nor{\Lambda \pa{\Gamma - \Lambda} \Lambda}{2}^2 =\tr \pa{\Lambda \pa{\Gamma - \Lambda} \Lambda}^2 = \tr \Gamma \Lambda \Gamma \Lambda - 2 \tr \Gamma \Lambda + \nu.
\end{align*}
On the other hand,
\begin{align*}
	\pa{\Gamma - \Lambda}^4 &= \pa{\Gamma + \Lambda - \Lambda \Gamma - \Gamma \Lambda}^2 \\
&=\Gamma + \Lambda - \Lambda \Gamma \Lambda - \Gamma \Lambda \Gamma - \Gamma \Lambda - \Lambda \Gamma + \Gamma \Lambda \Gamma \Lambda + \Lambda \Gamma \Lambda \Gamma,
\end{align*}
so
\begin{align*}
\nor{\pa{\Gamma - \Lambda}^2}{2}^2 = 2 \pa{\tr \Gamma \Lambda \Gamma \Lambda - 2 \tr \Gamma \Lambda + \nu}. 
\end{align*}
Then, we see that
\begin{align}\label{eq:ineq_lambda_gamma_lambda} 
	\nor{\Lambda \pa{\Gamma - \Lambda} \Lambda}{2} = \f{1}{\sqrt{2}} \nor{\pa{\Gamma - \Lambda}^2}{2} \le \f{1}{\sqrt{2}} \nor{\Gamma - \Lambda}{2}^2.
\end{align}
Finally,
\begin{align}\label{eq:first_ine} 
\nor{\Lambda \pa{\Gamma - \Lambda} \Lambda}{2,\delta} &= \nor{A^\delta \Lambda \pa{\Gamma - \Lambda} \Lambda }{2} = \nor{A^\delta \Lambda^2 \pa{\Gamma - \Lambda} \Lambda }{2}\nonumber \\
&\le \nor{A^\delta \Lambda}{} \nor{\Lambda \pa{\Gamma - \Lambda} \Lambda}{2} = \nor{A \Lambda}{}^{\delta} \nor{\Lambda \pa{\Gamma - \Lambda} \Lambda}{2}\nonumber \\
& \underset{\substack{\eqref{eq:ineq_lambda_gamma_lambda}}}{\le} \; 2^{-\f 12} \nor{A \Lambda}{}^{\delta} \nor{\Gamma - \Lambda}{2}^2 \le 2^{-\f 12} \pa{c_A^2 \nor{A \Lambda}{}}^{\delta} \nor{\Gamma - \Lambda}{2,\delta}^2 \nonumber \\
& \le  \pa{1 + c_A (1+c_A)\nor{A \Lambda}{}}^{2\delta} \nor{\Gamma - \Lambda}{2,\delta}^2.
\end{align}

\subsection{Treating $\cP \Lambda^\perp \Gamma\Lambda^\perp \cP$}%
\label{sub:Treating lambdaperp}

We have
\begin{align}\label{eq:bound_APLambdapAiv} 
\nor{A \cP \Lambda^\perp A^{-1}}{}  = \nor{A \cP A^{-1} \pa{1 - A \Lambda A^{-1}}}{} \le c_\cP \pa{1 +  c_A \nor{A \Lambda}{}}.
\end{align}
Now, we develop
\begin{align}\label{eq:double_p_perp} 
&\nor{\cP \Lambda^\perp \Gamma \Lambda^\perp \cP}{2,\delta} \le c_A^{\delta} \nor{A^\delta \cP \Lambda^\perp \Gamma^2 \Lambda^\perp \cP A^\delta}{2} \le c_A^\delta \nor{A^\delta \cP \Lambda^\perp \Gamma}{2}^2 \nonumber \\
&\qquad = c_A^\delta \nor{A^\delta \cP \Lambda^\perp \pa{\Gamma - \Lambda}}{2}^2 = c_A^\delta \nor{A^\delta \cP \Lambda^\perp A^{-\delta} A^\delta \pa{\Gamma - \Lambda} }{2}^2 \nonumber\\
&\qquad\le c_A^\delta \nor{A^\delta \cP \Lambda^\perp A^{-\delta}}{}^2 \nor{\Gamma - \Lambda }{2,\delta}^2   \le c_A^\delta \nor{A\cP \Lambda^\perp A^{-1}}{}^{2\delta} \nor{\Gamma - \Lambda }{2,\delta}^2 \nonumber\\
&\qquad \underset{\substack{\eqref{eq:bound_APLambdapAiv}}}{\le} \; \pa{c_A c_\cP^2}^\delta  \pa{1 +  c_A (1+c_A) \nor{A \Lambda}{}}^{2\delta}\nor{\Gamma - \Lambda }{2,\delta}^2.
\end{align}

\subsection{Definition and properties of partial inverses}%
\label{sub:Definition of partial inverses}

To prepare the next section, we need Liouvillian operators, which are standard tools to partially invert Hamiltonians acting on density matrices, see for instance \cite{Kato,Teufel03,BacDerWoj18,MonTeu19} and~\cite[Section 5.1]{CanKamLev21}. We show several basic equations that will be used.

We recall that $R_\mu = \pa{\cE_\mu - \cP H \cP}^{-1}_\perp$. We define the super-operators $\cL$ and $\cL^+$ acting on $\hls$ by
\begin{align}\label{eq:def_cLs} 
B \mapsto \cL B := [\cP H \cP, B], \qquad B \mapsto \cL^+ B := - \sum_{\mu=1}^\nu R_\mu B P_{\psi_\mu}
\end{align}
and the subspaces
\begin{align*}
\cO_1 := \{ B \in \hls , B = \cP \Lambda^\perp B \cP^\perp\}, \qquad \cO_2 := \{ B \in \hls , B = \cP^\perp B \cP \Lambda^\perp\}.
\end{align*}
By definition of $R_\mu$ we have 
\begin{align}\label{eq:KjHE} 
\cP\pa{\cE_\mu  - H}  R_\mu = \Lambda^\perp \cP = R_\mu\pa{\cE_\mu  - H} \cP .
\end{align}
We compute, for any $B \in \hls$,
\begin{align*}
	\cL^+ \cL B &= - \sum_{\mu=1}^\nu R_\mu [\cP H \cP, B] P_{\psi_\mu} \underset{\substack{\cP H \cP P_{\psi_\mu} \\= \cE_\mu P_{\psi_\mu}}}{=} \;  \sum_{\mu=1}^\nu \pa{ \cE_\mu R_\mu B P_{\psi_\mu} - R_\mu \cP H \cP B P_{\psi_\mu}} \\
		    & =\sum_{\mu=1}^\nu R_\mu \pa{\cE_\mu - H} \cP B P_{\psi_\mu}  \underset{\substack{\eqref{eq:KjHE}}}{=} \; \cP \Lambda^\perp B \sum_{\mu=1}^\nu P_{\psi_\mu} = \cP \Lambda^\perp B \Lambda.
\end{align*}
We can show that $\cL^+ \cL$ is the orthogonal projection onto $\cO_1$. We provide the details here as well for the sake of completeness. For any $B \in \hls$, we have
\begin{align*}
	\cL\cL^+ B &= \seg{\cP H \cP, - \sum_{\mu=1}^\nu R_\mu B P_{\psi_\mu} } = \sum_{\mu=1}^\nu \pa{\cE_\mu R_\mu B P_{\psi_\mu} - \cP H \cP R_\mu B P_{\psi_\mu}} \\
		   &= \sum_{\mu=1}^\nu \cP \pa{\cE_\mu - H} R_\mu B P_{\psi_\mu} = \Lambda^\perp \cP B \Lambda,
\end{align*}
hence $\cL \cL^+ = \cL^+ \cL$. Moreover, for any $F, B \in \hls$,
\begin{align*}
	\hs{F, \cL^+ \cL B} &= \tr F^* \Lambda^\perp \cP B \Lambda = \tr \Lambda F^* \Lambda^\perp \cP B = \tr \bpa{\cP \Lambda^\perp F \Lambda}^* B \\
	&= \hs{ \cL^+ \cL F, B}
\end{align*}
thus $(\cL^+ \cL)^* = \cL^+ \cL$. Finally, 
\begin{align*}
	\pa{\cL^+ \cL}^2 B = \pa{\cL^+ \cL} \bpa{\Lambda^\perp \cP B \Lambda} = \Lambda^\perp \cP\bpa{\Lambda^\perp \cP B \Lambda}\Lambda = \Lambda^\perp \cP B \Lambda = \cL^+ \cL B,
\end{align*}
hence $(\cL^+ \cL)^2 = \cL^+ \cL$, and we can conclude that $\cL^+ \cL$ is the orthogonal projection onto $\cO_1$.

From~\eqref{eq:commu}, we have that $\Lambda$ and $\Lambda^\perp$ commute with $\cP$ and $H$, hence for $Q, G \in \{\Lambda,\cP \Lambda^\perp\}$ and for any operator $B \in \hls$,
\begin{align}\label{eq:com_B} 
\cL \pa{Q B G} = Q \pa{\cL B} G.
\end{align}

\subsection{Treating $\cP \Lambda^\perp \Gamma \Lambda$ and $\Lambda \Gamma \cP \Lambda^\perp $}%
\label{sub:Treating lambda_lambda_perp}

We now use the Liouvillian operator to treat $\cP \Lambda^\perp \Gamma \Lambda$. The Euler-Lagrange equation for $\Gamma$ is $[H,\Gamma] = 0$ and can be verified by developing $\Gamma$ into projectors. There holds
\begin{align}\label{eq:this_one} 
	\cL \Gamma  & = [\cP H \cP , \Gamma] = \cP [H,\Gamma] \cP - \pa{\cP H [\cP^\perp, \Gamma] + [\cP^\perp, \Gamma] H \cP} \nonumber \\
		    & \underset{\substack{[H,\Gamma] = 0}}{=} \; -\cP H [\cP^\perp, \Gamma] - [\cP^\perp, \Gamma] H \cP.
\end{align}
Next,
\begin{align}\label{eq:lqlql} 
	\cP \Lambda^\perp \Gamma \Lambda &= \cL^+ \cL \Gamma \underset{\substack{\eqref{eq:def_cLs}\\\eqref{eq:this_one}}}{=} \; \sum_{\mu=1}^\nu R_\mu \pa{  \cP H [\cP^\perp, \Gamma] + [\cP^\perp, \Gamma] H \cP} P_{\psi_\mu} \nonumber \\
	& \underset{\substack{\cP^\perp P_{\psi_\mu} = 0 \\ R_\mu \cP^\perp = 0}}{=} \;\sum_{\mu=1}^\nu R_\mu \bpa{  H \cP^\perp\Gamma -  \Gamma \cP^\perp H} P_{\psi_\mu}.
\end{align}
Moreover,
\begin{align*}
	\cP^\perp \Gamma \cP = \cP^\perp \Gamma \Lambda + \cP^\perp \Gamma \Lambda^\perp \cP = \cP^\perp \Gamma \sum_{\mu=1}^{\nu} P_{\psi_\mu} + \cP^\perp \pa{\Gamma - \Lambda}^2 \Lambda^\perp \cP,
\end{align*}
where we see that the last term is quadratic in $\Gamma - \Lambda$ and hence will be negligible. Thus we remark the natural association
\begin{multline}\label{eq:association} 
	\cP \Lambda^\perp \Gamma \Lambda + \cP^\perp \Gamma \cP = \sum_{\mu=1}^\nu \pa{1 +R_\mu H}  \cP^\perp\Gamma P_{\psi_\mu} \\
	- \sum_{\mu=1}^{\nu}  R_\mu \Gamma \cP^\perp H P_{\psi_\mu} + \cP^\perp \pa{\Gamma - \Lambda}^2 \Lambda^\perp \cP.
\end{multline}
Taking the adjoint operator of~\eqref{eq:lqlql} yields
\begin{align*}
	\Lambda \Gamma \cP \Lambda^\perp  = \sum_{\mu=1}^\nu  P_{\psi_\mu} \bpa{\Gamma \cP^\perp H - H \cP^\perp \Gamma } R_\mu.
\end{align*}
As for the bounds, we have
\begin{align}\label{eq:minus_one_last} 
&\nor{(1+ R_\mu H) \cP^\perp \Gamma  P_{\psi_\mu} }{2,\delta} + \nor{ P_{\psi_\mu} \Gamma \cP^\perp (1+ H R_\mu) }{2,\delta} \nonumber \\
&\quad= \nor{A^\delta (1+ R_\mu H) \cP^\perp \Gamma  P_{\psi_\mu} \Lambda A^\delta A^{-\delta}}{2} + \nor{A^\delta   \Lambda P_{\psi_\mu}\Gamma \cP^\perp (1+ H R_\mu) A^\delta A^{-\delta}}{2}\nonumber \\
&\quad\le 2 c_A^\delta \nor{A^\delta (1+ R_\mu H) \cP^\perp \Gamma  P_{\psi_\mu} \Lambda A^\delta }{2}\nonumber\\
&\quad= 2 c_A^\delta \nor{A^\delta (1+ R_\mu H) \cP^\perp A^{-\delta} A^\delta \cP^\perp \Gamma  P_{\psi_\mu} \Lambda A^\delta }{2} \nonumber \\
&\quad\le 2 \pa{c_A\nor{A \Lambda }{}}^\delta \nor{A^\delta (1+ R_\mu H) \cP^\perp A^{-\delta}}{} \nor{\cP^\perp \Gamma}{2,\delta},
\end{align}
which is of order $1$ in $\nor{\Gamma-\Lambda}{2,\delta} $. Similarly,
\begin{align}\label{eq:last_ine} 
&\nor{R_\mu \Gamma \cP^\perp H P_{\psi_\mu}}{e,2} + \nor{P_{\psi_\mu} H \cP^\perp  \Gamma R_\mu }{e,2} \le 2 c_A^\delta \nor{A^\delta R_\mu \Gamma \cP^\perp H P_{\psi_\mu} A^\delta}{2} \nonumber \\
	&\qquad \underset{\substack{R_\mu \Lambda = 0}}{=} \; 2 c_A^\delta \nor{A^\delta R_\mu  \pa{\Gamma - \Lambda} \Gamma \bpa{\cP^\perp}^2 H \Lambda P_{\psi_\mu} \Lambda A^\delta}{2} \nonumber \\
	&\qquad \le 2 c_A^\delta  \nor{A^\delta R_\mu }{} \nor{ \Gamma-\Lambda}{} \nor{\Gamma \cP^\perp}{2} \nor{\cP^\perp H \Lambda}{} \nor{ \Lambda A}{}^\delta \nonumber \\
	&\qquad \le 2 \pa{c_A^2 \nor{ A\Lambda}{}}^\delta  \nor{A^\delta R_\mu}{} \nor{ \Gamma-\Lambda}{} \nor{\cP^\perp H \Lambda}{}  \nor{\cP^\perp\Gamma }{2,\delta} \nonumber\\
	&\qquad \le 2 \pa{c_A^2 \pa{1 + c_A \nor{ A\Lambda}{}}}^\delta  \nor{A^\delta R_\mu}{} \nor{ \Gamma-\Lambda}{2,\delta} \nor{\cP^\perp H \Lambda}{}  \nor{\cP^\perp\Gamma }{2,\delta}.
\end{align}
Finally, 
\begin{align*}
	&\nor{\cP^\perp \Gamma \Lambda^\perp \cP}{2,\delta} + \nor{\cP  \Lambda^\perp \Gamma\cP^\perp}{2,\delta}  \le 2 c_A^\delta \nor{A^\delta \cP^\perp \Gamma \pa{\Gamma - \Lambda} \Lambda^\perp \cP A^\delta}{2} \\
	&\qquad = 2 c_A^\delta \nor{A^\delta \cP^\perp \Gamma \pa{\Gamma - \Lambda} A^\delta A^{-\delta} \Lambda^\perp A^\delta A^{-\delta} \cP A^\delta}{2} \\
	&\qquad \le 2 c_A^\delta c_\cP^\delta \nor{\cP^\perp \Gamma}{2,\delta}\nor{\Gamma - \Lambda}{2,\delta} \nor{A^{-\delta} \Lambda^\perp A^\delta }{} \\
	&\qquad \le 2 \pa{c_A c_\cP \pa{1 + c_A\nor{A \Lambda}{}} }^\delta \nor{\cP^\perp \Gamma}{2,\delta} \nor{\Gamma - \Lambda}{2,\delta}.
\end{align*}

\subsection{First form}%
\label{sub:First form for}

Remark that in this form, gathering all the terms, we have
\begin{multline*}
\Gamma - \Lambda =\cP^\perp \Gamma \cP^\perp + \cP \Gamma \cP^\perp +\cP^\perp \Gamma \cP  + \Lambda \pa{\Gamma - \Lambda} \Lambda + \cP \Lambda^\perp \Gamma \Lambda^\perp \cP \\
 +\sum_{\mu=1}^\nu  \pa{P_{\psi_\mu} \bpa{\Gamma \cP^\perp H - H \cP^\perp \Gamma } R_\mu  -  R_\mu \bpa{\Gamma \cP^\perp H - H \cP^\perp \Gamma }  P_{\psi_\mu}}.
 \end{multline*}
 Now using~\eqref{eq:association} to associate $\cP^\perp \Gamma \cP$ with $\cP \Lambda \Gamma \Lambda$, we obtain~\eqref{eq:explicit_diff} where
\begin{multline}\label{eq:omega_def} 
\Omega = \cP^\perp \pa{\Gamma -\Lambda}^2 \cP^\perp + \Lambda \pa{\Gamma - \Lambda} \Lambda + \pa{\cP^\perp \pa{\Gamma - \Lambda}^2 \Lambda^\perp \cP  + \text{adj.}} \\
+ \cP \Lambda^\perp \pa{\Gamma - \Lambda}^2 \Lambda^\perp \cP - \sum_{\mu=1}^\nu  \pa{R_\mu \pa{\Gamma -\Lambda}^2 \cP^\perp H P_{\psi_\mu}  + \textup{adj.}}.
\end{multline}
From~\eqref{eq:ineq_lambda_gamma_lambda} we know that $\Lambda \pa{\Gamma - \Lambda} \Lambda$ is quadratic in $\Gamma-\Lambda$. Hence, we immediately see with the form~\eqref{eq:explicit_diff} that when $\Gamma - \Lambda$ is small, the leading term is $\sum_{\mu=1}^\nu  \pa{(1 + H R_\mu )\cP^\perp \Gamma P_{\psi_\mu}+ \text{adj.}}$, and $\Omega$ is quadratic in $\Gamma - \Lambda$, and thus much smaller.

We obtain~\eqref{eq:bound_Omega} from the developed inequalities.

\section{Proof of Proposition~\ref{prop:non-deg_main_thm}}
\label{sec:Proof of Theorem non deg}
We now treat the vector case and aim at showing~\eqref{eq:equality_diff} and~\eqref{eq:equality_E_diff}.

\subsection{Equality on eigenvectors}%
\label{sub:Equality}
Let us keep $\nu \in \Np$ general first, we will assume $\nu = 1$ later.

\begin{lemma}
	Given the setting of Proposition~\ref{prop:non-deg_main_thm}, for any $\mu \in \{1,\dots,\nu\}$, assuming that $\ps{\phi_\mu,\psi_\mu} \in \R$, there holds
\begin{multline}\label{eq:lema_comp} 
	\bpa{\Lambda^\perp + P_{\psi_\mu}} \pa{\phi_\mu - \psi_\mu} = \pa{1 + R_\mu H} \cP^\perp \phi_\mu \\
	- \f 12 \nor{\phi_\mu - \psi_\mu}{}^2 \psi_\mu + \pa{\cE_\mu - E_\mu} R_\mu \pa{\phi_\mu - \psi_\mu}.
\end{multline}
\end{lemma}
The remaining component of $\phi_\mu - \psi_\mu$ which is not taken into account in this lemma is $\Lambda P_{\psi_\mu}^\perp (\phi_\mu - \psi_\mu)$.
\begin{proof}
We have $\pa{H - E_\mu} \phi_\mu = 0$ and $\cP \pa{H - \cE_\mu} \psi_\mu = 0$, thus
\begin{align}\label{eq:funda} 
\cP \pa{H - \cE_\mu} \pa{\phi_\mu - \psi_\mu} = \pa{E_\mu - \cE_\mu} \cP \phi_\mu.
\end{align}
We first use $[\cP H \cP, \Lambda] = 0$, hence $\cP H \cP \Lambda^\perp = \Lambda^\perp \cP H \cP$ and applying $\Lambda^\perp$ on the left we obtain $\Lambda^\perp \cP H \cP \Lambda^\perp = \Lambda^\perp \cP H \cP$, so
\begin{align*}
&\cP \Lambda^\perp \pa{\cE_\mu - H} \Lambda^\perp \cP \pa{\phi_\mu - \psi_\mu}  =   \cP \Lambda^\perp\pa{\cE_\mu - H}  \cP \pa{\phi_\mu - \psi_\mu} \\
&\qquad = \Lambda^\perp \cP \pa{\cE_\mu - H} \pa{\phi_\mu - \psi_\mu}- \Lambda^\perp \cP \pa{\cE_\mu - H} \cP^\perp \pa{\phi_\mu - \psi_\mu} \\
&\qquad \underset{\substack{\eqref{eq:funda}}}{=} \; \pa{\cE_\mu - E_\mu} \Lambda^\perp\cP \phi_\mu +  \Lambda^\perp \cP H \cP^\perp \phi_\mu.
\end{align*}
We apply $R_\mu$ and use~\eqref{eq:KjHE}, $R_\mu \Lambda^\perp \cP = R_\mu$ and $R_\mu \psi_\mu = 0$ to obtain
\begin{align}\label{eq:phi_m_psi} 
	\cP \Lambda^\perp \pa{\phi_\mu - \psi_\mu} = R_\mu H \cP^\perp \phi_\mu +  \pa{\cE_\mu - E_\mu} R_\mu \pa{\phi_\mu  - \psi_\mu}.
\end{align}
Moreover, in a gauge where $\ps{\psi_\mu, \phi_\mu} \in \R$,
\begin{align*}
\ps{\psi_\mu, \phi_\mu} = 1 - \f 12 \nor{\phi_\mu - \psi_\mu}{}^2
\end{align*}
hence
\begin{align}\label{eq:proj_psi} 
	P_{\psi_\mu} \pa{\phi_\mu - \psi_\mu} = \bpa{\ps{\psi_\mu, \phi_\mu} - 1} \psi_\mu  = - \f 12 \nor{\phi_\mu - \psi_\mu}{}^2 \psi_\mu.
\end{align}
Finally, $\Lambda^\perp + P_{\psi_\mu} = \cP^\perp + \cP \Lambda^\perp + P_{\psi_\mu}$ and we obtain~\eqref{eq:lema_comp} by summing~\eqref{eq:phi_m_psi} and~\eqref{eq:proj_psi} with $\cP^\perp \pa{\phi_\mu - \psi_\mu} = \cP^\perp \phi_\mu$.
\end{proof}

We obtain~\eqref{eq:equality_diff} by applying this lemma to $\nu = 1$, in which case $\Lambda^\perp + P_{\psi_\mu} = 1$. For $\nu \ge 2$, this methods with vectors does not enable to obtain a bound on the remaining component $\Lambda  P_{\psi_\mu}^\perp$, that is why the previous density matrix approach is useful.

\subsection{Equality on eigenvalues}%
\label{sub:Equality on eigenvalues}

Let us first present a well-known and basic estimate showing that the errors between eigenvalues can be expressed as the square of the error between eigenvectors. We give a proof for the sake of completeness.

\begin{lemma}[Eigenvalue error is quadratic in eigenvector error]\label{lem:err_eigenvals_err_eigenvects} 
Take two self-adjoint operators $A$ and $H$, assume that $ \nor{A^{-1}}{} <+\infty$ and $c_H := \nor{A^{-1} H A^{-1}}{} <+\infty$. Take $\p$ in the form domain of $H$ and $\phi$ in the domain of $H$, such that $H \phi = E \phi$, $\nor{\psi}{} = \nor{\phi}{} = 1$, and define $\cE := \ps{\psi, H \psi}$. Then
\begin{align}\label{eq:diff_errs}
E - \cE &= \ps{\phi -\psi , (E - H) \pa{ \phi- \psi } }, \\
\ab{E - \cE} &\le  \nor{A^{-1} (H-E) A^{-1}}{}\mymin{\theta \in [0,2\pi[} \nor{A(\phi - e^{i\theta} \psi)}{}^2.\label{eq:diff_errs_abs}
\end{align}
\end{lemma}
Usually the bound~\eqref{eq:diff_errs_abs} is used as 
\begin{align*}
\ab{E - \cE} \le \pa{\nor{A^{-1}}{} ^2 \ab{E} + \nor{A^{-1} H A^{-1}}{}}\mymin{\theta \in [0,2\pi[} \nor{A(\phi - e^{i\theta} \psi)}{}^2.
\end{align*}
\begin{proof}
By using $(E-H) \phi = 0$ and $\nor{\psi}{}=1$, we have
 \begin{multline*}
	 \ps{\phi -\psi , (E-H) \pa{ \phi- \psi } } = -\ps{\phi -\psi , (E-H) \psi  }\\
	 = -\ps{(E-H) \pa{\phi -\psi} ,  \psi  } = \ps{(E-H) \psi ,  \psi  } = E -\ps{\psi , H \psi } = E - \cE.
\end{multline*}
Then
\begin{align*}
	\ab{E - \cE} &= \ab{\ps{A(\phi -\psi) , A^{-1} (E - H) A^{-1} A \pa{ \phi- \psi } }} \\
&\le \nor{A^{-1} (E-H) A^{-1}}{} \nor{A(\phi - \psi)}{}^2.
\end{align*}
To conclude, we change $\p \rightarrow e^{i\theta}\p$.
\end{proof}

We now have $\nu = 1$ and remove the subscript $1$ everywhere. Let us now show~\eqref{eq:equality_E_diff}. First,
\begin{align}\label{eq:interm_eq} 
	R \pa{H-\cE}\pa{1 + R H} \cP^\perp &= R (H-\cE) \cP^\perp + R \pa{H- \cE} R H \cP^\perp \nonumber\\
					   &\underset{\substack{R \cP^\perp = 0}}{=} \; R H \cP^\perp + R \pa{H- \cE} \cP R H \cP^\perp \nonumber\\
					   &\underset{\substack{R (\cE - H) \cP = \cP P_{\psi}^\perp}}{=} \; R H \cP^\perp - \cP P_{\psi}^\perp R H \cP^\perp \underset{\substack{\cP P_{\psi}^\perp R = R}}{=} \; 0.
\end{align}
Moreover, using~\eqref{eq:equality_diff} we have
\begin{align}\label{eq:intermios} 
&\pa{\cE - H} \pa{\phi - \psi} \underset{\substack{(\cP H-\cE) \psi = 0}}{=} \; \pa{\cE - H} \pa{\pa{1 + R H} \cP^\perp \phi + \pa{\cE - E} R \phi} \nonumber \\
&\qquad \qquad \qquad \qquad \qquad \qquad\qquad \qquad \qquad  + \tfrac 12 \nor{\phi - \psi}{}^2 \cP^\perp H \psi \nonumber \\
&= \pa{\cE - H} \pa{1 + R H} \cP^\perp \phi + \pa{\cE - E} \cP \pa{\cE - H}  R \phi \nonumber \\
& \qquad \qquad \qquad \qquad \qquad  + \pa{\cE - E}\cP^\perp \pa{\cE - H}  R \phi + \tfrac 12 \nor{\phi - \psi}{}^2 \cP^\perp H \psi \nonumber \\
&= \pa{\cE - H} \pa{1 + R H} \cP^\perp \phi + \pa{\cE -E} \cP P_{\psi}^\perp \phi + \pa{E - \cE}\cP^\perp H R \phi \nonumber \\
&\qquad \qquad \qquad \qquad \qquad \qquad\qquad \qquad \qquad  + \tfrac 12 \nor{\phi - \psi}{}^2 \cP^\perp H \psi.
\end{align}
Similarly as in~\eqref{eq:diff_errs}, using $(E-H) \phi = 0$ and $\nor{\psi}{}=1 $, we have
\begin{align}\label{eq:two_parts} 
	&\qquad E - \cE = \ps{\phi - \psi, (E - H) \pa{\phi - \psi}} \nonumber \\
	&= \ps{\phi - \psi, (\cE - H) \pa{\phi - \psi}} + (E-\cE) \nor{\phi - \psi}{}^2  \nonumber \\
	&\underset{\substack{\eqref{eq:equality_diff}}}{=} \; \ps{ \pa{1 + R H} \cP^\perp \phi, \pa{\cE - H} \pa{\phi - \psi}} - \tfrac 12 \nor{\phi - \psi}{}^2 \ps{ \psi , \pa{\cE - H} \pa{\phi - \psi}} \nonumber \\
	&\qquad + \pa{\cE - E} \ps{ R \phi , \pa{\cE - H} \pa{\phi - \psi}}+ (E-\cE) \nor{\phi - \psi}{}^2.
\end{align}
We now compute each of those terms. First, by~\eqref{eq:intermios} we have
\begin{multline*}
\ps{\pa{1 + R H} \cP^\perp \phi , (\cE - H) (\phi - \psi)} \\
=\ps{\pa{1 + R H} \cP^\perp \phi , \pa{\cE - H} \pa{1 + R H} \cP^\perp \phi} \\
+ (\cE - E) \ps{\pa{1 + R H} \cP^\perp \phi , \cP P_{\psi}^\perp \phi} + (E - \cE) \ps{\pa{1 + R H} \cP^\perp \phi, \cP^\perp H R \phi} \\
+\tfrac 12 \nor{\phi - \psi}{}^2 \ps{\pa{1 + R H} \cP^\perp \phi , \cP^\perp H \psi}.
\end{multline*}
Then using $R \cP^\perp = 0$ and~\eqref{eq:interm_eq}, we get 
\begin{multline*}
\ps{\pa{1 + R H} \cP^\perp \phi , (\cE - H) (\phi - \psi)} =\ps{ \cP^\perp \phi , \pa{\cE - H} \pa{1 + R H} \cP^\perp \phi} \\
+ (\cE - E) \ps{R H \cP^\perp \phi , \cP P_{\psi}^\perp \phi} + (E - \cE) \ps{ \cP^\perp \phi, \cP^\perp H R \phi} \\
+\tfrac 12 \nor{\phi - \psi}{}^2 \ps{\cP^\perp \phi , \cP^\perp H \psi} \\
\underset{\substack{P_{\psi}^\perp \cP R = R}}{=} \;\ps{ \cP^\perp \phi , \pa{\cE - H} \pa{1 + R H} \cP^\perp \phi} +\tfrac 12 \nor{\phi - \psi}{}^2 \ps{\cP^\perp \phi , H \psi},
\end{multline*}
giving the first term of~\eqref{eq:two_parts}. Using~\eqref{eq:intermios}, the second term of~\eqref{eq:two_parts} comes from 
\begin{multline*}
	\ps{\psi, (\cE - H) (\phi - \psi)} =   \ps{\cP\psi, \pa{\cE - H} \pa{1 +RH} \cP^\perp \phi} \\
= \ps{\cP\psi , (\cE - H) \cP^\perp \phi} + \ps{\cP\psi, (\cE - H) R H \cP^\perp \phi} \underset{\substack{\cP (\cE - H) R \\= P_{\psi}^\perp \cP}}{=} \;-\ps{H \psi, \cP^\perp \phi}.
\end{multline*}
 The third term of~\eqref{eq:two_parts} comes from
\begin{multline*}
	\ps{R \phi , \pa{\cE- H} \pa{\phi - \psi}} \underset{\substack{\eqref{eq:interm_eq},\eqref{eq:intermios} \\ R \cP^\perp = 0}}{=} \; \ps{R \phi , \pa{\cE - E} \cP P_{\psi}^\perp \phi} \\
						  \underset{\substack{R \cP P_{\psi}^\perp = R}}{=} \; \pa{\cE - E} \ps{\phi, R \phi} \underset{\substack{R \psi = 0}}{=} \; \pa{\cE - E} \ps{\phi -\psi, R \pa{\phi - \psi}}.
\end{multline*}
Summing all the terms of~\eqref{eq:two_parts} yields
\begin{multline*}
	E - \cE = \ps{\cP^\perp \phi, \pa{\cE - H} \pa{1 + R H} \cP^\perp \phi} + \pa{E-\cE}^2 \ps{\phi - \psi, R \pa{\phi - \psi}}  \nonumber\\
		+ (E-\cE) \nor{\phi - \psi}{}^2 + \nor{\phi - \psi}{}^2 \re \ps{\cP^\perp \phi, H \psi}.
\end{multline*}
Moreover, $\bpa{(\cE - H)(1+ RH)}^* = (1+HR)(\cE-H) = (\cE - H)(1+ RH)$ is self-adjoint so $\ps{\cP^\perp \phi, \pa{\cE - H} \pa{1 + R H} \cP^\perp \phi} \in \R$, $\ps{\phi - \psi, R \pa{\phi - \psi}} \in \R$. To conclude, we use that $\ps{\cP^\perp \phi, H \psi} = -\ps{\cP^\perp \phi, (H-E) (\phi - \psi)}$.

\subsection{Inequalities~\eqref{eq:ineq_cons} and~\eqref{eq:proof_E_cE_bounded}}%
\label{sub:Inequalities}

From~\eqref{eq:equality_diff} we have
\begin{align*}
	&\nor{\phi - \psi}{e} \le \nor{A \pa{1 + R H}\cP^\perp \phi}{} + \f 12 \nor{\phi - \psi}{}^2 \nor{\psi}{e} + \ab{\cE - E} \nor{AR \pa{\phi - \psi}}{} \\
	& \le \nor{\pa{1 + A R A A^{-1} H A^{-1} } A \cP^\perp \phi}{} + \f 12 c_A\nor{\phi - \psi}{}  \nor{\phi - \psi}{e} \nor{\psi}{e} \\
	&\qquad \qquad \qquad \qquad + \nor{A R}{} \ab{\cE - E}  \nor{\phi - \psi}{} \\
	& \le \pa{1 + \nor{ARA}{} c_H}\nor{\cP^\perp \phi}{e} + c_A \pa{\f 12 \nor{\phi - \psi}{} \nor{\psi}{e} + \nor{A R}{} \ab{\cE - E} }\nor{\phi - \psi}{e},
\end{align*}
and we obtain~\eqref{eq:ineq_cons} when
\begin{align*}
c_A \pa{\f 12 \nor{\phi - \psi}{} \nor{\psi}{e} + \nor{A R}{} \ab{\cE - E} }\le \f 12.
\end{align*}
Proving~\eqref{eq:proof_E_cE_bounded} uses~\eqref{eq:diff_errs}.

We can obtain~\eqref{eq:bbb} with the same method, which also needs to use~\eqref{eq:minus_one_last}.


\section{Bounds on the Rayleigh-Schrödinger series\\in perturbation theory}%
\label{sec:Bounds on the Rayleigh-Schrödinger series in perturbation theory}

In the proofs of the theorems of Section~\ref{sec:Application to RBM+PT}, we will need some general results about perturbation theory, which we show here. The main results are Lemma~\ref{lem:bound_phi_n} and Lemma~\ref{lem:bound_phi_n_deg} on the boundedness of the Rayleigh-Schrödinger series $E^n_\mu$ and $\phi^n_\mu$. 

We take the context of an analytic and self-adjoint operator family $H(\lambda)$, presented in Section~\ref{sec:definitions_rbm_pt}. In particular, we consider a series of operators
\begin{align*}
H(\lambda) = \sum_{n=0}^{+\infty} \lambda^n H^n,
\end{align*}
and a cluster of eigenmodes $\pa{E_\mu(\lambda), \phi_\mu(\lambda)}_{\mu=1}^\nu$ of $H(\lambda)$, where all those maps are analytic in $\lambda \in ]-\lambda_0,\lambda_0[$. We define respectively an energy norm and a parameter norm, for any operator $B$ of $\cH$, by respectively
\begin{align*}
\nor{B}{2,ee} := \nor{A B A}{2}, \qquad \qquad \nor{B}{p} := \nor{A^{-1} B A^{-1}}{}.
\end{align*}
Moreover, for any $B$ we have $\nor{B}{2,1} \le \nor{A^{-1}}{} \nor{B}{2,ee}$ where the energy norm $\nor{\cdot}{2,1} $ was defined in~\eqref{eq:energy_norm}. We use intermediate normalization, which is reviewed in Section~\ref{sec:interm_normalization}. We set
\begin{multline*} 
\Phi_\mu(\lambda) := \f{\phi_\mu(\lambda)}{\ps{\phi_\mu^0,\phi_\mu(\lambda)}}, \qquad \phi_\mu^n := \f{1}{n!}  \pa{\f{\d^n }{\d \lambda^n} \phi_\mu(\lambda) }_{\mkern 1mu \vrule height 2ex\mkern2mu \lambda = 0}, \\
\Phi_\mu^n := \f{1}{n!}  \pa{\f{\d^n }{\d \lambda^n} \Phi_\mu(\lambda) }_{\mkern 1mu \vrule height 2ex\mkern2mu \lambda = 0}, \qquad E_\mu^n := \f{1}{n!}  \pa{\f{\d^n }{\d \lambda^n} E_\mu(\lambda) }_{\mkern 1mu \vrule height 2ex\mkern2mu \lambda = 0}.
\end{multline*}

\subsection{A note on analyticity}%
\label{sub:A note on analyticity}

Let us make a comment on analyticity in Section~\ref{sub:Analytic family of operators}. Since $H^n \pa{H^0  - r}^{-1}$ is bounded, then $D(H^0) \subset D(H^n)$ and $H^n$ is $H^0$-bounded, so by the same proof as in~\cite[Lemma p16]{ReeSim4}, $H(\lambda)$ is an analytic family of type (A) (in the sense of definition~\cite[Definition p16]{ReeSim4}). Also using a similar proof as in~\cite[Theorem XII.9 p17]{ReeSim4}, $H(\lambda)$ is an analytic family in the sense of Kato. It is thus possible to apply~\cite[Theorem XII.13]{ReeSim4}.

\subsection{A preliminary bound on ``Cauchy squares''}%
\label{sub:A preliminary bound on the Rayleigh-Schrödinger series}

First we will need the following result, which is a bound on a series that we can call the ``Cauchy square'' series.

\begin{lemma}[Upper bound on the Cauchy square series]\label{lem:cauchy_square} 
Take $\alpha, \beta >0$ and let us define $x_1 := \alpha$ and for any $n \in \Np$, $n\ge 2$,
\begin{align*}
x_n := \beta \sum_{s=1}^{n-1} x_{n-s} x_s.
\end{align*}
Then for any $n \in \Np$,
\begin{align}\label{eq:bound_x_n} 
x_n \le \alpha  \f{\pa{2 \zeta\pa{\tfrac{3}{2}} \alpha \beta}^{n-1}}{n^{\f 32}}.
\end{align}
where $\zeta$ is Riemann's zeta function so $2 \zeta\pa{\tfrac{3}{2}} \simeq 5.2248...$
\end{lemma}
\begin{remark}
We made a numerical study giving evidence that 
\begin{align*}
	\f{x_n}{\alpha \pi^{-\f 12} (4 \alpha \beta)^{n-1} n^{-\f 32}} \underset{\substack{n \rightarrow +\infty}}{\longrightarrow} \; 1.
\end{align*}
\end{remark}

\begin{proof}
First, we show that proving the result with $\beta = 1$ enables to show it for any $\beta >0$. Take a general $\beta > 0$. By using $y_n := \beta x_n$, we have $y_1 = \alpha \beta$ and $y_n = \sum_{s=1}^{n-1} y_{n-s} y_s$. We use the result for $\beta = 1$ on $y_n$, which yields the claimed result for $x_n$.

Hence without loss of generality we can take $\beta = 1$. Let us prove~\eqref{eq:bound_x_n} by induction. We define $\xi := 2 \zeta\pa{\tfrac{3}{2}}$, for any $x \in ]0,n[$ we define $g(x) := \pa{x(n-x)}^{-\f 32}$, we extend it on $\C \backslash \{0,n\}$, and we define $S_n := \sum_{s=1}^{n-1} g(s)$ for any $n \ge 2$.

For $n=1$ the right hand side of~\eqref{eq:bound_x_n} is $\alpha$ so the initial step is valid. Take $n \in \Np$, $n \ge 2$, such that for any $s \in \{1,\dots,n-1\}$, $x_s \le \alpha (\xi\alpha)^{s-1} s^{-\f 32}$. Then
\begin{align}\label{eq:bound_xn_interm} 
	x_n \le  \alpha^2 \pa{\xi\alpha}^{n-2} \sum_{s=1}^{n-1} \f{1}{\pa{s(n-s)}^{\f 32}} = \alpha^2 \pa{\xi\alpha}^{n-2} S_n.
\end{align}
Defining $G(x) := \f{n-2x}{\sqrt{x(n-x)}}$ we have $G'(x) = -\f{n^2}{2} g(x)$ so
\begin{align*}
	\int_1^{n-1} g = \f{2}{n^2} \pa{G(1) - G(n-1)} = \f{4 (n-2)}{n^2\sqrt{n-1}}.
\end{align*}
Moreover, $g(\overline{z}) = \overline{g(z)}$ and by the Abel-Plana formula,
\begin{align*}
	S_n&= \int_1^{n-1}g(s) \d s+\tfrac{1}{2}g(1)+\tfrac{1}{2}g(n-1) \\
	   &\qquad \qquad - 2 \im \int_0^\infty \f{g(1+i y)-g(n-1+i y)}{e^{2 \pi  y}-1} \d y \\
	   &=\f{5n^2-12n+8}{(n-1)^{\f 32} n^2}-4 \im \int_0^\infty \f{g(1+i y)}{e^{2 \pi  y}-1} \d y
\end{align*}
hence
\begin{align*}
\mylim{n \rightarrow +\infty} n^{\f 32} S_n = 5 -4 \im \int_0^\infty \f{(1+iy)^{\f 32}}{e^{2 \pi  y}-1} \d y = 2 \zeta\pa{\tfrac{3}{2}} = \xi,
\end{align*}
where we used the Abel-Plana formula again. Moreover, $(n^{3/2} S_n)_{n \ge 2}$ is an increasing sequence, thus $S_n \le \xi/n^{3/2}$. Then~\eqref{eq:bound_xn_interm} enables to conclude.
\end{proof}

\subsection{Bound on the Rayleigh-Schrödinger series : the non-degenerate case}%
\label{ssub:The non-degenerate case}

We are now ready to obtain a bound on $E^n_\mu$ and $\Phi^n_\mu$. In particular, this provides a bound on the convergence radius of the perturbation series. We take the non-degenerate case, that is $\nu = 1$, so $\mu = 1$.

We recall that $K_\mu(0)$ is defined in~\eqref{def:pseudoinv_lambda_K}. For any $m \in \Nz$ and any $n \in \Np$, we define
\begin{align}\label{eq:rec_Q} 
h^m_\mu := H^m-E^m_\mu, \qquad \qquad Q^{n}_\mu := h^n_\mu + \sum_{s=1}^{n-1} h^{n-s}_\mu K_\mu(0) Q^s_\mu.
\end{align}
Then, by a classical result which can be found in~\cite{Kato,Hir69} for instance, we have
\begin{align}\label{eq:Phi_n} 
\Phi_\mu^n = K_\mu(0) Q^n_\mu \Phi_\mu^0, \qquad \forall n \ge 1.
\end{align}
Moreover, we can compute 
\begin{align}\label{eq:expr_E_n} 
E^n_\mu = \ps{\Phi_\mu^0, \pa{Q^n_\mu + E^n_\mu} \Phi_\mu^0}.
\end{align}

\begin{lemma}[Bound on $E^n$, $\phi^n$ and $\Phi^n$, non-degenerate case]\label{lem:bound_phi_n} 
Let us consider the Hamiltonian family $H(\lambda) = \sum_{n=0}^{+\infty} \lambda^n H^n$ under the assumptions of Sections~\ref{sub:Analytic family of operators} and~\ref{sub:choose_set_eigvals}, with $\nu = 1$. The non-degenerate eigenmode is denoted by $(E_\mu(\lambda),\phi_\mu(\lambda))$, we fix the phasis of $\phi_\mu(\lambda)$ such that $\ps{\phi_\mu^0,\phi_\mu(\lambda)} \in \R_+$, the intermediate normalization eigenvector is $\Phi_\mu(\lambda) := \f{\phi_\mu(\lambda)}{\ps{\phi_\mu^0,\phi_\mu(\lambda)}}$ and the Taylor series are written
\begin{align*}
E_\mu^n := \tfrac{1}{n!}  \pa{\tfrac{\d^n }{\d \lambda^n} E_\mu(\lambda) }_{\mkern 1mu \vrule height 2ex\mkern2mu \lambda = 0},
\phi_\mu^n := \tfrac{1}{n!}  \pa{\tfrac{\d^n }{\d \lambda^n} \phi_\mu(\lambda) }_{\mkern 1mu \vrule height 2ex\mkern2mu \lambda = 0},
\Phi_\mu^n := \tfrac{1}{n!}  \pa{\tfrac{\d^n }{\d \lambda^n} \Phi_\mu(\lambda) }_{\mkern 1mu \vrule height 2ex\mkern2mu \lambda = 0}.
 \end{align*}
Then for any $n \in \Nz$,
\begin{align}\label{eq:bound_Q}
\ab{E^n_\mu} + \nor{\phi_\mu^n}{e} + \nor{\Phi_\mu^n}{e} &\le a b^{n},
\end{align}
where $a,b \in \R_+$ are independent of $n$.
\end{lemma}
\begin{proof}
For clarity, we drop the subscripts 1, so $E := E_1$, $\phi := \phi_1$, $\Phi := \Phi_1$, $K := K_1(0)$, $Q^n := Q_1^n$, $h^n := h_1^n$ for any $n \in \Nz$. We define
\begin{align}\label{eq:def_cH_cK} 
c_{H,\infty} :=\mysup{n \in \Nz} \nor{A^{-1} H^n A^{-1}}{}, \qquad c_K := \nor{A K A}{}, \qquad c_A := \nor{A^{-1}}{} .
\end{align}
	For any $n \in \Np$ let us define
\begin{align}\label{def:q_n} 
q^n := H^n + \sum_{s=1}^{n-1} h^{n-s}K Q^s.
\end{align}
From~\eqref{eq:rec_Q} we have $Q^n = q^n - E^n$, from~\eqref{eq:expr_E_n} we have $E^n = \ps{\Phi^0, q^n \Phi^0}$, and we recall that $h^n = H^n - E^n$ hence 
\begin{align}\label{eq:_bounds_En_Qn_with_qn} 
	\ab{E^n} &= \ab{\ps{A \Phi^0, A^{-1} q^n A^{-1} A \Phi^0}} \le \nor{q^n}{p} \nor{\Phi^0}{e}^2,\nonumber \\
	\nor{Q^n}{p} &\le \nor{q^n}{p} + c_A^2 \ab{E^n} \le  \nor{q^n}{p}\pa{1 + c_A^2 \nor{\Phi^0}{e}^2}, \nonumber\\
	\nor{h^{n}}{p}& \le c_{H,\infty} +  c_A^2 \nor{q^n}{p} \nor{\Phi^0}{e}^2.
\end{align}
Thus from~\eqref{def:q_n} we have
\begin{align*}
\nor{q^n}{p} &\le c_{H,\infty} + c_K \sum_{s=1}^{n-1} \nor{h^{n-s}}{p} \nor{Q^s}{p} \\
& \le c_{H,\infty} + c_K \pa{1 + c_A^2 \nor{\Phi^0}{e}^2} \sum_{s=1}^{n-1} \pa{c_{H,\infty} +  \nor{q^{n-s}}{p} c_A^2 \nor{\Phi^0}{e}^2} \nor{q^s}{p} \\
& \le c_{H,\infty} +  \beta \sum_{s=1}^{n-1} \pa{1 +  \nor{q^{n-s}}{p} } \nor{q^s}{p} 
\end{align*}
where
\begin{align*}
\beta := c_K \pa{1 + c_A^2 \nor{\Phi^0}{e}^2} \max\pa{c_{H,\infty} ,c_A^2 \nor{\Phi^0}{e}^2}.
\end{align*}
Defining $y_n := \nor{q^n}{p} +1$, we have
\begin{align}\label{eq:bound_yn} 
	y_n & \le 1 + c_{H,\infty} + \beta \sum_{s=1}^{n-1}  y_{n-s} (y_s -1) \le (1+c_{H,\infty} + \beta) \sum_{s=1}^{n-1} y_{n-s} y_s.
\end{align}
By~\eqref{eq:Phi_n}, we have 
\begin{align}\label{eq:bphI} 
	\nor{\Phi^n}{e} \le c_K \nor{Q^n}{p} \nor{\Phi^0}{e} \le c_K \pa{1 + c_A^2 \nor{\Phi^0}{e}^2} \nor{\Phi^0}{e}y_n.
\end{align}

We now show the bound for $\nor{\phi^n}{e}$. For any $n \in \Nz$ we define  $Y^n$ and $X^n$ as in~\eqref{eq:Yn}, and
\begin{align*}
u_n := \max\pa{ \ab{Y^n} , \ab{X^{n-1}} , y_n }.
\end{align*}
We have 
\begin{align*}
	\ab{Y^n} &\le \f 12 \sum_{k=1}^{n-1} \pa{ \ab{\ps{\Phi^{n-k}, \Phi^k}} + \ab{ Y^{n-k}} \ab{ Y^k}} \\
	&\le \f 12 \sum_{k=1}^{n-1}  \pa{c_A^2 \nor{ \Phi^{n-k}}{e} \nor{\Phi^k}{e}  + u_{n-k} u_k}\\
	&\underset{\substack{\eqref{eq:bphI}}}{\le} \; g \sum_{k=1}^{n-1} u_{n-k} u_k,
\end{align*}
where
\begin{align*}
	g := \max \pa{1, c_A c_K \pa{1 + c_A^2 \nor{\Phi^0}{e}^2} \nor{\Phi^0}{e}}^2.
\end{align*}
Moreover,
\begin{align*}
X^{n-1}=  -\sum_{s=0}^{n-3} X^s Y^{n-1-s}  = - \sum_{s=1}^{n-2} X^{s-1} Y^{n-s} \quad \text{so} \quad \ab{X^{n-1}} \le \sum_{s=1}^{n-2} u_{n-s} u_s.
 \end{align*}
 We deduce that 
 \begin{align*}
u_n \le (1+c_{H,\infty} + \beta + g) \sum_{s=1}^{n-1} u_{n-s} u_s.
 \end{align*}
Using Lemma~\ref{lem:cauchy_square}, we deduce that there are $a,b > 0$ such that $u_n \le a b^n$ for any $n \in \Nz$. We can propagate this result for $\nor{q^n}{p}$, $\ab{E^n} $, $\nor{Q^n}{p}$ using~\eqref{eq:_bounds_En_Qn_with_qn}, for $\nor{\Phi^n}{e}$ by using~\eqref{eq:bphI}, and for $\nor{\phi^n}{e}$ by using~\eqref{eq:recover_phi}.
\end{proof}

\begin{remark}[On the radius of convergence of the perturbative approximation]
Defining the perturbation approximation in the intermediate normalization $\vp(\lambda) := \sum_{n=0}^{\ell} \lambda^n \phi_\mu^n$, \eqref{eq:bound_Q} yields
\begin{align*}
\nor{\vp(\lambda)}{e} \le a \sum_{n=0}^{\ell} \pa{\ab{\lambda}b}^n = a \f{1 - \pa{\ab{\lambda}b}^{\ell +1}}{1 - \ab{\lambda}b} \le  \f{a}{1 - \ab{\lambda}b},
\end{align*}
and the radius of convergence of the right-hand side is $b^{-1}$. Moreover, 
\begin{align}\label{eq:diff_pert} 
	\nor{\phi_\mu(\lambda) - \vp(\lambda) }{e} \le \f{a}{1 - \ab{\lambda}b} \pa{\ab{\lambda} b}^{\ell +1}.
\end{align}
\end{remark}

\subsection{Bound on the Rayleigh-Schrödinger series : the degenerate case, when degeneracy is lifted at first order}%

To show Theorem~\ref{thm:deg_vecs}, we will need a similar lemma as Lemma~\ref{lem:bound_phi_n} but for the degenerate case. We consider a degenerate case, so for all $\alpha,\beta \in \{1,\dots,\nu\}$, $E^0_\alpha = E^0_\beta$. As before, $\Gamma(0)$ is the orthogonal projector onto $\oplus_{\mu=1}^\nu \Ker \pa{H^0 - E^0_\mu}$. It is well-known that the $\nu$ eigenvalues of $\pa{\Gamma(0) H^0 \Gamma(0)}_{\mkern 1mu \vrule height 2ex\mkern2mu \Gamma(0) \cH \rightarrow \Gamma(0) \cH}$, as an operator of $\Gamma(0) \cH$, are
\begin{align*}
E_\mu^1 =  \pa{\tfrac{\d }{\d \lambda} E_\mu(\lambda) }_{\mkern 1mu \vrule height 2ex\mkern2mu \lambda = 0} = \ps{\phi^0_\mu , H^1 \phi^0_\mu} .
\end{align*}
Here we assume that degeneracy is lifted at first order for some $\mu \in \{1,\dots,\nu\}$, meaning that for any $\alpha \in \{1,\dots,\nu\} \backslash \{\mu\}$, $E_\alpha^1 \neq E_\mu^1$.

\subsubsection{Pseudo-inverses}%
\label{ssub:Pseudo-inverses}

In the degenerate case, we need to introduce two kinds of partial inverse operators.  The ``zeroth order'' partial inverse $K_\mu(0)$ was defined in~\eqref{def:pseudoinv_lambda_K}, and we set 
\begin{align*}
K^0_\mu := K_\mu(0).
\end{align*}
 We take some $\mu \in \{1,\dots,\nu\}$ and assume that degeneracy is lifted at first order.  We also set 
\begin{align*}
K^1_\mu := \rr_\mu (0),
\end{align*}
where $\rr_\mu (0)$ was defined in~\eqref{eq:def_rr}.

\subsubsection{Series}%
\label{ssub:Series}

We present degenerate perturbation theory as in the work of Hirschfelder~\cite[Section 4.A]{Hir69}, in the case where all degeneracies are lifted at first order. We define $h^{n}_\mu := H^{n}-E^{n}_\mu$, the operators
\begin{align}\label{eq:q_s_degenerate} 
	q^n_{0,\mu} := H^n + \sum_{s=1}^{n-1} h^s_\mu K^0_\mu Q^{n-s}_{0,\mu}, \qquad q^n_{1,\mu} := q^n_{0,\mu} + \sum_{s=2}^{n-1} Q^s_{0,\mu} K^1_\mu Q^{n-s+1}_{1,\mu}
\end{align}
and $Q^n_{i,\mu} = q^n_{i,\mu} - E^n_\mu$ for $i \in \{0,1\}$. Then for any $m \in \Nz$ and any $n \in \Np$, we have
\begin{align}\label{eq:es_phis_degenerate} 
E^m_\mu = \ps{\phi^0_\mu, q^m_{1,\mu} \phi^0_\mu}, \qquad \Phi^n_\mu = \pa{K^0_\mu Q^n_{1,\mu} + K^1_\mu Q^{n+1}_{1,\mu}} \phi^0_\mu.
\end{align}
More precisely, in~\cite[Section 4.A]{Hir69}, the first term of $\Phi^n_\mu$ is not the same as in~\eqref{eq:es_phis_degenerate} but it can be showed that there is equality between them.

\begin{lemma}[Bound on $E^n_\mu$, $\phi^n_\mu$ and $\Phi^n_\mu$, degenerate case]\label{lem:bound_phi_n_deg} 
	Let us consider the Hamiltonian family $H(\lambda) = \sum_{n=0}^{+\infty} \lambda^n H^n$ under the assumptions of Sections~\ref{sub:Analytic family of operators} and~\ref{sub:choose_set_eigvals}. We take $\mu \in \{1,\dots,\nu\}$ and an eigenmode denoted by $(E_\mu(\lambda),\phi_\mu(\lambda))$. We consider the degenerate case, where the degeneracy is lifted at first order, as described in Section~\ref{sub:Vectors in the degenerate case}, i.e. $E^0_\mu = E^0_\alpha$ for all $\alpha \in \{1,\dots,\nu\} $ and $E^1_\mu \neq E^1_\alpha$ for all $\alpha \in \{1,\dots,\nu\} \backslash \{\mu\}$. We fix the phasis of $\phi_\mu(\lambda)$ such that $\ps{\phi_\mu^0,\phi_\mu(\lambda)} \in \R_+$, the intermediate normalization eigenvector is $\Phi_\mu(\lambda) := \f{\phi_\mu(\lambda)}{\ps{\phi_\mu^0,\phi_\mu(\lambda)}}$ and the Taylor series are written
\begin{align*}
E_\mu^n := \tfrac{1}{n!}  \pa{\tfrac{\d^n }{\d \lambda^n} E_\mu(\lambda) }_{\mkern 1mu \vrule height 2ex\mkern2mu \lambda = 0},
\phi_\mu^n := \tfrac{1}{n!}  \pa{\tfrac{\d^n }{\d \lambda^n} \phi_\mu(\lambda) }_{\mkern 1mu \vrule height 2ex\mkern2mu \lambda = 0},
\Phi_\mu^n := \tfrac{1}{n!}  \pa{\tfrac{\d^n }{\d \lambda^n} \Phi_\mu(\lambda) }_{\mkern 1mu \vrule height 2ex\mkern2mu \lambda = 0}.
 \end{align*}
Then for any $n \in \Nz$,
\begin{align}\label{eq:bound_Q_deg}
\ab{E^n_\mu} + \nor{\phi^n_\mu}{e} + \nor{\Phi^n_\mu}{e} \le a b^n,
\end{align}
where $a,b >0$ are independent of $n$ and $\mu$, and depend polynomially on $\nor{A K^0_\mu A}{}$.
\end{lemma}
\begin{proof}
We recall that $c_{H,\infty}$ was defined in~\eqref{eq:def_cH_cK}. From~\eqref{eq:es_phis_degenerate} we have
\begin{align}\label{eq:borne_En} 
	\ab{E^n_\mu} &\le \nor{q^n_{1,\mu}}{p} \nor{\Phi^0_\mu}{e}^2, \nonumber \\
	\nor{h^n_\mu}{p} &\le c_{H,\infty} + c_A^2 \ab{E^n_\mu} \le c_{H,\infty} + c_A^2 \nor{q^n_{1,\mu}}{p} \nor{\Phi^0_\mu}{e}^2,
\end{align}
and for $i \in \{0,1\}$, 
\begin{align}\label{eq:borne_Qni} 
\nor{Q^n_{i,\mu}}{p} \le \nor{q^n_{i,\mu}}{p} + c_A^2 \nor{q^n_{1,\mu}}{p} \nor{\Phi^0_\mu}{e}^2.
\end{align}
Next,
\begin{align*}
\nor{q^n_{0,\mu}}{p} &\le c_{H,\infty} + \nore{K^0_\mu} \sum_{s=1}^{n-1} \nor{h^s_{\mu}}{p} \nor{Q^{n-s}_{0,\mu}}{p} \\
&\le c_{H,\infty} + \nore{K^0_\mu} \sum_{s=1}^{n-1} \bpa{c_{H,\infty} + c_A^2 \nor{q^s_{1,\mu}}{p} \nor{\Phi^0_\mu}{e}^2 } \\
&\qquad \qquad \qquad \qquad \qquad \qquad \times \pa{\nor{q^{n-s}_{0,\mu}}{p} +c_A^2 \nor{q^{n-s}_{1,\mu}}{p} \nor{\Phi^0_\mu}{e}^2}.
\end{align*}
We define
\begin{align*}
	C := \max\pa{c_{H,\infty},\nore{K^0_\mu},\nore{K^1_\mu}, 1+ c_A^2 \nor{\Phi^0_\mu}{e}^2},
\end{align*}
and we have
\begin{align*}
\nor{q^n_{0,\mu}}{p} \le C + C^3 \sum_{s=1}^{n-1} \bpa{1 + \nor{q^s_{1,\mu}}{p}} \pa{\nor{q^{n-s}_{0,\mu}}{p} + \nor{q^{n-s}_{1,\mu}}{p}}.
\end{align*}
Moreover,
\begin{align*}
	&\nor{q^n_{1,\mu}}{p} \le \nor{q^n_{0,\mu}}{p}  + \nore{K^1_\mu} \sum_{s=2}^{n-1} \nor{Q^s_{0,\mu}}{p}\nor{Q^{n-s+1}_{1,\mu}}{p} \\
			    &\le \nor{q^n_{0,\mu}}{p} \\
			    & +\pa{1  +  c_A^2 \nor{\Phi^0_\mu}{e}^2}\nore{K^1_\mu} \sum_{s=2}^{n-1} \pa{\nor{q^s_{0,\mu}}{p}  +  c_A^2 \nor{q^{s}_{1,\mu}}{p} \nor{\Phi^0_\mu}{e}^2} \nor{q^{n-s+1}_{1,\mu}}{p} \\
			    &\le \nor{q^n_{0,\mu}}{p}  + C^3 \sum_{s=2}^{n-1} \pa{\nor{q^s_{0,\mu}}{p}  +  \nor{q^{s}_{1,\mu}}{p} } \nor{q^{n-s+1}_{1,\mu}}{p}.
\end{align*}
We define $x_n := \nor{q^n_{0,\mu}}{p} + \nor{q^n_{1,\mu}}{p} + 1$ and estimate
\begin{align*}
	x_n &\le 2C + 2C^3 \sum_{s=1}^{n-1} \bpa{1 + \nor{q^s_{1,\mu}}{p}} \pa{\nor{q^{n-s}_{0,\mu}}{p} + \nor{q^{n-s}_{1,\mu}}{p}} \\
	    &\qquad \qquad\qquad \qquad  + C^3 \sum_{s=2}^{n-1} \pa{\nor{q^s_{0,\mu}}{p}  +  \nor{q^{s}_{1,\mu}}{p} } \nor{q^{n-s+1}_{1,\mu}}{p} \\
	    &\le 2C\pa{1 + C^2} \pa{ \sum_{s=1}^{n-1} x_s x_{n-s} + \sum_{s=2}^{n-1} x_s x_{n-s+1}} \\
	      &=   2C(1+C^2) \pa{x_1 x_{n-1} + \sum_{s=1}^{n-2} \pa{x_s + x_{s+1}} x_{n-s}}\\
	      &\le   4C(1+C^2)  \sum_{s=1}^{n-2} \pa{x_s + x_{s+1}} x_{n-s}.
\end{align*}
Then with $y_n := x_n + x_{n+1}$, we have
\begin{align*}
	y_n &\le x_n + 4C(1+C^2) \sum_{s=1}^{n-1} \pa{x_s + x_{s+1}} x_{n+1-s}\\
	    &\le y_{n-1} + 4C(1+C^2) \sum_{s=1}^{n-1} y_s y_{n-s} \le (1 + 4C(1+C^2)) \sum_{s=1}^{n-1} y_s y_{n-s},
\end{align*}
where we used that $1/y_1 \le 1$ in the last inequality. Using Lemma~\ref{lem:cauchy_square}, we deduce that there are $a,b > 0$ such that for any $n \in \Nz$, $y_n \le a b^n$, and then $\nor{q^n_{i,\mu}}{p} \le ab^n$ for $i \in \{0,1\}$. We propagate this property for $\ab{E^n_\mu} $, $\nor{Q^n_{i,\mu}}{p} $ using~\eqref{eq:borne_En} and~\eqref{eq:borne_Qni} and for $\nor{\Phi^n_\mu}{e}$ by using
\begin{align*}
\nor{\Phi^n}{e} \le  \pa{ \nore{K^0_\mu} \nor{Q^n_{1,\mu}}{p} +  \nore{K^1_\mu} \nor{Q^{n+1}_{1,\mu}}{p} } \nor{\Phi^0_\mu}{e}.
\end{align*}
As we see in $C$, we can bound $a$ and $b$ using polynomials in 
\begin{align*}
 c_{H,\infty},\nore{K^0_\mu}, \nore{K^1_\mu}, \nor{\Phi^0_\mu}{e}, c_A.
\end{align*}
The bound on $\nor{\phi^n_\mu}{e}$ can be deduced with the same method as in Lemma~\ref{lem:bound_phi_n}.
\end{proof}

\section{Coefficients in density matrix perturbation theory}%
\label{sec:density matrix perturbation theory}

In this section we present how to compute the coefficients of the perturbative series in density matrix perturbation theory. We use the Liouvillian operator and its partial inverse, a classical tool in perturbation theory, see \cite{Kato,Teufel03}, used in~\cite{BacDerWoj18,MonTeu19}, with a detailed exposition in~\cite[Section 5.1]{CanKamLev21}. See also for instance \cite{TruDiaBow20}.

\subsection{Definitions}%
\label{sub:Definitions_DMPT}

We choose the same context and notations as in Section~\ref{sec:definitions_general}, in particular, we consider a series of operators
\begin{align*}
H(\lambda) = \sum_{n=0}^{+\infty} \lambda^n H^n.
\end{align*}
We consider a cluster of eigenmodes $\pa{E_\mu(\lambda), \phi_\mu(\lambda)}_{\mu=1}^\nu$ of $H(\lambda)$, where all those maps are analytic in $\lambda \in ]-\lambda_0,\lambda_0[$. Let us take $\lambda_0$ small enough so that there is $\kappa_H >0$, independent of $\lambda$, such that for any $\lambda \in ]-\lambda_0,\lambda_0[$,
\begin{align*}
\bpa{\sigma(H) \backslash \{E_\mu(\lambda)\}_{\mu=1}^\nu} \cap \bpa{\cup_{\mu=1}^\nu \; ]E_\mu(\lambda) - \kappa_{H}  , E_\mu(\lambda) + \kappa_{H} [ \;} = \varnothing.
\end{align*}
Take $(\vp_\mu(\lambda))_{\mu=1}^\nu \in \cH^\nu$ such that 
\begin{align}\label{eq:rel_frame} 
\ps{\vp_\mu(\lambda) , \vp_\alpha(\lambda)} = \delta_{\mu\alpha}, \qquad \vp_\mu(\lambda) \in \Ker \pa{H(\lambda) - E_\mu(\lambda)},
 \end{align}
the density matrix corresponding to those eigenmodes is
 \begin{align*}
\Gamma(\lambda) := \sum_{\mu=1}^{\nu} \proj{\vp_\mu(\lambda)} = \sum_{n=0}^{+\infty} \lambda^n \Gamma^n, \quad  \text{where} \quad \Gamma^n := \f{1}{n!}  \pa{\f{\d^n }{\d \lambda^n} \Gamma(\lambda) }_{\mkern 1mu \vrule height 2ex\mkern2mu \lambda = 0},
 \end{align*}
 and is independent of the choice of the frame $\vp_\mu(\lambda)$, as long as it respects~\eqref{eq:rel_frame}.

\subsection{Statement}%
\label{sub:Statement}

Let us take $(\vp_\mu)_{\mu=1}^{\nu} \in \cH^\nu$ such that $\vp_\mu \in \Ker \pa{H(0) - E_\mu(0)}$ and $\ps{\vp_\mu,\vp_\alpha} = \delta_{\mu\alpha}$ for any $\mu,\alpha \in \{1,\dots,\nu\}$, so $\Gamma^0 = \sum_{\mu=1}^{\nu} \proj{\vp_\mu}$. Let us define $A_0 := \Gamma^0$, $B_0 := 0$, $C_0 := 0$ and for any $n \in \Np$, $n \ge 1$,
\begin{multline}\label{eq:ABC} 
A_n := - \sum_{k=1}^{n-1} \pa{A_{n-k} A_k  + B_{n-k}^* B_k}, \qquad C_n :=  \sum_{k=1}^{n-1} \pa{C_{n-k} C_k  + B_{n-k} B_k^*} \\
b_n := \pa{\Gamma^0}^\perp \sum_{k=0}^{n-1} \pa{H^{n-k} \pa{A_k + B_k} - \pa{B_k + C_k} H^{n-k}} \Gamma^0 \\
B_n := \sum_{\mu =1}^\nu K_\mu(0) b_n P_{\vp_\mu},
\end{multline}
where $\sum_{m=a}^{b} := 0$ if $b < a$, and $K_\mu(0)$ is defined in~\eqref{def:pseudoinv_lambda_K}. We see that $A_n$ and $C_n$ are self-adjoint. The following result is classical and comes from~\cite{McWeeny62}. See also~\cite{TruDiaBow20} for other methods of computing $\Gamma^n$.

\begin{proposition}[Coefficients in density matrix perturbation theory, ~\cite{McWeeny62}]\label{prop:dmpt} 
	Let us consider a Hilbert space $\cH$, a self-adjoint energy operator $A$, an analytic family of self-adjoint operators $H(\lambda) = \sum_{n=0}^{+\infty} \lambda^n H^n$, we make the assumptions of Sections~\ref{sub:Analytic family of operators} and~\ref{sub:choose_set_eigvals}. Take $\nu \in \Np$, consider a set of eigenmodes $\pa{E_\mu(\lambda), \phi_\mu(\lambda)}_{\mu =1}^\nu$, analytic in $\lambda$, and the corresponding density matrix $\Gamma(\lambda) := \sum_{\mu=1}^{\nu} \proj{\phi_\mu(\lambda)} = \sum_{n=0}^{+\infty} \lambda^n \Gamma^n$. Then for any $n \in \Nz$, 
\begin{align}\label{eq:explicit_gamma} 
\Gamma^n = A_n + B_n + B_n^* + C_n,
\end{align}
where the involved operators are defined in~\eqref{eq:ABC}. 
\end{proposition}
For the sake of completeness, we give a rigorous proof in Section~\ref{sub:Proof of Proposition explicit_gamma}. 

Remark that $\Gamma^n$ is invariant under the gauge change $(\vp_\mu)_{\mu=1}^\nu =: \bvp \rightarrow U \bvp$, for any unitary $U \in \cU_\nu$, as long as~\eqref{eq:rel_frame} is respected. Hence $A_n$, $B_n$ and $C_n$ are also invariant under this transformation. So one does not need to compute the exact $\phi_\mu(0)$'s, which are notoriously hard to obtain.

Moreover, we have
\begin{align*}
A_n = \Gamma^0 \Gamma^n \Gamma^0, \;\;\; B_n = (\Gamma^0)^\perp \Gamma^n \Gamma^0, \;\;\; B_n^* = \Gamma^0 \Gamma^n (\Gamma^0)^\perp, \;\;\; C_n = (\Gamma^0)^\perp \Gamma^n (\Gamma^0)^\perp,
\end{align*}
and we can rewrite
\begin{multline*}
	A_n = - \Gamma^0 \sum_{k=1}^{n-1} \Gamma^{n-k} \Gamma^k \Gamma^0, \qquad C_n = \pa{\Gamma^0}^\perp \sum_{k=1}^{n-1} \Gamma^{n-k} \Gamma^k \pa{\Gamma^0}^\perp \\
	b_n = \pa{\Gamma^0}^\perp \sum_{k=0}^{n-1} [H^{n-k},\Gamma^k] \Gamma^0, \qquad B_n = \sum_{\mu =1}^\nu K_\mu(0) b_n P_{\vp_\mu}.
\end{multline*}

We recall that we defined in Section~\ref{sub:First definitions} an operator $A$ enabling to implement the energy norm, and we recall that $\nor{\cdot}{2} $ is the Hilbert-Schmidt norm. Then to prove~\eqref{eq:eigendm_bound} we will need to following bound on the $\Gamma^n$ series. 

\begin{proposition}[Bound for the coefficients of density matrix perturbation theory]\label{prop:dmpt_bound} 
There exist $a,b > 0$, independent of $n \in \Nz$, such that for any $n \in \Nz$, 
\begin{align}\label{eq:bound_Gamma_n} 
\nor{A \Gamma^n A}{2} \le a b^n.
\end{align}
\end{proposition}

\subsection{Proof of Proposition~\ref{prop:dmpt}}
\label{sub:Proof of Proposition explicit_gamma}

\subsubsection{First relations}%
\label{ssub:First relations}
The Euler-Lagrange equation $[H(\lambda),\Gamma(\lambda)] = 0$ gives that for any $n \in \Nz$,
\begin{align}\label{eq:euler_lagrange_series}
\sum_{k=0}^{n} \seg{H^{n-k}, \Gamma^k} = 0.
\end{align}
Moreover, $\Gamma(\lambda)^* = \Gamma(\lambda)$ and $\Gamma(\lambda)^2 = \Gamma(\lambda)$ so for any $n \in \Nz$, $\pa{\Gamma^n}^* = \Gamma^n$ and 
\begin{align}\label{eq:square_P_series}
\sum_{k=0}^{n} \Gamma^{n-k} \Gamma^k = \Gamma^n.
\end{align}

\subsubsection{Decomposition of the projection}%
\label{ssub:Decomposition of the projection}

We define
$P := \Gamma^0$, $P^\perp :=  1 - \Gamma^0$ and
\begin{align*}
\cA_n := P \Gamma^n P, \qquad \cB_n := P^\perp \Gamma^n P, \qquad \cC_n := P^\perp \Gamma^n P^\perp
\end{align*}
so $\cA_n^{*} = \cA_n$ and $\cC_n^* = \cC_n$ and
\begin{align}\label{eq:dec_gamma} 
\Gamma^n = \cA_n + \cB_n + \cB_n^* + \cC_n
\end{align}
so we want to compute the series $\cA_n$, $\cB_n$, $\cC_n$. 

\subsubsection{Formulas for $\cA_n$ and $\cC_n$}%

 We define the Liouvillian 
\begin{align*}
\cL := \seg{H^0, \cdot},
\end{align*}
and $\hls$ denotes the space of Hilbert-Schmidt operators, defined in~\eqref{eq:def_hls}. For any $B,F \in \hls$, we compute
\begin{align*}
	&(\cL B, F)_2 = \tr \pa{\cL B}^* F = \tr [H^0,B]^* F = \tr B^* H^0 F - \tr H^0 B^* F  \\
		     &\qquad = \tr B^* H^0F - \tr B^* F H^0 = \tr B^* [H^0,F] = \tr B^* \cL F = (B,\cL F)_2,
\end{align*}
hence $\cL$ is self-adjoint, or in other words, $\cL^* = \cL$. We define
\begin{align}\label{eq:def_XY} 
X_n := -\sum_{k=0}^{n-1} \seg{H^{n-k}, \Gamma^k}, \qquad Y_n := \sum_{k=1}^{n-1}\Gamma^{n-k} \Gamma^k.
\end{align}
Then~\eqref{eq:euler_lagrange_series} transforms to 
\begin{align*}
\cL \Gamma^n = X_n
\end{align*}
 and~\eqref{eq:square_P_series} transforms to
\begin{align}\label{eq:p_int}
\Gamma^n P + P \Gamma^n - \Gamma^n = -Y_n.
\end{align}
From~\eqref{eq:dec_gamma} we can compute $\Gamma^n P + P \Gamma^n - \Gamma^n = \cA_n - \cC_n$ so~\eqref{eq:p_int} implies $\cA_n - \cC_n = - Y_n$ and
\begin{align}\label{eq:AY} 
\cA_n = - P Y_n P \qquad \qquad \cC_n = P^\perp Y_n P^\perp.
\end{align}
So we can develop
\begin{multline*}
Y_n \underset{\substack{\eqref{eq:dec_gamma} \\ \eqref{eq:def_XY}}}{=} \; \sum_{k=1}^{n-1} \pa{\cA_{n-k} + \cB_{n-k} + \cB_{n-k}^* + \cC_{n-k}} \pa{\cA_k + \cB_k + \cB_k^* + \cC_k}\\
= \sum_{k=1}^{n-1} \cA_{n-k} \pa{\cA_k + \cB_k^*} + \cB_{n-k} \pa{\cA_k + \cB_k^{*}} + \cB_{n-k}^* \pa{\cC_k + \cB_k} + \cC_{n-k} \pa{\cB_k + \cC_k}.
\end{multline*}
Applying $P$ on the left and on the right, and applying $P^\perp$ on the left and on the right, together with~\eqref{eq:AY} we obtain 
\begin{align}\label{eq:formula_AC} 
	\cA_n = -\sum_{k=1}^{n-1} \pa{\cA_{n-k} \cA_k + \cB^*_{n-k} \cB_k}, \qquad \cC_n =  \sum_{k=1}^{n-1} \pa{\cC_{n-k} \cC_k  + \cB_{n-k} \cB_k^*}.
\end{align}

\subsubsection{$\cL \cB_n$}%

We have $H^0 P = P H^0 = PH^0 P$ and $H^0 P^\perp = P^\perp H^0 = P^\perp H^0 P^\perp$ hence for $Q, G \in \{P,P^\perp\}$ and for any operator $F \in \hls$,
\begin{align}\label{eq:property_cL}
\cL \pa{Q F G} = Q \pa{\cL F} G.
\end{align}
By taking $F = X_n$, $Q = P^\perp$, $G = P$, we have
\begin{align}\label{eq:LB} 
\cL \cB_n = P^\perp X_n P \underset{\substack{\eqref{eq:dec_gamma} \\ \eqref{eq:def_XY}}}{=} \; \sum_{k=0}^{n-1} \pa{\cB_k + \cC_k} H^{n-k} P - P^\perp H^{n-k} \pa{\cA_k + \cB_k}.
\end{align}

\subsubsection{Partial inverse of the Liouvillian}%
\label{ssub:partial_inv_liouvillian} 

We define
\begin{align*}
\cO := \{L \in \hls \;|\; L = P^\perp L P\}.
\end{align*}
Let us take $(\vp_\mu)_{\mu=1}^{\nu}$ such that $\vp_\mu \in \Ker \pa{H(0) - E_\mu(0)}$ and $\ps{\vp_\mu,\vp_\alpha} = \delta_{\mu\alpha}$ for any $\mu,\alpha \in \{1,\dots,\nu\}$
and the operator of $\hls$ 
	\begin{align}\label{eq:Lplus_full}
		\cL^+ : \hspace{0.5cm} 
	\begin{array}{rcl}
		\hls & \longrightarrow &  \hls \\
		F & \longmapsto & - \sum_{\mu=1}^\nu K_\mu(0) F P_{\vp_\mu}.
	\end{array}
	\end{align}
For any $F \in \hls$, we compute
\begin{align}\label{eq:LplusLF} 
	&\cL^+ \cL F \hspace{-0.05cm}=\hspace{-0.05cm} -\hspace{-0.05cm} \sum_{\mu =1}^\nu K_\mu(0) [H^0,F] P_{\vp_\mu} \hspace{-0.05cm}= \hspace{-0.05cm}\sum_{\mu =1}^\nu K_\mu(0) F H^0 P_{\vp_\mu} \hspace{-0.05cm}-\hspace{-0.05cm} K_\mu(0) H^0 F P_{\vp_\mu} \nonumber\\
&\quad = \sum_{\mu =1}^\nu E_\mu(0) K_\mu(0) F P_{\vp_\mu} + P^\perp F P_{\vp_\mu} - E_\mu(0) K_\mu(0) F P_{\vp_\mu} = P^\perp F P,
\end{align}
where we used that $K_\mu(0) \pa{H^0 - E_\mu(0)} = - P^\perp$. So if $F \in \cO$, then $\cL^+ \cL F = F$. By a similar computation, we have $\cL \cL^+ = \cL^+ \cL$. Moreover, $(\cL^+ \cL)^2 = \cL^+ \cL$ and $(\cL^+ \cL)^* = \cL^+ \cL$, where the dual operator is taken with respect to the scalar product $\hs{\cdot,\cdot}$. Hence $\cL^+ \cL$ is the orthogonal projection onto $\cO$, and $\cL^+$ is a partial inverse. 

\subsubsection{Formula for $\cB_n$}%

Since $\cB_n \in \cO$, and since $\cL^+ \cL$ is the orthogonal projection onto $\cO$, we have 
\begin{align}\label{eq:formula_Bn} 
	\cB_n &= \cL^+ \cL \cB_n \underset{\substack{\eqref{eq:LB}}}{=} \;\sum_{k=0}^{n-1} \cL^+ \bpa{ \pa{\cB_k + \cC_k} H^{n-k} P} - \cL^+ \bpa{P^\perp H^{n-k} \pa{\cA_k + \cB_k}} \nonumber\\
	&\underset{\substack{\eqref{eq:Lplus_full}}}{=} \; \sum_{\mu =1}^\nu\sum_{k=0}^{n-1} K_\mu(0) \pa{H^{n-k} \pa{\cA_k + \cB_k} - \pa{\cB_k + \cC_k} H^{n-k}} P_{\vp_\mu}.
\end{align}

\subsubsection{Conclusion}%
\label{ssub:Conclusion}

The recursive relations~\eqref{eq:formula_AC} and~\eqref{eq:formula_Bn} respected by $\cA_n$, $\cB_n$ and $\cC_n$ are the same as the ones respected by $A_n$, $B_n$ and $C_n$. Thus from $\cA_0 = A_0$, $\cB_0 = B_0$, $\cC_0 = C_0$, we conclude that $\cA_n = A_n$, $\cB_n = B_n$ and $\cC_n = C_n$ for any $n \in \Nz$. 

\subsection{Proof of Proposition~\ref{prop:dmpt_bound}}
We recall that
\begin{multline*}
	\nor{B}{2,ee} := \nor{A B A}{2},  \qquad c_K := \mymax{1 \le \mu \le \nu} \nor{A KA}{}, \\
	c_{H,\infty} := \mymax{n \in \Nz} \nor{A^{-1} H^n A^{-1}}{}.
\end{multline*}
For any $n \in \Nz$, we define
\begin{align*}
v_n := \max \pa{\nor{A_n}{2,ee}, \nor{B_n}{2,ee}, \nor{C_n}{2,ee}   }.
\end{align*}
Let us take $n \in \Np$ and. We have 
\begin{multline*}
A B_n A = \sum_{\mu =1}^\nu\sum_{k=0}^{n-1} A K_\mu(0) A  \\
\times \pa{A^{-1}H^{n-k} A^{-1} A \pa{A_k + B_k} A^{-1}- A^{-1}\pa{B_k + C_k} A A^{-1} H^{n-k} A^{-1}}  \\
\times A P_{\vp_\mu} A.
\end{multline*}
Moreover, for any $k \in \Nz$ and any $L \in \{A_k,B_k,C_k\}$, 
\begin{align*}
\nor{A L A^{-1}}{} = \nor{A L A A^{-2}}{} \le c_A^2 \nor{ALA}{} \le c_A^2 \nor{L}{2,ee} \le c_A^2 v_k, 
\end{align*}
 and 
\begin{align*}
\nor{A P_{\vp_\mu} A}{2} = \nor{A  \Gamma^0 P_{\vp_\mu}  \Gamma^0 A}{2} \le \nor{A \Gamma^0}{} \nor{P_{\vp_\mu}}{2} \nor{ \Gamma^0A}{}  = \nor{A \Gamma^0}{}^2,
\end{align*}
hence
\begin{align*}
\nor{B_n}{2,ee} \le 4 \nu c_K  c_{H,\infty} c_A^2 \nor{A \Gamma^0}{}^2   \sum_{k=0}^{n-1} v_k.
\end{align*}
For any $k \in \Nz$ we define $u_k := v_k +1$, we have $v_k \le u_k \le u_k u_{n-k}$ so for any $n \ge 1$,
\begin{align*}
	\nor{B_n}{2,ee} + 1 \le \pa{1+4 \nu c_K  c_{H,\infty} c_A^2 \nor{A \Gamma^0}{}^2 }\sum_{k=0}^{n-1} u_k u_{n-k}.
\end{align*}

Similarly, we have
\begin{align*}
\nor{A_n}{2,ee} \le 2 c_A^2 \sum_{k=1}^{n-1} v_{k} v_{n-k}, \qquad \qquad \nor{C_n}{2,ee} \le 2 c_A^2 \sum_{k=1}^{n-1} v_{k} v_{n-k},
\end{align*}
so for any $n \ge 2$,
\begin{align*}
	\nor{A_n}{2,ee} + 1 \le \pa{1 + 2 c_A^2} \sum_{k=1}^{n-1} u_{k} u_{n-k}, \quad \nor{C_n}{2,ee} + 1 \le \pa{1 + 2 c_A^2} \sum_{k=1}^{n-1} u_{k} u_{n-k}
\end{align*}
and we can conclude that
\begin{align*}
	u_n \le \pa{1+ 2 c_A^2 \max \pa{1, 2 \nu c_K  c_{H,\infty} \nor{A \Gamma^0}{}^2}}\sum_{k=1}^{n-1} u_{k} u_{n-k}.
\end{align*}
We obtain~\eqref{eq:bound_Gamma_n} by applying Lemma~\ref{lem:cauchy_square} to $u_n$.

\section{Proof of Corollaries~\ref{cor:main_ec_dm} and~\ref{cor:main_ec_vec}}%
\label{sec:Proof of EC}


We only give a proof of Corollary~\ref{cor:main_ec_vec} in detail, because the proof of Corollary~\ref{cor:main_ec_dm} uses the exact same method.

To apply Rellich's theorem, we remark that we automatically have 
\begin{align*}
\mymax{n \in \Nz} \nor{A^{-1} \cP H^n \cP A^{-1}}{} < +\infty
\end{align*}
 because $\nor{A^{-1} \cP H^n \cP A^{-1}}{} \le c_\cP^2 \nor{A^{-1} H^n A^{-1}}{}$, which was already assumed to be bounded.

\subsection{Proof of Corollary~\ref{cor:main_ec_vec} }%
\label{sub:Inequality on vec}

The proof of Corollary~\ref{cor:main_ec_vec} uses Proposition~\ref{prop:non-deg_main_thm}, and Lemma~\ref{lem:bound_phi_n}.

 Defining
\begin{align*}
c_{H,\infty} := \mymax{n \in \Nz} \nor{A^{-1} H^n A^{-1}}{} < +\infty,
\end{align*}
for any $\lambda < 1$ we have $\nor{A^{-1} H(\lambda) A^{-1}}{} \le c_{H,\infty} \pa{1 - \ab{\lambda}}^{-1}$ so for any $\lambda < 1/2$, 
\begin{align}\label{eq:bound_Hlambda} 
\nor{A^{-1} H(\lambda) A^{-1}}{} \le 2c_{H,\infty}.
\end{align}

We have
\begin{align}
&\nor{- \f 12 \nor{\phi(\lambda) - \psi(\lambda)}{}^2 \psi(\lambda) + \pa{\cE(\lambda) - E(\lambda)} R(\lambda) \pa{\phi(\lambda) - \psi(\lambda)}}{e} \nonumber \\
&\underset{\substack{\eqref{eq:hypo_RP}}}{\le} \; \f 12 c_A^2 \nor{\phi(\lambda) - \psi(\lambda)}{e}^2 \nor{\psi(\lambda)}{e}^2 + c_{\cP,R} \ab{ E(\lambda) - \cE(\lambda) } \nor{\phi(\lambda) - \psi(\lambda)}{e} \nonumber \\
&\underset{\substack{\eqref{eq:diff_errs_abs} \\ \eqref{eq:bound_Hlambda}}}{\le}  \nor{\phi(\lambda) - \psi(\lambda)}{e}^2 \nonumber \\
&\qquad \times \pa{\tfrac 12 c_A^2 \nor{\psi(\lambda)}{e}^2 + c_{\cP,R} \pa{c_A^2 \ab{E(\lambda)} +2c_{H,\infty}} \nor{\phi(\lambda) - \psi(\lambda)}{e}}. \label{eq:cela} 
\end{align}
Since $\phi(0) = \psi(0)$, and by continuity of the maps $\lambda \mapsto \phi(\lambda)$ and $\lambda \mapsto \psi(\lambda)$, we have $\nor{\phi(\lambda) - \psi(\lambda)}{e} \rightarrow 0$ as $\lambda \rightarrow 0$, and then we can take $\lambda_0$ small enough such that for any $\lambda \in ]-\lambda_0, \lambda_0[$,
\begin{multline}\label{eq:en_smaller_12} 
	\nor{\phi(\lambda) - \psi(\lambda)}{e} \\
	\le  \f 12  \pa{\tfrac 12  c_A^{2}\nor{\psi(\lambda)}{e}^2 + c_{\cP,R} \pa{c_A^2 \ab{E(\lambda)} + 2 c_{H,\infty}}\nor{\phi(\lambda) - \psi(\lambda)}{e}}^{-1}.
\end{multline}
We use $\ps{\phi(\lambda) , \psi(\lambda)} \in \R$, to apply Proposition~\ref{prop:non-deg_main_thm} at each $\lambda$. Thus from~\eqref{eq:equality_diff} and~\eqref{eq:cela} (see also~\eqref{eq:ineq_cons}) we obtain that for any $\lambda \in ]-\lambda_0,\lambda_0[$,
\begin{align*}
\nor{\phi(\lambda) - \psi(\lambda)}{e} \le 2 \pa{1 + c_A \nor{ARA}{} } \nor{\cP^\perp \phi(\lambda)}{e}.
\end{align*}
We recall from Appendix~\ref{sec:interm_normalization} that
\begin{align*}
\Phi(\lambda) := \f{\phi(\lambda)}{\ps{\phi^0, \phi(\lambda)}},  \qquad \Phi^n := \f{1}{n!} \pa{\f{\d^n}{\d \lambda^n} \Phi(\lambda)}_{\mkern 1mu \vrule height 2ex\mkern2mu \lambda = 0}, \qquad \phi(\lambda) = \f{\Phi(\lambda)}{\nor{\Phi(\lambda)}{}}.
\end{align*}
so
\begin{align*}
\cP^\perp \phi(\lambda) = \nor{\Phi(\lambda)}{}^{-1} \cP^\perp \Phi(\lambda) \underset{\substack{~\eqref{eq:cond_der_vec}}}{=} \; \nor{\Phi(\lambda)}{}^{-1} \sum_{n=\ell +1}^{+\infty} \lambda^n \Phi^n.
\end{align*}
We use the phasis gauge $\ps{\phi^0 , \phi(\lambda)} \in \R_+$ to obtain the bounds on the derivatives~\eqref{eq:bound_Q}, and thus there is $c,b > 0$ independent of $\ell$ and $\lambda$ such that
\begin{align*}
\nor{\phi(\lambda) - \psi(\lambda)}{e} \le c \pa{\ab{\lambda} b}^{\ell +1}.
\end{align*}
From this, we deduce that $\phi^n = \psi^n$ for any $n \in \{0,\dots,\ell\}$. From this last inequality and~\eqref{eq:diff_errs_abs} we also obtain that for some $c,b > 0$ independent of $\ell$ and $\lambda$,
\begin{align}\label{eq:ths} 
	\nor{\phi(\lambda) - \psi(\lambda)}{}^2 \le c \pa{\ab{\lambda} b}^{2(\ell +1)}, \qquad  \ab{\cE(\lambda) - E(\lambda)}  \le c \pa{\ab{\lambda} b}^{2(\ell +1)}.
\end{align}

Using~\eqref{eq:equality_diff} once more, we have
\begin{multline}\label{eq:diff_precise} 
\phi(\lambda) - \psi(\lambda) - \lambda^{\ell +1 }\pa{1 + R(0)H(0)} \cP^\perp \phi^{\ell +1} \\
= \pa{1 + R(\lambda)H(\lambda)} \cP^\perp \pa{\phi(\lambda) -  \lambda^{\ell +1 }\phi^{\ell +1}} \\
+  \lambda^{\ell +1 }\pa{R(\lambda) H(\lambda) - R(0)H(0)} \cP^\perp \phi^{\ell +1} - \f 12 \nor{\phi(\lambda) - \psi(\lambda)}{}^2 \psi(\lambda) \\
+ \pa{\cE(\lambda) - E(\lambda)} R(\lambda) \phi(\lambda).
\end{multline}
We now seek to bound each of those terms. First,
\begin{align*}
\nor{\cP^\perp \pa{ \phi(\lambda)  - \lambda^{\ell +1 }\phi^{\ell +1}}}{e} = \nor{\sum_{n=\ell +2}^{+\infty} \lambda^n \phi^n}{e} \underset{\substack{\eqref{eq:bound_Q}}}{\le} \; c \pa{\ab{\lambda} b}^{\ell + 2}
\end{align*}
for some $c,b > 0$ independent of $\ell$ and $\lambda$. Then, by analyticity of $\lambda \mapsto R(\lambda) H(\lambda)$, at $\lambda = 0$, we have
\begin{align*}
\nor{ A \pa{R(\lambda) H(\lambda) - R(0)H(0) } A^{-1}}{} \le c \lambda,
\end{align*}
where $c$ does not depend on $\lambda$. We can reproduce the same reasoning for the norm $\nor{\cdot}{} $. Finally, also using~\eqref{eq:ths}, ~\eqref{eq:diff_precise} yields, for $\delta \in \{0,1\}$,
\begin{align*}
&\ab{\nor{\phi(\lambda) - \psi(\lambda)}{e,\delta} - \ab{\lambda}^{\ell +1} \xi^{\textup{non-deg}}_{\textup{RBM+PT},\ell}}\\
& \qquad \qquad \le \nor{\phi(\lambda) - \psi(\lambda) - \lambda^{\ell +1 }\pa{1 + R(0)H(0)} \cP^\perp \phi^{\ell +1} }{e,\delta}  \\
& \qquad \qquad \le c \pa{\ab{\lambda} b}^{\ell + 2}.
\end{align*}

The proof of the eigenvalue bound~\eqref{eq:eigenvectors_bound} is similar.

\subsection{Proof of Corollary~\ref{cor:main_ec_dm}}%

We remark directly from~\eqref{eq:explicit_diff} that the leading order of $\Gamma(\lambda) - \Lambda(\lambda)$ is
\begin{multline*}
	\sum_{\mu=1}^{\nu} \pa{ \pa{1 + R_\mu(\lambda)H(\lambda)} \cP^\perp \Gamma(\lambda) P_{\psi_\mu(\lambda)} + \text{adj.}} \\
	= \lambda^{\ell +1} \sum_{\mu=1}^{\nu} \pa{ \pa{1 + R_\mu(0)H(0)} \cP^\perp \Gamma^{\ell +1} P_{\phi_\mu(0)} + \text{adj.}} + O(\lambda^{\ell +2}).
\end{multline*}
We used that $\Gamma(\lambda) \cP^\perp = \sum_{n=\ell+1}^{+\infty} \lambda^n \Gamma^n \cP^\perp$, because $\cP^\perp \Gamma^k = 0$ for all $k \in \{0,\dots,\ell\}$ by the assumption~\eqref{eq:cond_der_vec} stating that $\im \Gamma^k \subset \cP \cH$. The bounds~\eqref{eq:eigendm_bound} and~\eqref{eq:eigenvals_bound} are obtained by using similar arguments, and follow the same steps. We need the bound~\eqref{eq:bound_Gamma_n} on the derivatives $\Gamma^n$ for showing~\eqref{eq:eigendm_bound}.

\subsection{Proof of Lemma~\ref{lem:build_red_space_dm} }%
\label{sub:Proof of Lemma span space }

We have $\phi_\mu(\lambda) = \sum_{n=0}^{+\infty} \lambda^n \phi_\mu^n$ and
\begin{align*}
	\Gamma(\lambda) = \sum_{\mu=1}^{\nu} \proj{\phi_\mu(\lambda)} = \sum_{\substack{0 \le k,p < +\infty \\ 1 \le \mu \le \nu}} \lambda^{k+p} \ketbra{\phi_\mu^k}{\phi_\mu^p}
\end{align*}
hence identifying the coefficients of $\lambda^n$ gives $\Gamma^n = \sum_{k=0}^{n} \sum_{\mu=1}^{\nu} \ketbra{\phi_\mu^{n-k}}{\phi_\mu^k}$. From this we see that
\begin{align}\label{eq:ima} 
\bigoplus_{n=0}^\ell \im \Gamma^n \subset \Span\pa{\phi^n_\mu \;|\; 0\le n \le \ell, 1 \le \mu \le \nu}.
\end{align}
Moreover, for any $\xi \in \cH$,
\begin{align}\label{eq:gamma_applied} 
	\Gamma^n \xi = \sum_{\alpha=1}^{\nu} \phi^n_\alpha \ps{\phi_\alpha^0 , \xi} + \sum_{\substack{1 \le \alpha \le \nu \\ 0 \le k \le n-1}} \phi_\alpha^{k} \ps{\phi^{n-k}_\alpha, \xi}.
\end{align}
Since $(\vp_\mu)_{\mu=1}^\nu$ is a basis of $\Span\pa{(\phi^0_\mu)_{\mu=1}^\nu}$, this relation enables to show recursively the following proposition for any $n \in \Nz$,
\begin{multline*}
	\mathbb{P}(n) : \\
\Span\pa{\Gamma^k \vp_\alpha \;|\; 0\le k \le n, 1 \le \alpha \le \nu} = \Span\pa{\phi_\alpha^k \;|\; 0\le k \le n, 1 \le \alpha \le \nu}.
\end{multline*}
Indeed, $\mathbb{P}(0)$ holds because $(\vp_\mu)_{\mu=1}^\nu$ is a basis of $\Span\pa{(\phi^0_\mu)_{\mu=1}^\nu}$ by assumption. Then take $n \in \Nz$ such that $\mathbb{P}(n)$ holds, we want to prove that $\mathbb{P}(n+1)$ holds. First, take $\mu\in\{1,\dots,\nu\}$, by~\eqref{eq:gamma_applied} applied to $\xi = \vp_\mu$ and with $n \rightarrow n+1$, we have $\Gamma^{n+1} \vp_\mu \in \Span\pa{\pa{\phi^k_\alpha}_{1 \le \alpha \le \nu}^{0\le k \le n+1}}$, hence, also using $\mathbb{P}(n)$, we obtain
\begin{multline*}
	\Span\pa{\Gamma^k \vp_\alpha \;|\; 0\le k \le n+1, 1 \le \alpha \le \nu}\\
	\subset \Span\pa{\phi^k_\alpha \;|\; 0\le k \le n+1, 1 \le \alpha \le \nu}.
\end{multline*}
There exists $(c_\alpha)_{1 \le \alpha \le \nu} \in \C^\nu$ such that $\phi^0_\mu = \sum_{\alpha=1}^{\nu} c_\alpha \vp_\alpha$, so~\eqref{eq:gamma_applied} applied to $\xi = \phi^{0}_\mu$ and with $n \rightarrow n +1$ gives
\begin{multline*}
	\phi^{n+1}_\mu = \Gamma^{n+1} \phi^0_\mu - \sum_{\substack{1 \le \alpha \le \nu \\ 0 \le k \le n}} \phi_\alpha^{k} \ps{\phi^{n+1-k}_\alpha, \phi^0_\mu} \\
	= \sum_{\alpha=1}^{\nu} c_{\alpha}  \Gamma^{n+1} \vp_\alpha - \sum_{\substack{1 \le \alpha \le \nu \\ 0 \le k \le n}} \phi_\alpha^{k} \ps{\phi^{n+1-k}_\alpha, \phi^0_\mu}.
\end{multline*}
We then use $\mathbb{P}(n)$, more precisely that
\begin{align*}
 \Span\pa{\phi^k_\alpha \;|\; 0\le k \le n, 1 \le \alpha \le \nu}\subset\Span\pa{\Gamma^k \vp_\alpha \;|\; 0\le k \le n, 1 \le \alpha \le \nu},
\end{align*}
to conclude that $\phi^{n+1}_\mu \in \Span\pa{\pa{\Gamma^k \vp_\alpha}_{1 \le \alpha \le \nu}^{0\le k \le n+1}}$ and hence
\begin{multline*}
	\Span\pa{\phi^k_\alpha \;|\; 0\le k \le n+1, 1 \le \alpha \le \nu} \\
	\subset \Span\pa{\Gamma^k \vp_\alpha \;|\; 0\le k \le n+1, 1 \le \alpha \le \nu},
\end{multline*}
and we obtain $\mathbb{P}(n+1)$.

From $\mathbb{P}(\ell)$ we have that
\begin{align*}
\Span\pa{\phi^n_\mu \;|\; 0\le n \le \ell, 1 \le \mu \le \nu} \subset\bigoplus_{n=0}^\ell \im \Gamma^n,
\end{align*}
and together with~\eqref{eq:ima} this gives
\begin{align*}
\bigoplus_{n=0}^\ell \im \Gamma^n = \Span\pa{\phi^n_\mu \;|\; 0\le n \le \ell, 1 \le \mu \le \nu}.
\end{align*}

\section{Proof of Theorem~\ref{thm:deg_vecs} }%
\label{sec:Proof deg vec}

We consider Appendix~\ref{sec:interm_normalization} for intermediate normalization. Let us recall that
\begin{align*} 
\Phi_\mu(\lambda) &:= \f{\phi_\mu(\lambda)}{\ps{\phi_\mu^0,\phi_\mu(\lambda)}}, \qquad  \Psi_\mu(\lambda) := \f{\psi_\mu(\lambda)}{\ps{\phi_\mu^0,\psi_\mu(\lambda)}}, \\
\phi_\mu^n &:= \f{1}{n!}  \pa{\f{\d^n }{\d \lambda^n} \phi_\mu(\lambda) }_{\mkern 1mu \vrule height 2ex\mkern2mu \lambda = 0}, \qquad\Phi_\mu^n := \f{1}{n!}  \pa{\f{\d^n }{\d \lambda^n} \Phi_\mu(\lambda) }_{\mkern 1mu \vrule height 2ex\mkern2mu \lambda = 0}, \\
&\qquad \Psi_\mu^n := \f{1}{n!}  \pa{\f{\d^n }{\d \lambda^n} \Psi_\mu(\lambda) }_{\mkern 1mu \vrule height 2ex\mkern2mu \lambda = 0}, \qquad  \cE_\mu^n := \f{1}{n!}  \pa{\f{\d^n }{\d \lambda^n} \cE_\mu(\lambda) }_{\mkern 1mu \vrule height 2ex\mkern2mu \lambda = 0}. \nonumber
\end{align*}
The proof of this result is different from the proof of Corollaries~\ref{cor:main_ec_dm} and~\ref{cor:main_ec_vec}. In particular it does not use the results of Section~\ref{sec:main result on the RBM}.   

\subsection{Core lemma}%
\label{sub:Core lemma}

Before starting the proof, we show the following lemma, giving the error at order $n+1$ when the previous orders are equal.

\begin{lemma}\label{lem:tech_compute_diff} 
Make the same assumptions as in Theorem~\ref{thm:deg_vecs}. Take $n \in \Nz$. If for all $k \in \{0,\dots,n\}$, $\Phi^k_\mu = \Psi_\mu^{k}$ and $E^k_\mu = \cE_\mu^{k}$, then $E^{n+1}_\mu = \cE^{n+1}_\mu$ and
\begin{align}\label{eq:lems_eq} 
\Phi^{n+1}_\mu - \Psi^{n +1}_\mu = \pa{1+\rr_\mu(0) H^1} \pa{1 + R_\mu(0) H^0} \cP^\perp \Phi^{n+1}_\mu.
\end{align}
\end{lemma}
\begin{proof}

We have $\phi^0_\alpha \in \cP \cH$ for any $\alpha \in \{1,\dots,\nu\}$ so
\begin{align}\label{eq:commute_gamma} 
\Gamma^0 \cP = \cP \Gamma^0 = \cP,
\end{align}
hence $\cP \pa{\Gamma^0}^\perp = \pa{\Gamma^0}^\perp \cP$, we have
\begin{align*}
1 = \cP^\perp + \cP \pa{\Gamma^0}^\perp + \Gamma^0 P_{\phi^0_\mu}^\perp + P_{\phi^0_\mu}
\end{align*}
and we will split
\begin{align*}
	\cH = \cP^\perp \cH \oplus \cP \pa{\Gamma^0}^\perp \cH \oplus \Gamma^0 P_{\phi^0_\mu}^\perp \cH  \oplus P_{\phi^0_\mu} \cH.
\end{align*}
We will compute $\Phi^{n+1}_\mu - \Psi^{n +1}_\mu$ on each of those subspaces.

We define $\xi^q_\mu := \Phi^{q}_\mu - \Psi^{q}_\mu$ for any $q \in \Nz$. For any $q \in \Np$, 
\begin{align}\label{eq:perpos} 
\Phi^q_\mu \perp \Phi^0_\mu \qquad \text{and} \qquad \Psi^q_\mu \perp \Psi^0_\mu = \Phi^0_\mu, \qquad \text{ hence } \quad P_{\phi^0_\mu} \xi^{q}_\mu = 0.
\end{align}

We define $w^{k}_\mu := H^k - \cE_\mu^{k}$, $h^k_\mu := H^k - E_\mu^k$ and 
\begin{align*}
w_\mu(\lambda) := H(\lambda) - \cE_\mu(\lambda) = \sum_{k=0}^{+\infty} \lambda^k w^{k}_\mu, \quad h_\mu(\lambda) := H(\lambda) - E_\mu(\lambda) = \sum_{k=0}^{+\infty} \lambda^k h^{k}_\mu.
\end{align*}
Since $\Psi_\mu(\lambda)$ is en eigenvector of $\PHlP$ with eigenvalue $\cE_\mu(\lambda)$, $\cP w_\mu(\lambda) \Psi_\mu(\lambda) = 0$, 
so identifying the different factors of $\lambda^q$ of the last equation, for any $q \in \Nz$ we have that 
\begin{align}\label{eq:schro_pert_w_psi} 
\sum_{k=0}^{q} \cP w^{q-k}_\mu \Psi^k_\mu = 0.
\end{align}
Since $\pa{E_\mu(\lambda),\Phi_\mu(\lambda)}$ is an eigenmode of $H(\lambda)$, $h_\mu(\lambda) \Phi_\mu(\lambda) = 0$, and this yields that for any $q \in \Nz$,
\begin{align}\label{eq:schro_Phi} 
\sum_{k=0}^{q} h^{q-k}_\mu \Phi^k_\mu = 0.
\end{align}
Applying $\cP$ to~\eqref{eq:schro_Phi} and substracting~\eqref{eq:schro_pert_w_psi} yields
\begin{align}\label{eq:fu} 
	0 &= \sum_{k=0}^{q} \cP \pa{h^{q-k}_\mu \Phi^k_\mu - w^{q-k}_\mu \Psi^k_\mu} = \sum_{k=0}^{q} \cP \pa{h^{q-k}_\mu \xi^k_\mu + \bpa{\cE^{q-k}_\mu - E^{q-k}_\mu} \Psi^k_\mu} \nonumber \\
	  &=\sum_{k=0}^{q} \cP \pa{h^{q-k}_\mu \xi^k_\mu + \bpa{\cE^{k}_\mu - E^{k}_\mu} \Psi^{q-k}_\mu}.
\end{align}
We know that $\cE^{k}_\mu - E^{k}_\mu = 0$ and $\xi_\mu^k = 0$ for all $k\in \{0,\dots,n\}$. So using~\eqref{eq:fu} with $q=n+1$ gives
\begin{align*}
	0 = \cP h^0_\mu \xi_\mu^{n+1} + \pa{\cE^{n+1}_\mu - E^{n+1}_\mu} \Phi^0_\mu.
\end{align*}
Taking the scalar product with $\Phi^0_\mu$ gives $\cE^{n+1}_\mu = E^{n+1}_\mu$ and applying $P_{\phi^0_\mu}^\perp$ gives $\cP h^0_\mu \xi^{n+1}_\mu = 0$, so
\begin{align*}
\cP h^0_\mu \cP \xi^{n+1}_\mu =  -  \cP h^0_\mu \cP^\perp \xi^{n+1}_\mu = -  \cP H^0 \cP^\perp \Phi^{n+1}_\mu
\end{align*}
and applying $R_\mu(0)$ yields 
\begin{align}\label{eq:xi_zero_gamma_perp} 
\cP \pa{\Gamma^0}^\perp \xi^{n+1}_\mu = R_\mu(0) H^0 \cP^\perp \Phi^{n+1}_\mu,
\end{align}
and
\begin{align}\label{eq:xi_zero_gamma_perp2} 
 \pa{\Gamma^0}^\perp \xi^{n+1}_\mu =\cP \pa{\Gamma^0}^\perp \xi^{n+1}_\mu  + \cP^\perp \xi^{n+1}_\mu  = (1+R_\mu(0) H^0) \cP^\perp \Phi^{n+1}_\mu.
\end{align}

Next, applying~\eqref{eq:fu} with $q=n+2$ gives
\begin{align*}
0 = \cP h^1_\mu \xi^{n+1}_\mu + \cP h^0_\mu \xi^{n+2}_\mu + \pa{\cE^{n+2}_\mu - E^{n+2}_\mu} \Phi^0_\mu.
\end{align*}
Applying $\Gamma^0 P_{\phi^0_\mu}^\perp$ and using~\eqref{eq:commute_gamma} gives
\begin{align*}
	0 &=  \Gamma^0 P_{\phi^0_\mu}^\perp h^1_\mu \xi^{n+1}_\mu = \Gamma^0 P_{\phi^0_\mu}^\perp h^1_\mu \pa{\cP^\perp + \cP \pa{\Gamma^0}^\perp +  \Gamma^0 P_{\phi^0_\mu}^\perp + P_{\phi^0_\mu}}\xi^{n+1}_\mu \\
	  &\underset{\substack{P_{\phi^0_\mu} \xi^{q}_\mu = 0 \\\eqref{eq:xi_zero_gamma_perp}}}{=} \;\Gamma^0 P_{\phi^0_\mu}^\perp h^1_\mu  \Gamma^0 P_{\phi^0_\mu}^\perp \xi^{n+1}_\mu + \Gamma^0 P_{\phi^0_\mu}^\perp h^1_\mu \pa{1 + R_\mu(0) H^0} \cP^\perp \Phi^{n+1}_\mu.
\end{align*}
We now apply $\rr_\mu(0)$, being such that $\rr_\mu(0) \Gamma^0 P_{\phi^0_\mu}^\perp h^1_\mu  \Gamma^0 P_{\phi^0_\mu}^\perp = -\Gamma^0 P_{\phi^0_\mu}^\perp$, which gives
\begin{align*}
\Gamma^0 P_{\phi^0_\mu}^\perp \xi^{n+1}_\mu = \rr_\mu(0) h^1_\mu \pa{1 + R_\mu(0) H^0} \cP^\perp \Phi^{n+1}_\mu.
\end{align*}

Finally, using it, together with~\eqref{eq:xi_zero_gamma_perp2} and $P_{\phi^0_\mu} \xi^{q}_\mu = 0$ yields
\begin{align*}
	\xi^{n+1}_\mu &= \pa{\pa{\Gamma^0}^\perp +  \Gamma^0 P_{\phi^0_\mu}^\perp + P_{\phi^0_\mu}} \xi^{n+1}_\mu  \\
	&= \pa{1+\rr_\mu(0) h^1_\mu} \pa{1 + R_\mu(0) H^0} \cP^\perp \Phi^{n+1}_\mu \\
	& = \pa{1+\rr_\mu(0) H^1 } \pa{1 + R_\mu(0) H^0} \cP^\perp \Phi^{n+1}_\mu,
\end{align*}
where we used that $\Gamma(0) R_\mu(0) = 0$, and hence $\rr_\mu(0) R_\mu(0)= 0$, in the last line.
\end{proof}

We then transform the last result into a result on the intermediate normalization series.
\begin{lemma}\label{lem:tech_compute_diff_unit} 
Take $n \in \Nz$. If for all $k \in \{0,\dots,n\}$, $\Phi^k_\mu = \Psi_\mu^{k}$, then for all $k \in \{0,\dots,n\}$, $\phi_\mu^{k} = \psi_\mu^{k}$, and $\phi_\mu^{n+1} - \psi_\mu^{n+1} = \Phi_\mu^{n+1} - \Psi_\mu^{n+1}$.
\end{lemma}

\begin{proof}
As in Lemma~\ref{lem:interm_to_unit_norm}, for $\Theta \in \{\Phi, \Psi\}$, we define $Y_\Theta^0 := 1, Y_\Theta^1 := 0$, and for any $q \in \Np$,
\begin{align*}
Y_\Theta^q := \f 12 \sum_{k=1}^{q-1}  \pa{\ps{\Theta_\mu^{q-k}, \Theta_\mu^k} - Y_\Theta^{q-k} Y_\Theta^k},
\end{align*}
and we have
\begin{align*}
\phi_\mu^q = \Phi_\mu^q - \sum_{k=0}^{q-2} Y^{q-k}_\Phi \phi_\mu^k, \qquad \psi_\mu^q = \Psi_\mu^q - \sum_{k=0}^{q-2} Y^{q-k}_\Psi \psi_\mu^k.
\end{align*}
Since for any $k \in \{0,\dots,n\}$, $\Phi_\mu^k = \Psi_\mu^k$, then one can prove by induction that $Y^k_{\Phi}= Y^k_{\Psi}$ for any $k \in \{0,\dots,n +1\}$, then $\phi_\mu^k = \psi_\mu^k$ for any $k \in \{0,\dots,n\}$ and $\phi_\mu^{n +1} - \psi_\mu^{n +1} = \Phi_\mu^{n +1} - \Psi_\mu^{n +1}$. 
\end{proof}

\subsection{Proof of~\eqref{eq:eigenvectors_bound_deg}}
\label{sub:Proof deg case}

We are now ready to prove~\eqref{eq:eigenvectors_bound_deg}.

\subsubsection{From $n=0$ to $n=\ell$}%

We make a recursive proof on $n \in \{0,\dots,\ell\}$ of the proposition
\begin{multline}
\cF(n) : \\
\quad \forall k \in \{0,\dots,n\}, p \in \{0,\dots,2n \}, \quad \Phi^k_\mu = \Psi_\mu^{k}, \quad \phi^k_\mu = \psi^k_\mu \text{ and } E^p_\mu = \cE_\mu^{p}.
\end{multline}
We have $\Phi^0_\mu = \Psi_\mu^{0} = \phi^0_\mu = \phi_\mu(0)$ and $E^0_\mu = \cE_\mu^{0} = E_\mu(0)$, proving $\cF(0)$. 

Let us now take $n \in \{0,\dots,\ell-1\}$, assume $\cF(n)$ and we want to show $\cF(n+1)$, that is we want to show that $\Phi^{n+1}_\mu = \Psi_\mu^{n+1}$, $\phi^{n+1}_\mu = \psi_\mu^{n+1}$,and that $E^p_\mu = \cE_\mu^{p}$ for $p\in \{2n+1,2n+2\}$. Since $\cP^\perp \Phi^{n+1}_\mu = 0$, applying Lemma~\ref{lem:tech_compute_diff} yields $\Phi^{n+1}_\mu = \Psi_\mu^{n+1}$ and applying Lemma~\ref{lem:tech_compute_diff_unit} yields $\phi^{n+1}_\mu = \psi_\mu^{n+1}$. Then we have
\begin{align*}
\phi_\mu(\lambda) - \psi_\mu(\lambda) = \sum_{k=n +2 }^{+\infty}  \lambda^k \bpa{\phi_\mu^k - \psi_\mu^k}.
\end{align*}
We use Lemma~\ref{lem:bound_phi_n}, i.e. that $\nor{\phi_\mu^k}{e}+ \nor{\psi_\mu^k}{e} \le ab^k$ for any $k \in \Nz$, some $a,b >0$. We have
\begin{align*}
	\nor{\phi_\mu(\lambda) - \psi_\mu(\lambda)}{e} &\le \sum_{k=n+2}^{+\infty} \ab{\lambda}^k \pa{\nor{\phi_\mu^k}{e} + \nor{\psi_\mu^k}{e}} \\
& \le \f{2a}{1 - \ab{\lambda}b} \pa{\ab{\lambda} b}^{n+2} \le c  \pa{\ab{\lambda} b}^{n+2},
\end{align*}
for some constant $c > 0$ independent of $\lambda$ and $n$. Applying it with~\eqref{eq:diff_errs_abs} gives
\begin{align*}
	\ab{E_\mu(\lambda) - \cE_\mu(\lambda)} \le c \pa{\ab{\lambda} b}^{2n+4}
\end{align*}
where $c$ is independent of $\lambda$ and $n$. Letting $\lambda \rightarrow 0$ gives $E^p_\mu = \cE_\mu^{p}$ for $p\in \{2n+1,2n+2\}$ as expected, and this concludes the induction, showing $\cF(n)$ for all $n \in \{0,\dots,\ell\}$.

\subsubsection{$n=\ell$ and the conclusion}%

By Lemma~\ref{lem:tech_compute_diff_unit}, we have
\begin{align*}
\phi_\mu^{\ell+1} - \psi_\mu^{\ell+1} = \Phi_\mu^{\ell+1} - \Psi_\mu^{\ell+1}.
\end{align*}
Applying $\cP^\perp$ yields $\cP^\perp \phi_\mu^{n +1} = \cP^\perp \phi_\mu^{n +1}$ and thus with~\eqref{eq:lems_eq},
\begin{align*}
\phi_\mu^{\ell+1} - \psi_\mu^{\ell+1} = \pa{1 + \rr_\mu(0) H^1}\pa{1 + R_\mu(0) H^0} \cP^\perp \phi_\mu^{\ell+1}.
\end{align*}
Returning to the series,
\begin{align*}
\phi_\mu(\lambda) - \psi_\mu(\lambda) = \lambda^{\ell +1} \bpa{\phi^{\ell+1}_\mu - \psi^{\ell+1}_\mu} + \sum_{n=\ell +2}^{+\infty} \lambda^n \pa{\phi^n_\mu - \psi^n_\mu}.
\end{align*}
We obtain~\eqref{eq:eigenvectors_bound_deg} by using the same reasoning as in Section~\ref{sub:Inequality on vec}, and we need Lemma~\ref{lem:bound_phi_n_deg}.

\section{Proof of Lemma~\ref{lem:perturbation_theory_bound}}%
\label{sec:Proof of Lemma}

We have 
\begin{align*}
	\phi_\mu(\lambda) - \vp_\mu(\lambda) &= \phi_\mu(\lambda) - \sum_{n=0}^{\ell} \lambda^n \phi_\mu^n + \sum_{n=0}^{\ell} \lambda^n \phi_\mu^n - \vp_\mu(\lambda) \\
&= \sum_{n=\ell +1}^{+\infty} \lambda^n \phi_\mu^n + \pa{1 - \nor{\sum_{n=0}^{\ell} \lambda^n \phi_\mu^n}{}^{-1} } \sum_{n=0}^{\ell} \lambda^n \phi_\mu^n.
\end{align*}
Then we write $1 = \nor{\sum_{n=0}^{+\infty} \lambda^n \phi_\mu^n}{}^{-1}$ and use that for any $u,v \in \R_+^*$, 
\begin{align*}
\ab{u^{-1} - v^{-1}}\le \ab{u - v} u^{-1} v^{-1}
\end{align*}
 so   
\begin{multline*}
	\nor{\phi_\mu(\lambda) - \vp_\mu(\lambda)}{e,\delta} \le \nor{\sum_{n=\ell +1}^{+\infty} \lambda^n \phi_\mu^n}{e,\delta} \\
	+ \ab{\nor{\sum_{n=0}^{+\infty} \lambda^n \phi_\mu^n}{} - \nor{\sum_{n=0}^{\ell} \lambda^n \phi_\mu^n}{} }\nor{\sum_{n=0}^{\ell} \lambda^n \phi_\mu^n}{}^{-1} \nor{\sum_{n=0}^{\ell} \lambda^n \phi_\mu^n}{e,\delta} \\
\le \nor{\sum_{n=\ell +1}^{+\infty} \lambda^n \phi_\mu^n}{e,\delta} \pa{1 + \nor{A^{-1}}{}^\delta \nor{\sum_{n=0}^{\ell} \lambda^n \phi_\mu^n}{}^{-1} \nor{\sum_{n=0}^{\ell} \lambda^n \phi_\mu^n}{e,\delta}}.
\end{multline*}
Finally, $\nor{\sum_{n=\ell +1}^{+\infty} \lambda^n \phi_\mu^n}{e,\delta} \le \sum_{n=\ell +1}^{+\infty} \ab{\lambda}^n \nor{\phi_\mu^n}{e,\delta}$ and we apply Lemma~\ref{lem:bound_phi_n}. 

The bound en eigenvalues can be deduced from the previous one.

\subsection*{Acknowledgement}
We warmly thank Long Meng, Eric Cancès and Antoine Levitt for useful discussions. The code of simulations will be made available on reasonable request.

\appendix

\section{Table of notations}%
\label{sec:Table of notations}

We recall here some of the main notations used in this document, together with their significations.

\vspace{1cm}
\begin{addmargin}[-2cm]{2cm}
	\begin{tabular}{c|l}
		Object & Notation \\
		\hline \\
		$\cH$ & Hilbert space on which operators act, or total ``variational'' space \\
		$\cP \cH$ & Reduced space \\
		$\cP$ & Orthogonal projection onto the reduced space \\
		$\cP^\perp$ & Orthogonal projection onto $(\cP \cH)^\perp$, equals $1-\cP$ \\
		$A$ & ``Energy operator'' enabling to define the energy norms, see Section~\ref{sub:First definitions}  \\
		$H$ & Exact operator considered \\
		$\PHP$ & Effective operator supposed to reproduce a subset of the eigenmodes of $H$ \\
		$\nu$ & Number of approximated eigenmodes \\
		$(\phi_\mu,E_\mu)$ & Exact eigenmode (the one of $H$) \\
		$(\psi_\mu,\cE_\mu)$ & Approximate eigenmode (the one of $\PHP$) \\
		$\Gamma$ & Exact eigen-density matrix (the one of $H$) \\
		$\Lambda$ & Approximate eigen-density matrix (the one of $\PHP$) \\
		$R_\mu$ & Partial inverse, defined in~\eqref{eq:def_R} \\
		$c_A$ & Norm of $A^{-1}$, defined in~\eqref{eq:cA_cH} \\
		$c_H$ & Energy norm of $H$, defined in~\eqref{eq:cA_cH} \\
		$c_\cP$ & Energy norm of $\cP$, defined in~\eqref{eq:ccP} \\
		$\nor{\cdot}{e}$ & Energy norm for vectors, defined in~\eqref{eq:energy_norm_vects} \\
		$\nor{\cdot}{2,e}$ & Energy norm for density matrices, defined in~\eqref{eq:energy_norm} \\
		$\lambda$ & Perturbative parameter \\
		$\ell$ & Perturbative order \\
		$f^\ell$ & $\ell^{\text{th}}$ derivative of $f(\lambda)$ at $\lambda = 0$, up to a factor, where $\lambda \mapsto f(\lambda)$ is any map \\
		$\vp^\ell(\lambda)$ & Perturbative approximation of $\phi(\lambda)$ at order $\ell$, defined in~\eqref{eq:pt_approx}
	\end{tabular}
\end{addmargin}

\section{Other kinds of couplings between reduced basis methods and perturbation theory}%
\label{sec:Other kinds}

In the physics literature, the reduced basis method or subspace projections methods are called variational approximation. A number of approaches relating perturbation theory and reduced basis methods were developped, they are often called ``variational perturbation theory'' but are not equivalent to RBM+PT that we studied in this document. We present some of them in this section. Other contributions using variational methods and perturbation theory are present in
~\cite{SisSolSol94,SisSolShe94,ChrJorHat98}.

\subsection{General reduced basis developped in $\lambda$}%
\label{sub:General reduced basis developped in lambda}

In~\cite{SilLeu67}, Silverman and van Leuven consider a self-adjoint operator $H(\lambda)$ depending on the real parameter $\lambda \in \R$. They take a general basis of $m$ functions $(g_a)_{1 \le a \le m}$, forming $\Span (g_a)_{1 \le a \le m} = \cP \cH$ where $\cP$ is an orthogonal projection. They minimize the quadratic form associated to $H(\lambda)$ on this basis, which provides $m$ eigenvectors $(f_a(\lambda))_{1 \le a \le m}$ depending on $\lambda$, and then they study the Taylor expansion of the $f_a(\lambda)$ in $\lambda$. Nevertheless, this is equivalent to RBM+PT in the case where the initial basis $(g_a)_{1 \le a \le m}$ is formed by the derivatives of the exact eigenvector with respect to $\lambda$, but we did not see them following this path in their developments. They called this method the perturbation-variational Rayleigh-Ritz (PVRR) approach, which was used in~\cite{SilLeu70,Silverman77,SilSob78,Silverman83,Silverman84,Silverman86,SilHin88} for instance.

\subsection{Hylleraas variational perturbation theory}%
\label{sub:Hylleraas variational perturbation theory}

Hylleraas variational perturbation theory~\cite{Hylleraas30}, also presented in~\cite[Section 25.$\beta$]{BetSal57} and~\cite{CavDav88,CavDav88b}, enables to compute approximations of the coefficients of the Taylor expansion of $\phi(\lambda)$. It consists in computing $\phi(0)$, and defining a new operator $A$~\cite[(25.13)]{BetSal57} which has $\phi^1$ as exact eigenvector and which only involves $H^0$, $H^1$ and $\phi(0)$. Considering an orthogonal projector $\cP$ producing a reduced space $\cP \cH$, an approximation of $\phi^1$ is obtained as being an eigenvector of $\pa{\cP A \cP}_{\mkern 1mu \vrule height 2ex\mkern2mu \cP \cH \rightarrow \cP \cH}$.

\subsection{Approximating the quantum partition function}%
\label{sub:Approximating the quantum partition function}

To approach the quantum partition function of one particle in a potential $V$, Feynman and Kleinert~\cite{FeyKle86} approximate it by a classical partition function in an effective potential $V_\text{cl}$. This effective potential is supposed to incorporate quantum fluctuations. It is build from a Gaussian smearing of $V$ having a parameter $a$ which is obtain by the optimization of a free energy. Then $V_\text{cl}$ can be obtained from $V$, by a series which smallness parameter is the inverse temperature $\beta$. The method was continued in~\cite{Kleinert93,Kleinert92,KleJan95}, see also the textbook~\cite[Chapter 5]{Kleinert06}. The method was called variational perturbation theory.

\section{Intermediate normalization}%
\label{sec:interm_normalization}

In this section, we show several results about intermediate normalization, which is aimed to be applied to Rayleigh-Schrödinger series eigenvectors in another part of this document, for both degenerate and non-degenerate cases. 

\subsection{Unit normalization}%
\label{sub:Unit normalization}

We consider a Hilbert space $\cH$ with scalar product $\ps{\cdot,\cdot}$ and norm $\nor{\cdot}{}$, and a map $\phi : \R \rightarrow \cH$ depending on one real parameter $\lambda$. We consider that 
\begin{align*}
\nor{\phi(\lambda)}{} = 1
\end{align*}
 for any $\lambda \in \R$, which is called unit normalization. We assume that $\phi$ is analytic at $0$ so we can expand it
\begin{align*}
\phi^n := \f{1}{n!} \pa{\f{\d^n}{\d \lambda^n} \phi(\lambda)}_{\mkern 1mu \vrule height 2ex\mkern2mu \lambda = 0}, \qquad \phi(\lambda) = \sum_{n=0}^{+\infty} \lambda^n \phi^n.
\end{align*}

\subsection{Definition of intermediate normalization}
\label{sub:Intermediate normalization}

Let us define
\begin{align*}
\Phi(\lambda) := \f{\phi(\lambda)}{\ps{\phi^0, \phi(\lambda)}},\qquad  \qquad \Phi^n := \f{1}{n!} \pa{\f{\d^n}{\d \lambda^n} \Phi(\lambda)}_{\mkern 1mu \vrule height 2ex\mkern2mu \lambda = 0}.
\end{align*}
We then define $Z(\lambda) := \f{1}{\ps{\phi^0, \phi(\lambda)}}$. Let us denote by $P_{\phi^0}$ the projector onto $\C \phi^0$, so $P_{\phi^0} \phi(\lambda) = \ps{\phi^0,\phi(\lambda)} \phi^0$. Then we have
\begin{align*}
	\Phi(\lambda) &= Z(\lambda) \pa{P_{\phi^0} \phi(\lambda) + \bpa{1-P_{\phi^0}} \phi(\lambda)} = \phi^0 + \bpa{1-P_{\phi^0}} Z(\lambda)\phi(\lambda),
\end{align*}
Thus $\Phi^0 = \phi^0 = \Phi(0) = \phi(0)$, and for any $n \in \Np$, $\Phi^n \in \bpa{1-P_{\phi^0}} \cH = \{\Phi^0\}^\perp$. We conclude that
\begin{align}\label{eq:inter_normal} 
\Phi^n \perp \Phi^0, \qquad \forall n \ge 1.
\end{align}

The normalization of $\Phi(\lambda)$ is called the intermediate normalization. It is not a unit vector for all $\lambda \neq 0$ in general, but has the convenient property~\eqref{eq:inter_normal}. For instance in the case of families of eigenvectors, it is computable as recalled in Section~\ref{sec:Bounds on the Rayleigh-Schrödinger series in perturbation theory}. For this reason, this is usually the one that is computer first in eigenvalue problems depending on one parameter.

\subsection{From standard normalization to unit normalization}%
\label{sub:Rayleigh-Schrödinger series in standard normalization}

Once $\Phi^n$ is computed, or once one has proved properties on it, one can need to work with $\phi^n$ again. One way of going from intermediate normalization to unit normalization is to fix the phasis gauge of $\phi(\lambda)$ such that 
\begin{align*}
\ps{\phi^0, \phi(\lambda)} \in \R_+.
\end{align*}
 Then $\f{\Phi(\lambda)}{\nor{\Phi(\lambda)}{}} =  \pa{\nor{\Phi(\lambda)}{} \ps{\phi^0, \phi(\lambda)}}^{-1}\phi(\lambda)$ so $\f{\Phi(\lambda)}{\nor{\Phi(\lambda)}{} }$ and $\phi(\lambda)$ have the same phasis, and since they both have unit normalization they are equal, 
\begin{align}\label{eq:interm_to_unit_norm} 
\phi(\lambda) = \f{\Phi(\lambda)}{\nor{\Phi(\lambda)}{}}.
\end{align}

The next result shows how to obtain the series $\phi^n$ from the $\Phi^n$'s.

\begin{lemma}[Obtaining the unit normalization series from the intermediate normalization one]\label{lem:interm_to_unit_norm} 
 We define $Y^0 := X^0 := 1$, $Y^1 := X^1 := 0$ and, recursively, for any $n \ge 2$,
\begin{align}\label{eq:Yn} 
Y^n := \f 12 \sum_{k=1}^{n-1}  \pa{\ps{\Phi^{n-k}, \Phi^k} - Y^{n-k} Y^k}, \qquad  X^n := - \sum_{k=0}^{n-2}  X^k Y^{n-k}.
\end{align}
Then $\phi^0 = \Phi^0$, $\phi^1 = \Phi^1$ and for any $n \ge 2$,
\begin{align}\label{eq:recover_phi} 
\phi^n = \Phi^n + \sum_{k=0}^{n-2} X^{n-k}\Phi^k.
\end{align}
\end{lemma}
We remark that $\phi^2 = \Phi^2 - \f 12 \nor{\phi^1}{}^2 \phi^0$.
\begin{proof}
We define $y(\lambda) := \nor{\Phi(\lambda)}{}$ and consider its Taylor series 
\begin{align*}
y^n := \f{1}{n!} \pa{\f{\d^n}{\d \lambda^n} y(\lambda)}_{\mkern 1mu \vrule height 2ex\mkern2mu \lambda = 0},
\end{align*}
the relation $y(\lambda)^2 = \nor{\Phi(\lambda)}{}^2$ gives, for any $n \in \Nz$,
\begin{align*}
\sum_{k=0}^{n} y^{n-k} y^k = \sum_{k=0}^{n} \ps{\Phi^{n-k}, \Phi^k}
\end{align*}
hence, using $y^0 = \nor{\phi^0}{}^{-2} = 1$ and~\eqref{eq:inter_normal}, we get a recursive way of obtaining the $y^n$'s, which is $y^1 = 0$ and for any $n \ge 2$, via
\begin{align*}
y^n = \f 12 \sum_{k=1}^{n-1}  \pa{\ps{\Phi^{n-k}, \Phi^k} - y^{n-k} y^k},
\end{align*}
so $y^n = Y^n$ for any $n \in \Nz$. We then define $x(\lambda) := 1/y(\lambda)$, and its Taylor series $x^n := \f{1}{n!} \pa{\f{\d^n}{\d \lambda^n} x(\lambda)}_{\mkern 1mu \vrule height 2ex\mkern2mu \lambda = 0}$. The relation $x(\lambda) y(\lambda) = 1$ gives $\sum_{k=0}^{n} x^{k} y^{n-k} = \delta_n$ for any $n \in \Nz$, yielding $x^0 = 1$, $x^1 = 0$ and for any $n \ge 2$,
\begin{align*}
x^n = - \sum_{k=0}^{n-2}  x^k y^{n-k},
\end{align*}
so $x^n = X^n$ for any $n \in \Nz$. Finally, from~\eqref{eq:interm_to_unit_norm} we have $\phi(\lambda) = \Phi(\lambda) x(\lambda)$ hence we deduce~\eqref{eq:recover_phi}.
\end{proof}

\section{Error bounds between eigenvectors,\\density matrices and eigenvalues}%
\label{sec:appendix_error_bounds_DM}

Eigenvectors are controlled by eigen-density matrices, so it is equivalent to obtain bounds using eigenvectors or bounds using density matrices. This is the object of this appendix, and it enables to provide precisions on how to derive the bounds~\eqref{eq:bound_rot} and~\eqref{eq:sum_eigenvals_converge}.

For any set of eigenvalues $\bvp := \pa{\vp_\alpha}_{\alpha =1}^\nu\in \cH^\nu$, we define the norm
\begin{align*}
\nor{\bvp}{}^2 := \sum_{\mu=1}^{\nu} \nor{\vp_\mu}{}^2.
\end{align*}
The following Lemma is well-known, see~\cite[Lemma 3.3]{CanDusMad20} and~\cite[Lemma 2.1]{CanDusMad21}.

\begin{lemma}[Comparing errors between eigenvectors, density matrices and eigenvalues, \cite{CanDusMad21,CanDusMad20}]\label{lem:compare_errors} 
Take two self-adjoint operators $A$ and $H$ acting on a Hilbert space $\cH$, assume that there exists $a \in \R$ such that $0$ is in the resolvent set of $H + a$. Take an orthogonal projection $\cP$, consider $\bphi := \pa{\phi_\alpha}_{\alpha =1}^\nu \in \cH^\nu$ and $\bpsi := \pa{\psi_\alpha}_{\alpha =1}^\nu \in \pa{\cP \cH}^\nu$ such that $\pa{E_\alpha, \phi_\alpha}_{\alpha =1}^\nu$ are eigenmodes of $H$ and $\pa{\cE_\alpha, \psi_\alpha}_{\alpha =1}^\nu$ are eigenmodes of $\cP H \cP$. Define the density matrices $\Gamma := \dm{\bphi}$, $\Lambda := \dm{\bpsi}$, define $U^{\bphi,\bpsi}$ as one of the optimizer(s) of the problem
\begin{align*}
\mymin{U \in \cU_\nu} \nor{\bphi - U \bpsi}{}
\end{align*}
and define $\bpsi^{\bphi} :=  U^{\bphi,\bpsi} \bpsi$. Then we have
\begin{multline}\label{eq:error_dm_vecs_energy} 
\nor{A \bpa{\bphi - \bpsi^\bphi}}{} \le \nor{A \ggm}{} \nor{\ggp A^{-1}}{} \\
\times \pa{1 + \f 14 \nor{\pa{H+a}^{-\f 12}}{}^2\nor{\Gamma - \Lambda}{2}^2 \mymax{1 \le \alpha \le \nu} \ab{E_\alpha + a}}^{\f 12} \nor{A \pa{\Gamma - \Lambda}}{2},
\end{multline}
and
\begin{align}\label{eq:control_energies} 
\ab{\sum_{\alpha=1}^{\nu} \pa{E_\alpha - \cE_\alpha}} \le \pa{\nor{A^{-1} H A^{-1}}{}+ \nor{A^{-1}}{} ^2 \mymax{1 \le \alpha \le \nu} \ab{E_\alpha}} \nor{A \bpa{\bphi - \bpsi^\bphi}}{}^2.
\end{align}
\end{lemma}
We provide a proof in our context for the sake of completeness. It closely follows~\cite[Lemma 3.3]{CanDusMad20} and~\cite[Lemma 2.1]{CanDusMad21}.
\begin{proof}
	First, 
\begin{align}\label{eq:error_dm_vecs} 
	2^{-\f 12} \nor{\Gamma - \Lambda}{2} \le \nor{\bphi - \bpsi^\bphi}{} \le \nor{\Gamma - \Lambda}{2}
\end{align}
is obtained from~\cite[Lemma 2.1]{CanDusMad21} and~\cite[Lemma 4.3]{CanChaMad12}. In~\cite{CanDusMad21} and \cite{CanChaMad12} it is proved for orthogonal matrices, i.e. in the real case, but the proof extends naturally to the complex case. Defining the $\nu \times \nu$ matrix $M$ by $M_{\alpha,\mu} := \ps{\psi^{\bphi}_\alpha, \phi_\mu}$ for any $\alpha, \mu \in \{1,\dots,\nu\}$, by~\cite[Lemma 4.3]{CanChaMad12}, $M$ is hermitian (again, we apply the results to the complex case), so
\begin{multline*}
\ps{\phi_\alpha , \phi_\mu - \psi^{\bphi}_\mu} - \f 12 \ps{\phi_\alpha - \psi^{\bphi}_\alpha , \phi_\mu - \psi^{\bphi}_\mu} \\
= \f 12 \ps{\phi_\alpha , \phi_\mu - \psi^{\bphi}_\mu} + \f 12 \ps{\psi^{\bphi}_\alpha , \phi_\mu - \psi^{\bphi}_\mu} = \f 12 \pa{\ps{\psi^{\bphi}_\alpha , \phi_\mu} - \ps{\phi_\alpha , \psi^{\bphi}_\mu}} = 0
\end{multline*}
and for any $\alpha, \mu \in \{1,\dots,\nu\}$ we have
\begin{align}\label{eq:tcomp} 
\ps{\phi_\alpha , \phi_\mu - \psi^{\bphi}_\mu} = \f 12 \ps{\phi_\alpha - \psi^{\bphi}_\alpha , \phi_\mu - \psi^{\bphi}_\mu}.
\end{align}
Then
\begin{align*}
	\nor{\ab{H + a}^{\f 12} \Lambda}{2}^2 &=\sum_{\alpha=1}^{\nu}  \ps{\psi_\alpha , \ab{H + a} \psi_\alpha} = \sum_{\alpha=1}^{\nu}  \ab{\cE_\alpha + a} \\
	\nor{\ab{H + a}^{\f 12} \Gamma}{2}^2 &= \sum_{\alpha=1}^{\nu}  \ab{E_\alpha + a},
\end{align*}
and
\begin{align*}
	&\hs{\ab{H + a}^{\f 12} \Lambda , \ab{H + a}^{\f 12} \Gamma} = \tr \Lambda \ab{H + a} \Gamma= \tr \sum_{\alpha=1}^{\nu}  \Lambda \ab{H + a} \proj{\phi_\alpha} \\
	&\qquad \qquad = \sum_{\alpha=1}^{\nu} \ab{E_\alpha + a} \ps{\phi_\alpha, \Lambda \phi_\alpha} \underset{\substack{\Lambda = \Lambda^2}}{=} \; \sum_{\alpha=1}^{\nu} \ab{E_\alpha + a} \nor{\Lambda \phi_\alpha}{}^2 \\
	&\qquad \qquad = \sum_{\alpha=1}^{\nu} \ab{E_\alpha + a} \pa{1 - \nor{\Lambda^\perp \phi_\alpha}{}^2}.
\end{align*}
We can hence compute
\begin{align*}
	\nor{ \ggp \pa{\Gamma - \Lambda}}{2}^2 &= \nor{\ab{H + a}^{\f 12} \Gamma}{2}^2+ \nor{\ab{H + a}^{\f 12} \Lambda}{2}^2 \\
&\qquad \qquad \qquad \qquad - 2 \re \hs{\ab{H + a}^{\f 12} \Lambda , \ab{H + a}^{\f 12} \Gamma}  \\
&=\sum_{\alpha=1}^{\nu} \ab{\cE_\alpha + a} - \ab{E_\alpha + a} + 2\ab{E_\alpha +a} \nor{\Lambda^\perp \phi_\alpha}{}^2.
\end{align*}
Then,
\begin{align*}
&\nor{\ggp \pa{\bphi - \bpsi^{\bphi}}}{}^2 =\sum_{\alpha=1}^{\nu} \nor{ \ggp \pa{\phi_\alpha - \psi^{\bphi}_\alpha}}{}^2 \\
&\qquad = \sum_{\alpha=1}^{\nu} \nor{ \ggp \phi_\alpha }{}^2 + \nor{ \ggp \psi^{\bphi}_\alpha}{}^2  - 2\re \ps{ \ab{H + a} \phi_\alpha ,  \psi^{\bphi}_\alpha} \\
&\qquad = \sum_{\alpha=1}^{\nu}  \ab{\cE_\alpha + a} + \ab{E_\alpha + a} - 2 \ab{E_\alpha +a} \re \ps{  \phi_\alpha ,  \psi^{\bphi}_\alpha} \\
&\qquad = \sum_{\alpha=1}^{\nu} \ab{\cE_\alpha + a} - \ab{E_\alpha + a} + \ab{E_\alpha +a} \nor{\phi_\alpha - \psi^{\bphi}_\alpha}{}^2.
\end{align*}
where we used $\ps{\phi_\alpha , \psi^{\bphi}_\alpha} = 1 - \f 12  \nor{\phi_\alpha - \psi^{\bphi}_\alpha}{}^2$ in the last equality, which comes from~\eqref{eq:tcomp}. We define $\lambda_{\textup{max}} := \mymax{1 \le \alpha \le \nu} \ab{E_\alpha + a}$ and following~\cite[Appendix A]{CanDusMad20}, we have
\begin{align*}
&\nor{\ggp \pa{\bphi - \bpsi^{\bphi}}}{}^2 - \nor{ \ggp \pa{\Gamma - \Lambda}}{2}^2 \\
&\qquad  = \sum_{\alpha=1}^{\nu} \ab{E_\alpha +a}\pa{\nor{\phi_\alpha - \psi^{\bphi}_\alpha}{}^2 - 2 \nor{\Lambda^\perp \pa{ \phi_\alpha - \psi^{\bphi}_\alpha}}{}^2} \\
& \qquad \le \sum_{\alpha=1}^{\nu} \ab{E_\alpha +a}\pa{\nor{\phi_\alpha - \psi^{\bphi}_\alpha}{}^2 - \nor{\Lambda^\perp \pa{ \phi_\alpha - \psi^{\bphi}_\alpha}}{}^2} \\
& \qquad = \sum_{\alpha=1}^{\nu} \ab{E_\alpha +a}\nor{\Lambda \pa{ \phi_\alpha - \psi^{\bphi}_\alpha}}{}^2 = \sum_{\substack{1 \le \alpha, \mu \le \nu}} \ab{E_\alpha + a} \ab{\ps{\phi_\alpha - \psi^{\bphi}_\alpha , \phi_\mu}}^2 \\
&\qquad  \underset{\substack{\eqref{eq:tcomp}}}{\le} \; \f 14\lambda_{\textup{max}} \sum_{\substack{1 \le \alpha, \mu \le \nu}} \ab{\ps{\phi_\alpha - \psi^{\bphi}_\alpha , \phi_\mu - \psi^{\bphi}_\mu}}^2 \le \f 14 \lambda_{\textup{max}}\nor{\bphi - \bpsi^{\bphi}}{2}^4   \\
&\qquad \underset{\substack{\eqref{eq:error_dm_vecs}}}{\le} \; \f 14 \lambda_{\textup{max}}\nor{\Gamma - \Lambda}{2}^4  \\
&\qquad \le \f 14 \lambda_{\textup{max}} \nor{\pa{H+a}^{-\f 12}}{}^2  \nor{\Gamma - \Lambda}{2}^2 \nor{\ggp \pa{\Gamma - \Lambda}}{2}^2.
\end{align*}
Then
\begin{multline}\label{eq:this_ineq} 
\nor{\ggp \pa{\bphi - \bpsi^{\bphi}}}{} \le \pa{1 + \f 14 \lambda_{\textup{max}} \nor{\pa{H+a}^{-\f 12}}{}^2  \nor{\Gamma - \Lambda}{2}^2}^{\f 12} \\
\times\nor{\ggp \pa{\Gamma - \Lambda}}{2}.
\end{multline}
Next,
\begin{align*}
&\nor{A \bpa{\bphi - \bpsi^{\bphi}}}{} \le \nor{A \ggm}{}  \nor{ \ggp \pa{\bphi - \bpsi^{\bphi}}}{} \\
& \qquad \le \nor{A \ggm}{} \pa{1 + \f 14 \lambda_{\textup{max}} \nor{\pa{H+a}^{-\f 12}}{}^2  \nor{\Gamma - \Lambda}{2}^2}^{\f 12} \\
&\qquad \qquad \qquad \qquad \qquad \qquad \times\nor{\ggp \pa{\Gamma - \Lambda}}{2},
\end{align*}
and we deduce~\eqref{eq:error_dm_vecs_energy} by using~\eqref{eq:this_ineq}.

Let us now show~\eqref{eq:control_energies}. For any $U \in \cU_\nu$, we have
\begin{align*}
	&\sum_{\alpha=1}^{\nu} \ps{U \psi_\alpha , H U \psi_\alpha} = \sum_{\substack{1 \le \alpha, \mu, \beta \le \nu}} \overline{U}_{\alpha\mu} U_{\alpha\beta} \ps{\psi_\mu, H \psi_\beta} \\
	&\qquad = \sum_{\substack{1 \le \alpha, \mu, \beta \le \nu}} \overline{U}_{\alpha\mu} U_{\alpha\beta} \cE_\beta \ps{\psi_\mu, \psi_\beta} = \sum_{\substack{1 \le \alpha, \mu, \beta \le \nu}} \overline{U}_{\alpha\mu} U_{\alpha\beta} \cE_\beta \delta_{\mu-\beta} \\
	&\qquad = \sum_{\substack{1 \le \alpha, \mu\le \nu}} \overline{U}_{\alpha\mu} U_{\alpha\mu} \cE_\mu = \sum_{\mu=1}^{\nu} \cE_\mu \sum_{\alpha=1}^{\nu} U^*_{\mu\alpha} U_{\alpha\mu}  = \sum_{\mu=1}^{\nu} \cE_\mu \pa{U^* U}_{\mu\mu} = \sum_{\mu=1}^{\nu} \cE_\mu.
\end{align*}
Hence similarly as in~\eqref{eq:diff_errs},
\begin{align*}
	&\sum_{\alpha=1}^{\nu} \ps{\phi_\alpha - \psi^\bphi_\alpha , \pa{E_\alpha - H} \pa{ \phi_\alpha - \psi^\bphi_\alpha } } = \sum_{\alpha=1}^{\nu} \ps{ \psi^\bphi_\alpha , \pa{E_\alpha - H}  \psi^\bphi_\alpha } \\
	&=\sum_{\alpha=1}^{\nu} E_\alpha - \ps{ U^{\bphi,\bpsi} \psi_\alpha, H U^{\bphi,\bpsi} \psi_\alpha } =\sum_{\alpha=1}^{\nu} \pa{ E_\alpha -  \cE_\alpha}.
\end{align*}
Thus
\begin{align*}
&\ab{\sum_{\alpha=1}^{\nu} \pa{E_\alpha - \cE_\alpha}} \le \sum_{\alpha=1}^{\nu} \ab{\ps{A \bpa{\phi_\alpha - \psi^\bphi_\alpha} , A^{-1}\bpa{E_\alpha - H}A^{-1} A \bpa{ \phi_\alpha - \psi^\bphi_\alpha } }} \\
&\le \pa{ \mymax{1 \le \alpha \le \nu} \nor{A^{-1}\pa{E_\alpha - H}A^{-1}}{} }\sum_{\alpha=1}^{\nu} \nor{\phi_\alpha - \psi^\bphi_\alpha}{e}^2 \\
&\le \pa{c_H + c_A^2 \mymax{1 \le \alpha \le \nu} \ab{E_\alpha}} \nor{A \bpa{\bphi - \bpsi^\bphi}}{}^2.
\end{align*}
\end{proof}

\section{Alternative bound for $\Gamma - \Lambda$}%
\label{sec:Alternative bound}

In this section, we provide another way of bounding $\Gamma - \Lambda$, which represent an alternative to~\eqref{eq:explicit_diff}. It uses another method.

For any $z \in \{E_\mu\}_{\mu=1}^\nu \cup \pa{\C \backslash \sigma( H )}$ we define
\begin{align}\label{def:pseudoinv_z_H}
\pa{z - H}_\perp^{-1} := 
\left\{
\begin{array}{ll}
\pa{\pa{z -  H }_{\mkern 1mu \vrule height 2ex\mkern2mu  \Gamma^\perp \cH \rightarrow \Gamma^\perp \cH}}^{-1} & \mbox{on }  \Gamma^\perp \cH, \\
0 & \mbox{on } \Gamma \cH
\end{array}
\right.
\end{align}
extended by linearity on $\cH$. 
\begin{proposition}[Alternative bound for $\Gamma - \Lambda$]\label{prop:another_bound_PerrP} 
Under the same assumptions as in Theorem~\ref{thm:main_deg_thm}, and if further the following gap assumptions
\begin{align}\label{eq:assum_sec_bound} 
	\dist \pa{\{\cE_\mu\}_{\mu=1}^\nu , \sigma_{\textup{d}}\bpa{H_{\mkern 1mu \vrule height 2ex\mkern2mu \Gamma^\perp \cH \rightarrow \Gamma^\perp \cH}}} &> 0, \nonumber\\
	\dist\pa{\{E_\mu\}_{\mu=1}^\nu , \sigma_{\textup{d}}\bpa{\pa{\cP H \cP}_{\mkern 1mu \vrule height 2ex\mkern2mu \Lambda^\perp \cP \cH \rightarrow \Lambda^\perp \cP \cH}}} &> 0,
 \end{align}
hold, then 
\begin{multline}\label{eq:total_PerrP_2} 
\nor{\Gamma - \Lambda}{2,\delta} \le c_A^\delta \nor{ \cP^\perp \Gamma }{2,\delta}^2 + \nu c_\cP^\delta \nor{ \cP^\perp H \cP \Lambda }{} \mymax{1 \le \mu \le \nu} \nor{ A^\delta \pa{\cE_\mu - H}_\perp^{-1} }{}  \\
+ \pa{c_A c_\cP\nor{A \Gamma}{}}^\delta \nor{\cP^\perp \Gamma}{2,\delta} \pa{2 + \nu \mymax{1 \le \mu \le \nu} \nor{\pa{E_\mu - \cP H \cP}_\perp^{-1} \cP H \cP^\perp}{} } .
\end{multline}
\end{proposition}

The proof of this result is provided in Section~\ref{sec:Proof of Theorem mauin}.

	The quantity $\cP^\perp H \Lambda$ can be interpreted as an \textup{a posteriori} one. When $\Gamma$ is close to $\Lambda$ it is small because, since $[H,\Gamma]=0$ and $\cP^\perp \Lambda = 0$, then $\cP^\perp H \Lambda = \cP^\perp [H, \Lambda - \Gamma]$. The bound~\eqref{eq:bound_Omega} involves $\nor{ \Gamma - \Lambda }{2,\delta}^2$ while~\eqref{eq:total_PerrP_2} does not. If one rather needs an \textup{a posteriori} quantification, ~\eqref{eq:total_PerrP_2} might be better.

\subsection{Proof of Proposition~\ref{prop:another_bound_PerrP}}%
\label{sub:Proof of}

In this section we present another way of treating $\Gamma - \Lambda$.

For any $z \in \C \backslash \sigma \bpa{(\cP H \cP)_{\mkern 1mu \vrule height 2ex\mkern2mu \cP \cH \rightarrow \cP \cH} }$, we define the partial inverse
\begin{align}\label{def:pseudoinv_z_PHP_restr}
\pa{z - \cP H \cP}^{-1}_{\cP \cH} := 
\left\{
\begin{array}{ll}
\pa{\pa{z -  \cP H \cP}_{\mkern 1mu \vrule height 2ex\mkern2mu \cP \cH \rightarrow \cP \cH}}^{-1} & \mbox{on }  \cP \cH, \\
0 & \mbox{on } \cP^\perp \cH,
\end{array}
\right.
\end{align}
extended by linearity on $\cH$. By the definition~\eqref{def:pseudoinv_z_PHP_restr}, 
\begin{align*}
\pa{z - \cP H \cP} \pa{z - \cP H \cP}^{-1}_{\cP \cH} \vp =
\left\{
\begin{array}{ll}
\vp & \mbox{if } \vp \in \cP \cH,  \\
0 & \mbox{if } \vp \in \cP^\perp \cH,
\end{array}
\right.
\end{align*}
hence
\begin{align}\label{eq:res_PHP} 
\pa{z - \cP H \cP} \pa{z - \cP H \cP}^{-1}_{\cP \cH} = \cP.
\end{align}
Then
\begin{align*}
&\cP (z-H)^{-1} \cP - (z-\cP H \cP)^{-1}_{\cP \cH} = \cP (z-H)^{-1} \pa{\cP - (z-H) (z- \cP H \cP)^{-1}_{\cP \cH}} \\
&\qquad\qquad = \cP (z-H)^{-1} \pa{\cP - (z- \cP H \cP + \cP H \cP - H) (z- \cP H \cP)^{-1}_{\cP \cH} } \\
&\qquad\qquad \underset{\substack{\eqref{eq:res_PHP}}}{=} \; \cP (z-H)^{-1} (H - \cP H \cP ) (z- \cP H \cP)^{-1}_{\cP \cH} \\
&\qquad\qquad = \cP (z-H)^{-1} \cP^\perp H (z- \cP H \cP)^{-1}_{\cP \cH},
\end{align*}
where we used that $(z- \cP H \cP)^{-1}_{\cP \cH} = \cP (z- \cP H \cP)^{-1}_{\cP \cH}$ in the last step.
We now use that 
\begin{align*}
\Gamma (z-H)^{-1} = \sum_{\mu=1}^{\nu} P_{\phi_\mu}  (z-H)^{-1}= \sum_{\mu=1}^{\nu} P_{\phi_\mu}  (z-E_\mu)^{-1}
\end{align*}
to deduce $\Gamma (z-H)^{-1} \Gamma^\perp = \Gamma^\perp (z-H)^{-1} \Gamma = 0$ and $\Gamma (z-H)^{-1} = \Gamma (z-H)^{-1} \Gamma$, so we can write
\begin{align*}
	&(z-H)^{-1} \\
	&\qquad = \Gamma (z-H)^{-1} \Gamma + \Gamma^\perp (z-H)^{-1} \Gamma + \Gamma (z-H)^{-1} \Gamma^\perp + \Gamma^\perp (z-H)^{-1} \Gamma^\perp \\
	&\qquad = (z-H)^{-1}_\perp + \sum_{\mu=1}^{\nu} P_{\phi_\mu}  (z-E_\mu)^{-1}.
\end{align*}
Similarly,
\begin{align*}
(z- \cP H \cP)^{-1}_{\cP \cH} &= \Lambda (z- \cP H \cP)^{-1}_{\cP \cH} +  \Lambda^\perp (z- \cP H \cP)^{-1}_{\cP \cH} \Lambda^\perp \\
&= (z- \cP H \cP)^{-1}_\perp  +  \sum_{\mu=1}^{\nu} P_{\psi_\mu}  (z-\cE_\mu)^{-1}.
\end{align*}
The operators $(z-H)^{-1}_\perp$ and $(z- \cP H \cP)^{-1}_\perp$ are holomorphic in the interior of $\cC$ so they will ``participate passively'' to the Cauchy integral. 
We have
\begin{align}\label{eq:double_sing_zero} 
\f{1}{2\pi i} \oint_\cC (z-E_\mu)^{-1} (z-\cE_\alpha)^{-1} \d z = \delta_{E_\mu \neq \cE_\alpha} \pa{(\cE_\alpha-E_\mu)^{-1} + (E_\mu-\cE_\alpha)^{-1}} = 0.
\end{align}
We are ready to compute the Cauchy integral
\begin{align*}
	&\cP \pa{\Gamma - \Lambda} \cP = \f{1}{2\pi i} \oint_\cC \pa{\cP (z-H)^{-1} \cP - (z-\cP H \cP)^{-1}_{\cP \cH}} \d z \\
&\qquad  = \f{1}{2\pi i} \oint_\cC \cP (z-H)^{-1} \cP^\perp H (z- \cP H \cP)^{-1}_{\cP \cH}\d z \\
&\qquad = \sum_{\mu=1}^{\nu} \pa{\cP P_{\phi_{\mu}} \cP^\perp H  \pa{E_\mu - \cP H \cP}_\perp^{-1}  + \cP \pa{\cE_\mu - H}_\perp^{-1} \cP^\perp H P_{\psi_\mu}}.
\end{align*}
As for inequalities, we have
\begin{align*}
&\nor{A^\delta \cP P_{\phi_{\mu}} \cP^\perp H  \pa{E_\mu - \cP H \cP}_\perp^{-1} }{2} \\
&\qquad =\nor{A^\delta  \cP A^{-\delta} A^\delta \Gamma P_{\phi_{\mu}} \Gamma \cP^\perp \cP^\perp H  \pa{E_\mu - \cP H \cP}_\perp^{-1} }{2} \\
&\qquad \le\nor{A^\delta  \cP A^{-\delta}}{} \nor{ A^\delta \Gamma}{} \nor{P_{\phi_\mu}}{2} \nor{ \Gamma \cP^\perp}{} \nor{\cP^\perp H  \pa{E_\mu - \cP H \cP}_\perp^{-1} }{} \\
&\qquad \le \pa{c_A c_\cP\nor{A \Gamma}{}}^\delta \nor{\pa{E_\mu - \cP H \cP}_\perp^{-1} H \cP^\perp}{}  \nor{\cP^\perp \Gamma}{2,\delta},
\end{align*}
and
\begin{align*}
& \nor{A^\delta \cP \pa{\cE_\mu - H}_\perp^{-1} \cP^\perp H P_{\psi_\mu} }{2} = \nor{A^\delta \cP A^{-\delta} A^\delta \pa{\cE_\mu - H}_\perp^{-1} \cP^\perp H \Lambda P_{\psi_\mu} }{2} \\
&\qquad \le \nor{A^\delta \cP A^{-\delta}}{} \nor{ A^\delta \pa{\cE_\mu - H}_\perp^{-1} }{}  \nor{ \cP^\perp H \Lambda }{} \nor{P_{\psi_\mu}}{2} \\
&\qquad  \le c_\cP^\delta \nor{ A^\delta \pa{\cE_\mu - H}_\perp^{-1} }{}  \nor{ \cP^\perp H \Lambda }{},
\end{align*}
and also using the inequalities of Section~\ref{sub:Treating first terms}, we can deduce~\eqref{eq:total_PerrP_2} of Proposition~\ref{prop:another_bound_PerrP}.

\bibliographystyle{siam}
\bibliography{rbm_pt}
\end{document}